\documentclass[12pt]{article}
\usepackage[latin9]{inputenc}
\usepackage[letterpaper]{geometry}
\usepackage{amsmath}
\usepackage{amsthm}
\usepackage{amssymb}
\usepackage{setspace}
\usepackage{comment}
\usepackage{graphicx}
\usepackage{rotating}
\usepackage[authoryear]{natbib}
\usepackage[unicode=true,
 bookmarks=false,
 breaklinks=false,pdfborder={0 0 1},backref=false,colorlinks=false]
 {hyperref}
\hypersetup{
 colorlinks,citecolor=DarkBlue,bookmarksnumbered=true,urlcolor=Indigo,linkcolor=blue}

\makeatletter
\theoremstyle{plain}
\newtheorem{assumption}{\protect\assumptionname}
\theoremstyle{plain}
\newtheorem{thm}{\protect\theoremname}
\theoremstyle{definition}
\newtheorem{condition}{\protect\conditionname}
\theoremstyle{plain}
\newtheorem{lem}{\protect\lemmaname}
\theoremstyle{definition}
 \newtheorem{example}{\protect\examplename}
\theoremstyle{plain}
\newtheorem{cor}{\protect\corollaryname}
\newtheorem{remark}{\protect\remarkname}
\theoremstyle{plain}
\usepackage[svgnames]{xcolor}

\usepackage{graphicx}
\usepackage{caption}
\usepackage{bm}
\usepackage{bbm}
\usepackage{tikz}
\usepackage{capt-of}
\usepackage{extarrows}
\usepackage{enumerate}
\usepackage{booktabs}
\usepackage{multirow}

\setcitestyle{authoryear,open={(},close={)}}

\def\var{\mathop{\rm Var}}

\makeatother

\providecommand{\assumptionname}{Assumption}
\providecommand{\conditionname}{Condition}
\providecommand{\corollaryname}{Corollary}
\providecommand{\examplename}{Example}
\providecommand{\lemmaname}{Lemma}
\providecommand{\theoremname}{Theorem}
\providecommand{\remarkname}{Remark}

\begin{document}

\sloppy

\global\long\def\a{\alpha}%
 
\global\long\def\b{\beta}%
 
\global\long\def\g{\gamma}%
 
\global\long\def\d{\delta}%
 
\global\long\def\e{\epsilon}%
 
\global\long\def\l{\lambda}%
 
\global\long\def\t{\theta}%
 
\global\long\def\o{\omega}%
 
\global\long\def\s{\sigma}%

\global\long\def\G{\Gamma}%
 
\global\long\def\D{\Delta}%
 
\global\long\def\L{\Lambda}%
 
\global\long\def\T{\Theta}%
 
\global\long\def\O{\Omega}%
 
\global\long\def\R{\mathbb{R}}%
 
\global\long\def\N{\mathbb{N}}%
 
\global\long\def\Q{\mathbb{Q}}%
 
\global\long\def\I{\mathbb{I}}%
 
\global\long\def\P{\mathbb{P}}%
 
\global\long\def\E{\mathbb{E}}%
\global\long\def\B{\mathbb{\mathbb{B}}}%
\global\long\def\S{\mathbb{\mathbb{S}}}%
\global\long\def\V{\mathbb{\mathbb{V}}\text{ar}}%
 
\global\long\def\GG{\mathbb{G}}%
\global\long\def\TT{\mathbb{T}}%

\global\long\def\X{{\bf X}}%
\global\long\def\cX{\mathscr{X}}%
 
\global\long\def\cY{\mathscr{Y}}%
 
\global\long\def\cA{\mathscr{A}}%
 
\global\long\def\cB{\mathscr{B}}%
\global\long\def\cF{\mathscr{F}}%
 
\global\long\def\cM{\mathscr{M}}%
\global\long\def\cN{\mathcal{N}}%
\global\long\def\cG{\mathcal{G}}%
\global\long\def\cC{\mathcal{C}}%
\global\long\def\sp{\,}%

\global\long\def\es{\emptyset}%
 
\global\long\def\mc#1{\mathscr{#1}}%
 
\global\long\def\ind{\mathbf{\mathbbm1}}%
\global\long\def\indep{\perp}%

\global\long\def\any{\forall}%
 
\global\long\def\ex{\exists}%
 
\global\long\def\p{\partial}%
 
\global\long\def\cd{\cdot}%
 
\global\long\def\Dif{\nabla}%
 
\global\long\def\imp{\Rightarrow}%
 
\global\long\def\iff{\Leftrightarrow}%

\global\long\def\up{\uparrow}%
 
\global\long\def\down{\downarrow}%
 
\global\long\def\arrow{\rightarrow}%
 
\global\long\def\rlarrow{\leftrightarrow}%
 
\global\long\def\lrarrow{\leftrightarrow}%

\global\long\def\abs#1{\left|#1\right|}%
 
\global\long\def\norm#1{\left\Vert #1\right\Vert }%
 
\global\long\def\rest#1{\left.#1\right|}%

\global\long\def\bracket#1#2{\left\langle #1\middle\vert#2\right\rangle }%
 
\global\long\def\sandvich#1#2#3{\left\langle #1\middle\vert#2\middle\vert#3\right\rangle }%
 
\global\long\def\turd#1{\frac{#1}{3}}%
 
\global\long\def\ellipsis{\textellipsis}%
 
\global\long\def\sand#1{\left\lceil #1\right\vert }%
 
\global\long\def\wich#1{\left\vert #1\right\rfloor }%
 
\global\long\def\sandwich#1#2#3{\left\lceil #1\middle\vert#2\middle\vert#3\right\rfloor }%

\global\long\def\abs#1{\left|#1\right|}%
 
\global\long\def\norm#1{\left\Vert #1\right\Vert }%
 
\global\long\def\rest#1{\left.#1\right|}%
 
\global\long\def\inprod#1{\left\langle #1\right\rangle }%
 
\global\long\def\ol#1{\overline{#1}}%
 
\global\long\def\ul#1{\underline{#1}}%
 
\global\long\def\td#1{\tilde{#1}}%
\global\long\def\bs#1{\boldsymbol{#1}}%

\global\long\def\upto{\nearrow}%
 
\global\long\def\downto{\searrow}%
 
\global\long\def\pto{\overset{p}{\longrightarrow}}%
 
\global\long\def\dto{\overset{d}{\longrightarrow}}%
 
\global\long\def\asto{\overset{a.s.}{\longrightarrow}}%

\setlength{\abovedisplayskip}{6pt} \setlength{\belowdisplayskip}{6pt}
\title{IV Regressions without Exclusion Restrictions\thanks{We thank Jason Blevins, Xiaohong Chen, Xu Cheng, Stephen Cosslett, Francis Diebold, Keisuke Hirano, Robert de Jong, Lixiong Li, Xiao Lin, Ce Liu, Joris Pinkse, Neslihan Sakarya, Frank Schorfheide, Petra Todd, Bruce Weinberg, as well as participants at various conferences and seminars for helpful comments and suggestions.}}

\author{Wayne Yuan Gao\thanks{Department of Economics, University of Pennsylvania, 133 S 36th St, Philadelphia, PA 19104, USA. Email: waynegao@upenn.edu}$\ \ $and
Rui Wang\thanks{Department of Economics, The Ohio State University, 1945 N High St, Columbus, OH 43210, USA. Email: wang.16498@osu.edu} }
\date{July 30, 2023}

\maketitle
\begin{abstract}

We study identification and estimation of endogenous linear and nonlinear regression models without excluded instrumental variables, based on the standard mean independence condition and a nonlinear relevance condition. Based on the identification results, we propose two semiparametric estimators
as well as a discretization-based estimator that does not require
any nonparametric regressions. We establish their asymptotic
normality and demonstrate
via simulations their robust finite-sample performances with respect to exclusion restrictions violations and endogeneity.  Our approach is applied to study the returns to education, and to test the direct effects of college proximity indicators as well as family background variables on the outcome.


\noindent \textbf{~}\\
\textbf{Keywords}: linear regression, quantile regression, endogeneity, instrumental variable,
exclusion restriction, semiparametric two-stage estimation
\end{abstract}
\newpage{}

\section{\label{sec:Intro}Introduction}

The method of instrumental variables (IV) has been a central approach to identify and estimate linear regression models with endogeneity. The conventional IV regression exploits excluded instrumental variables that have no direct effects on the outcome variable. However, finding valid instruments that satisfy the exclusion restriction can be challenging in many applications.

In this paper, we show that even in the absence of excluded instruments, the endogenous linear regression model can still be identified by leveraging the nonlinear relevance between the included exogenous regressor and the endogenous variable. In contrast to the traditional IV regression that uses a linear first-stage projection, our approach applies a mean projection of the endogenous variable on only the included exogenous regressor in the first stage. 
More generally, this approach can also be applied to nonlinear regressions with known functional form, which may naturally arise from structure models. In such cases, we provide local identification of the model parameter under a full-rank condition.

To illustrate, let us consider the following simple linear regression model: 
\begin{equation}
Y_{i}=\a_{0}+\b_{0}Z_{i}+\g_{0}X_{i}+\e_{i}\label{eq:ModelScalar}
\end{equation}
with a scalar endogenous variable $X_{i}$ and a scalar exogenous
variable $Z_{i}$ such that 
\[
\E\left[\rest{\e_{i}}Z_{i}\right]=0.
\]

To identify the coefficient $\t_0:=\left(\a_{0},\b_{0},\g_{0}\right)$, the standard IV regression, or the two-stage least square (2SLS) regression, relies on the availability of an additional variable $Z_{i, exc}$ and applies a linear first-stage projection as follows:
\[
X_{i}=\l_{0}+\l_{1}Z_{i}+\l_{2}Z_{i, exc}+U_{i},
\]
where the instrumental variable $Z_{i, exc}$ is required to be  \emph{exogenous }with respect to $\e_{i}$\emph{,}
\emph{relevant }for $X_{i}$, and \emph{excluded} from the regression
model \eqref{eq:ModelScalar}.


By contrast, this paper investigates the identification and estimation of $\t_0$ without excluded instrumental variables. 
To see the idea, take
conditional expectations of both sides of \eqref{eq:ModelScalar}
given $Z_{i}$, we have
\begin{equation}
\label{eq:no_error}
\E\left[\rest{Y_{i}-\a_{0}-\b_{0}Z_{i}-\g_{0} X_i }Z_{i}=z\right]=0,
\end{equation}
because the term $\E\left[\rest{\e_i}Z_i \right]=0$ by the exogeneity of $Z_i$.

Instead of linearly projecting endogenous variable $X_i$ on $Z_i$, we adopt the mean projection of $X_i$ on $Z_i$, i.e., 
$\pi_{0}\left(Z_{i}\right):=\E\left[\rest{X_{i}}Z_{i}\right]$. 
 Then the moment restriction in \eqref{eq:no_error} can be written as
\begin{equation}
\label{eq:mom_res}
\E\left[\rest{Y_{i}-\a_{0}-\b_{0}Z_{i}-\g_{0}\pi_{0}\left(Z_{i}\right)  }Z_{i}=z\right]=0.
\end{equation}
Our key idea is based on the simple observation that condition \eqref{eq:mom_res} can be viewed as a linear regression of $Y_i$
on 1, $Z_{i}$, and $\pi_{0}\left(Z_{i}\right)$ with \emph{no} endogeneity
issue since $Z_i$ satisfies the exogeneity condition. It is thus clear that the parameter $\t_0$ is identified in \eqref{eq:mom_res} as long as $\left(1,Z_{i},\pi_{0}\left(Z_{i}\right)\right)$
are not (perfectly) multicollinear, which is equivalent to the following requirement:
\[
\pi_{0}\left(z\right)\text{ is nonlinear in }z,
\]
i.e, there exist no constants $a,b\in\R$ such that $\pi_{0}\left(z\right)=a+bz$
for any $z$ in the support of $Z_i$. This nonlinearity condition is testable, as $\pi_0$ only involves observed variables.

One natural setting this nonlinear relationship arises is when the endogenous regressor $X_{i}\in \{0, 1\}$ is a binary variable. Then the propensity score function, $\pi_0(z)$, is naturally nonlinear in $z$. For example, consider the following binary choice model for $X_i$,
\[ X_{i}=\ind\left\{ \eta_0+\eta_1 Z_i \geq u_{i}\right\},  \]
where $u_i \indep Z_i$ and $ u_i$ follows some distribution $F_u$ (e.g., a normal distribution for the Probit model). So
\[
\pi_{0}\left(z\right)=\E\left[\rest{X_{i} }Z_{i}=z\right]=F_u(\eta_0+\eta_1z).
\]
Then, the function $\pi_{0}\left(z\right)$
is nonlinear as long as $Z_i$ takes at least three values and $F_u$ is not a uniform distribution. Note that in our identification approach, the distribution $F_u$ does not need to be known, which is distinct from the Heckman correction approach.\footnote{Section \ref{sec:disc} provides a more detailed discussion of the differences between our approach and the Heckman correction approach.}


As a more concrete example, suppose we are interested in the effect of a college degree $X_i$ on (log) wage $Y_i$. Then the included instrument $Z_i$ could be years of parents' education, which takes more than three values. Alternatively, $Z_i$ can include two binary variables $Z_{i1}, Z_{i2}$, with $Z_{i1}$ representing gender and $Z_{i2}$ representing whether one's mother has a college degree.\footnote{The nonlinearity of $\pi_0$ can be satisfied under a mild condition on their coefficients: $\Phi(\eta_0+\eta_1+\eta_2)+\Phi(\eta_0)-\Phi(\eta_0+\eta_1)-\Phi(\eta_0+\eta_2)\neq 0$, where $\eta_1, \eta_2$ are the coefficients of $Z_{i1}, Z_{i2}$, respectively.} 


More generally, the analysis extends to endogenous nonlinear and quantile regression. By adopting a mean projection of the nonlinear function onto the included exogenous regressor, we show that local identification is achieved under a full-rank condition. In particular, for quantile regression, we show that the full-rank condition is equivalent to a different  nonlinear relevance condition.

~

The identification result suggests a natural semiparametric two-step
 estimator. We describe the estimator for the linear regression model, and the estimator for the nonlinear and quantile regression is provided in Section \ref{sec:nonlinear}.
 Specifically, given
 $\hat{\pi}$ obtained via first-stage nonparametric
regression of $X_{i}$ on $Z_{i}$, we
construct our first estimator $\hat{\t}$ by 
\[
\text{regressing \ensuremath{Y_{i}\text{ on }1,\ }}Z_{i}\ \text{and}\ \hat{\pi}_{0}\left(Z_{i}\right)\ \ \text{via OLS}.
\]

We also propose an estimator that uses the mean projection $h_0\left( Z_i\right):=\E\left[\rest{Y_{i}}Z_{i}\right]$ of $Y_i$ on $Z_i$ as the dependent variable. Let $\hat{h}$ denote the estimator of nonparametric regression of $Y_i$ on $Z_i$, the second estimator
$\hat{\t}^{*}$ is constructed by 
\[
\text{regressing \ensuremath{\hat{h}\left(Z_{i}\right)\text{ on }1,\ }}Z_{i}\ \text{and}\ \hat{\pi}_{0}\left(Z_{i}\right)\ \text{via OLS}.
\]

The only difference between $\hat{\t}$ and $\hat{\t}^{*}$ lies in
the dependent variable used in the second step: $\hat{\t}$ uses the raw observed variable $Y_{i}$, while $\hat{\t}^{*}$ uses
the fitted value $\hat{h}\left(Z_{i}\right)$ obtained through nonparametric regression of $Y_{i}$ on $Z_{i}$. We propose the second estimator $\hat{\t}^{*}$,  as it can perform slightly better than $\hat{\t}$ under some specifications in simulations.
%

We further propose a third estimator, $\hat{\t}_{disc}$, which does
not require any nonparametric estimation, based on a discretization of the support
of $Z_{i}$ into $K$ (finite and fixed) partitions. Under
this discretization, the first-stage estimation 
simplifies to sample averages in each partition. Furthermore, the estimator can be computed as a standard 2SLS estimator with partition dummies as IVs. While the
discretization results in some loss of information and asymptotic efficiency,
there is no ``discretization bias'' in our setting and the number of partition cells $K$ is not required to grow large with the sample size. 

We establish the $\sqrt{n}$-consistency of our three proposed estimators $\hat{\t}$, $\hat{\t}^{*}$, and $\hat{\t}_{disc}$ for $\t_{0}$, along with their asymptotic normality. We show that $\hat{\t}$ and $\hat{\t}^{*}$  share exactly the same asymptotic variance, while that of $\hat{\t}_{disc}$ is in general different and, when error are homoskedastic, larger under the partial order of positive semi-definiteness. 

Monte Carlo simulations support our theoretical results, and demonstrate
the good finite-sample performance of the three estimators $\hat{\t}$, $\hat{\t}^{*}$, and $\hat{\t}_{disc}$ with the presence of violation of the exclusion restriction and endogeneity.
For comparison, we also implement the standard 2SLS estimator which treats the included regressor as excluded instrument, as well as the OLS estimator which does not account for endogeneity. 
The root mean squared error (RMSE) of the three estimators $\hat{\t}$, $\hat{\t}^{*}$, and $\hat{\t}_{disc}$ are reasonably
small, and the coverage probabilities of the 95\% confidence
intervals are very close to their nominal level, even with a relatively
modest sample size of $n=250$.
In contrast, the standard 2SLS estimator has much larger bias and standard errors when the exclusion restriction is violated.  As expected, the OLS estimator perform poorly in the presence of endogeneity. 

Our approach is applied to study the returns to education and to test the direct effects of different instruments. Our first application, in line with \cite{card1993using}, studies two indicators of college proximity: the presence of a nearby 2-year college and a nearby  4-year college. Our findings show that after controlling for regional characteristics, the two college proximity indicators have no significant effects on the outcome. However, the 2SLS estimator varies substantially when using different instruments, while our estimators remain robust under various specifications. 
In the second application, we investigate two family background variables as potential instruments: parents' education and number of siblings. The results indicate that the number of siblings exerts no significant effect on wages, while parents' education significantly increases income. The estimated returns to education based on our three estimators appear to be smaller than those of 2SLS estimators, as our methods account for the direct effects of the two instruments. 

~


The main contribution of our paper is to provide identification and estimation of endogenous linear and nonlinear regression models using only included exogenous regressors.\footnote{Our proposed method can also be applied to test exclusion restrictions. One could estimate the regression model using our method, treating all exogenous variables as included IVs. Then testing the exclusion restrictionna of a specific exogenous variable corresponds to testing whether its coefficient is zero. Compared to the classic overidentification testing, our method only requires one instrument to conduct the test.} Our approach offers an alternative solution for endogeneity when it is challenging to find excluded IVs. 
In such scenarios, empirical researchers may consider using the included exogenous variables $Z_{i}$ as the ``included
IVs'' for $X_i$ (or the entire term containing $X_i$), and adopt the mean projection of $X_i$ on $Z_i$. 
We provide a corresponding set of easy-to-use semiparametric estimation and inference procedures, along with their theoretical properties. Hence, we believe our results could have broad applicability, given the general relevance of endogenous linear and nonlinear regression models in applied work.

Our paper is most closely related to the line of econometric literature
on the identification of endogenous regression models without
exclusion restrictions. See \citet{lewbel2019identification} for
a comprehensive survey of related work on this topic. In the standard
linear regression setting, \citet{rigobon2003identification}, \citet{klein2010estimating},
\citet{lewbel2012using}, and \citet{lewbel2018identification} utilize
heteroskedasticity of error terms, while \citet*{lewbel2023identification}
works with a specific decomposition of error and imposes independence
between them. Beyond the standard linear regression setting, \citet*{dong2010endogenous}
considers a binary response model with imposed independence assumption
among error terms, \citet{kolesar2015identification} studies a linear
regression with ``many IVs'' under an orthogonality condition between
the IVs' direct effects on the outcome variable and the effects on
the endogenous covariates, and \citet*{d2021testing} considers a linear
random coefficient model with exogeneity and independence assumptions
on the random coefficients. Another relevant paper is
\citet*{escanciano2016identification}, who studies a more general framework of semiparametric conditional moment models and provides high-level conditions for identification without exclusions. 
They adopt the control function approach
and impose the (full) conditional independence assumption of errors. Furthermore, they work with  the moment equation conditional on both the endogenous regressor and the included exogenous regressor. In contrast, our approach is based on the moment equation given only the included exogenous regressor under the mean independence assumption of this regressor.
More differently, \citet*{honore2020selection,honore2022sample}
investigate partial identification of sample selection models without
exclusion.

Another relevant line of literature is on optimal IV and asymptotic
efficiency in the estimation of conditional moment restriction models: see,
for example, \citet*{amemiya1974nonlinear,amemiya1977maximum}, \citet*{chamberlin1987asymptotic}
and \citet*{newey1990efficient,NEWEY1993419}, \citet*{ai2003efficient},
and \citet*{newey2004efficient}. The main focus of this line of literature
is on asymptotic efficiency and typically assumes identification as
a starting point. Consequently, this literature does
not explicitly distinguish between included and excluded IVs or between linear and nonlinear revelance of IVs.
In addition, many papers
in this literature, such as \citet*{donald2001choosing}, \citet*{hahn2002optimal}
and \citet*{stock2005asymptotic}, are more concerned with the scenario
where there are \emph{many IVs} (which are often implicitly excluded
IVs), while we focus on exactly the opposite scenario, where researchers
\emph{do not have any} excluded IVs. In addition, \citet*{escanciano2018simple} considers endogenous linear regressions  and proposes the ``integrated IV estimator'' as a simple
and robust alternative to the optimal IV approach. However, the focus
of \citet*{escanciano2018simple} is on robustness (especially with weak instruments), and similarly it does
not distinguish between included/excluded IVs or linear/nonlinear relevance of IVs.



In the special case where $X_{i}$ is binary, there is also a connection
between our paper and the literature on heterogeneous
treatment effects. This literature, as exemplified by \citet{imbens1994identification},
\citet*{angrist1996identification}, and \citet{heckman2005structural},
studies endogenous selection and instrumental variables within the
potential outcome framework. See, e.g., \citet{imbens2014instrumental},
\citet*{imbens2015causal}, \citet{mogstad2018identification}, and
\citet{abadie2018econometric} for more comprehensive reviews. This
framework allows for nonparametrically heterogeneous treatment effects,
but usually imposes full conditional independence assumptions along
with monotone relevance conditions on the IVs. Under this framework,
the most closely related line of work is on the identification of
treatment effects without exclusion restrictions:
\citet{manski2000monotone}, \citet{flores2013partial}, and \citet{mealli2013using} establish
partial identification without exclusion. Moreover, \citet{hirano2000assessing}
relaxes the exclusion condition by applying the Bayesian approach,
while \citet{wang2022} employs an additional instrument for identification.

The paper also relates to work on endogenous nonlinear and quantile regression models, such as \cite{newey2003instrumental}, \cite{chernozhukov2005iv}, \cite*{chernozhukov2007instrumental}, and \cite{chernozhukov2008instrumental}. The existing studies explore nonparametric identification with excluded instruments. In contrast, our paper focuses on parametric models and investigates identification using only included exogenous regressors. 

Our semiparametric two-stage estimation procedure with
nonparametric regression of the endogenous/outcome variables on the
included exogenous variables are also reminiscent of \citet{robinson1988root},
who considers a partially linear regression model without endogeneity. However, one of the key steps in \citet{robinson1988root}
is to transform the regression equation into a ``differenced
form'' that is free of the unknown nonparametric function in the
original equation. In contrast, the identification arguments in our
linear regression setup does not involve the ``differenced form''
equation. A recent paper by \citet*{antoine2022partially} studies
the partially linear model with endogenous covariates.
They again work with the differenced form in the style of \citet{robinson1988root},
and then rely on excluded IVs for identification.

Lastly, our discretization-based estimator bears some resemblance
to the inferential methods for conditional moment inequalities, as studied in \citet{khan2009inference} and \citet{andrews2013},
for example. 

The rest of the paper is organized as follows. Section \ref{sec:ID}
introduces the identification of linear regression models, along with
further discussions about the identification condition. Section \ref{sec:disc} discusses the comparison of our approach with existing methods in the literature. Section
\ref{sec:Estimation} derives asymptotic distributions of our three proposed estimators and provides corresponding variance estimators. 
Section \ref{sec:nonlinear} explores the identification of nonlinear and quantile regressions. Section \ref{sec:Simulation} presents simulation results
about the finite-sample performances of our estimators. Section \ref{sec:Application} studies the returns to education and examines the direct effects of various instruments. We conclude
with Section \ref{sec:Conclusion}.

\section{\label{sec:ID}Endogenous Linear Regression without Exclusion }

\subsection{\label{subsec:main_ID}Model and Identification}

Consider the following linear regression model with endogeneity:
\begin{equation}
Y_{i}=\a_{0}+Z_{i}^{'}\b_{0}+X_{i}^{'}\g_{0}+\e_{i},\label{eq:Model}
\end{equation}
where $X_{i}$ is a $d_{x}$-dimensional endogenous regressor that
can be dependent with $\e_{i}$, while $Z_{i}$ is a $d_{z}$-dimensional
included exogenous regressor satisfying the following mean independence, or strict exogeneity, assumption:

\begin{assumption}[Mean Independence]
\label{assu:exog}$\E\left[\rest{\e_{i}}Z_{i}=z\right]=0\label{eq:ExoZ}
$ for any $z\in{\cal Z}:=\text{Supp}\left(Z_{i}\right)$. 
\end{assumption}
Writing 
$\t_{0}:=\left(\a_{0},\b_{0}^{'},\g_{0}^{'}\right)^{'}\in\R^{d:=1+d_{x}+d_{z}}$,
we are interested in identifying and estimating $\t_{0}$. Section \ref{sec:nonlinear} explores the extension of endogenous nonlinear and quantile regressions.


Assumption \ref{assu:exog} on $Z_i$ leads to the following conditional moment restriction:  
\begin{equation} \E\left[\rest{Y_{i}-\a_{0}-Z_{i}^{'}\b_{0}-X_{i}^{'}\g_{0} }Z_{i}=z\right]=0, \label{eq:E_YZ} 
\end{equation}
which characterizes the identified set for $\t_0$. We show that the above restriction can point identify $\t_0$ under the no multicollinearity condition.

Define
$\pi_{0}\left(z\right):=\E\left[\rest{X_{i}}Z_{i}=z\right]$.
By employing the mean projection of $X_i$ on $Z_i$, we can rewrite \eqref{eq:E_YZ} by replacing the endogenous regressor $X_i$ with $\pi_0\left(Z_i\right) $:
\[
 \E\left[\rest{Y_{i}-\a_{0}-Z_{i}^{'}\b_{0}-\pi_0\left(Z_i\right)^{'}\g_{0} }Z_{i}=z\right]=0.\label{eq:RegZwNoError}
\]
When treating $\pi_0\left(Z_i\right)$ as a regressor, the above condition transforms into the moment restriction of a standard linear regression, which regresses $Y_i$ on 1, $Z_i$, $\pi_{0}\left(Z_{i}\right)$. After applying the mean projection, there is no endogeneity  since $Z_i$ satisfies the strict exogeneity condition. 

Let $W_{i}:=\left(1,Z_{i}^{'},\pi_{0}\left(Z_{i}\right)^{'}\right)'$.
We then apply the usual identification strategy by premultiplying
both sides of the above equation by $W_{i}$ and then taking unconditional
expectations:
\[
\E\left[W_{i}Y_i\right]=\E\left[W_{i}W_{i}^{'}\right]\t_{0}.\label{eq:IDeq_theta0}
\]
Since 
$\pi_{0}\left(z\right)$ is
nonparametrically identified from data, the terms $W_{i}, \E\left[W_{i}W_{i}^{'}\right]$, and $\E\left[W_{i}Y_i\right]$ are also identified. It is then clear that $\t_{0}$ is identified whenever $\E\left[W_{i}W_{i}^{'}\right]$
is invertible, which boils down to the familiar requirement of no
multicollinearity condition:
\begin{assumption}[No Multicollinearity]
\label{assu:FullRankW} $\left(1,Z_{i}^{'},\pi_{0}\left(Z_{i}\right)^{'}\right)$
are not (perfectly) multicollinear. Or equivalently, $\E\left[W_{i}W_{i}^{'}\right]$ has full rank.
\end{assumption}
The discussion regarding Assumption \ref{assu:FullRankW} is presented in Section \ref{sec:NLRelevance}. 
Under this assumption, $\t_0$ is identified as the standard OLS formula with $W_i$ as the regressor:
\[ \t_{0} =\left(\E\left[W_{i}W_{i}^{'}\right]\right)^{-1}\E\left[W_{i}Y_{i}\right]. \]

Since $W_{i}$ is a deterministic function of $Z_{i}$, we can also project $Y_i$ on $Z_i$ and obtain an alternative expression for $\t_0$. Defining $ h_0\left(Z_i \right):=\E\left[\rest{Y_i} Z_i\right]$, then $\t_0$ can be expressed as
\[\t_{0}=\left(\E\left[W_{i}W_{i}^{'}\right]\right)^{-1}\E\left[W_{i}h_{0}\left(Z_{i}\right)\right],
\]
which follows from the Law of Iterated Expectations. We conduct this additional projection because, through simulation, we find that the estimator based on this formula can exhibit slightly better performance under some specifications.

\begin{thm}[Identification with Included IV]
\label{thm:IDGen} Under Assumptions \ref{assu:exog} and \ref{assu:FullRankW},
\begin{align}
\t_{0} &=\left(\E\left[W_{i}W_{i}^{'}\right]\right)^{-1}\E\left[W_{i}Y_{i}\right] =\left(\E\left[W_{i}W_{i}^{'}\right]\right)^{-1}\E\left[W_{i}h_{0}\left(Z_{i}\right)\right].\label{eq:ID_theta0}
\end{align}
\end{thm}

Theorem \ref{thm:IDGen} suggests two natural semiparametric
two-step estimators for $\t_{0}$. Specifically, given
first-stage nonparametric estimators $\hat{\pi}$
for $\pi_{0}$ and $\hat{h}$ for $h_{0}$, the second-stage plug-in estimators for $\t_{0}$ is given
by, with $\hat{W}_{i}:=\left(1,Z_{i}^{'},\hat{\pi}\left(Z_{i}\right)^{'}\right)'$,
\[\label{eq:theta_hat}
\begin{aligned}
\hat{\t} & :=\left(\frac{1}{n}\sum_{i=1}^{n}\hat{W}_{i}\hat{W}_{i}^{'}\right)^{-1}\frac{1}{n}\sum_{i=1}^{n}\hat{W}_{i}Y_{i},\\
\hat{\t}^{*} & :=\left(\frac{1}{n}\sum_{i=1}^{n}\hat{W}_{i}\hat{W}_{i}^{'}\right)^{-1}\frac{1}{n}\sum_{i=1}^{n}\hat{W}_{i}\hat{h}\left(Z_{i}\right).
\end{aligned}
\]

As shown in Section \ref{sec:Estimation}, both $\hat{\t}$ and $\hat{\t}^*$  are $\sqrt{n}$-consistent and asymptotically normal. Furthermore, they share the same asymptotic variance, and are thus asymptotically equally efficient. In the meanwhile, $\hat{\t}$ does not need nonparametric estimation of $h_{0}$, and is thus simpler and faster to compute than $\hat{\t}^*$. However, we do find that $\hat{\t}^*$ can have better finite-sample performance under certain simulation setups. Hence, we keep the estimator $\hat{\t}^*$ in our paper and provide results for it along with $\hat{\t}$.

In Section \ref{subsec:DiscTSLS}, we propose a third estimator $\hat{\t}_{disc}$ based on a discretization of the support of $Z_i$, which does not require any nonparametric regressions in the first stage. However, $\hat{\t}_{disc}$ is not directly based on the sample analog of \eqref{eq:ID_theta0}. Hence, we defer $\hat{\t}_{disc}$ to Section \ref{subsec:DiscTSLS}.

\subsection{An Alternative Perspective}

In Section \ref{subsec:main_ID}, we establish point identification of $\t_0$ from the perspective of the standard linear regression models, which naturally leads to the familiar ``no multicollinearity" or ``full rank" condition in Assumption \ref{assu:FullRankW}. A slightly different perspective is to exploit the fact that the conditional moment equation \eqref{eq:E_YZ} is a system of \emph{deterministic} linear equations in $\theta$ across all $z\in{\cal Z}$. Therefore $\t_0$ is uniquely determined if the following condition holds:

\begin{condition}[Full-Dimensional Support]
\label{cond:rd_supp}There exist $d=1+d_{x}+d_{z}$ distinct points
$z_{1},...,z_{d}\in{\cal Z}$ such that
\[
\text{rank}\left(\begin{array}{ccc}
1 & z_{1}^{'} & \pi_{0}\left(z_{1}\right)^{'}\\
1 & z_{2}^{'} & \pi_{0}\left(z_{2}\right)^{'}\\
\vdots & \vdots & \vdots\\
1 & z_{d}^{'} & \pi_{0}\left(z_{d}\right)^{'}
\end{array}\right)=d.
\]
\end{condition}

It turns out that Condition \ref{cond:rd_supp} is equivalent to Assumption \ref{assu:FullRankW}, which is also intuitively so under linearity. Hence, the two perspectives for identification are equivalent.

\begin{lem}
\label{lem:d_values}Assumption \ref{assu:FullRankW} $\iff$ Condition
\ref{cond:rd_supp}.
\end{lem}
Condition \ref{cond:rd_supp} provides an alternative perspective for identification from the support of the included instrument $Z_i$, under the feature that $\left(1, Z_{i}^{'}, \pi_{0}\left(Z_{i}\right)^{'}\right)$ is a deterministic function of $Z_{i}$.
We see that even though the dimension $d_{z}$ of the included instrument $Z_{i}$ is by
construction smaller than the number of parameters $d$ (e.g., a scalar $Z_i$),
it is still possible for us to find $d$ linearly independent \emph{realizations}
of $\left(1, Z_{i}^{'}, \pi_{0}\left(Z_{i}\right)^{'}\right)$ on the
support of $Z_{i}$, which will guarantee the required ``no multicollinearity''
assumption.

This perspective also motivates our third estimator $\hat{\t}_{disc}$ by transforming the conditional moment equation into the following unconditional moment equation:
\[\E\left[\left(Y_i-\a_{0}-Z_i^{'}\b_{0}-X_i'\g_{0} \right)\ind\{Z_i\in {\cal Z}_k\} \right]=0,
 \]
where $({\cal Z}_k)_{k=1}^K$ is a finite partition of the support of $Z_i$ with $K\geq d$. The idea of transforming conditional moments into unconditional ones using instrumental functions such as indicator functions, has been well studied and applied in the literature: e.g., \cite{khan2009}, \cite{andrews2013}, and \cite*{shi2018estimating}. See Section \ref{subsec:DiscTSLS} for more details about the discretization-based estimator $\hat{\t}_{disc}$.

\subsection{\label{sec:NLRelevance}Discussion about Assumption \ref{assu:FullRankW}}

Since Assumption \ref{assu:FullRankW} is the foundation for
the identification of $\t_{0}$, we now provide some necessary and/or
sufficient conditions for it, along with some more detailed discussions on its relationship to nonlinearity, relevance, and order condition:

\begin{condition}[No Multicollinearity in $Z_{i}$]
\label{cond:NoMulColZ}$\left(1,Z_{i}^{'}\right)$ are not multicollinear.
\end{condition}
\begin{condition}[Nonlinearity]
\label{cond:NLpi}$\pi_{0,k}\left(z\right):=\E\left[\rest{X_{i,k}}Z=z\right]$
is nonlinear in $z$ on ${\cal Z}$, for each component $k=1,...,d_{x}$.
\end{condition}
%

\begin{condition}[Relevance]
\label{cond:Relevance}$\pi_{0,k}\left(z\right):=\E\left[\rest{X_{i,k}}Z=z\right]$
is not constant in $z$ on ${\cal Z}$, for each component $k=1,...,d_{x}$.
\end{condition}

%
\begin{condition}[Order Condition on ${\cal Z}$]
\label{cond:Cardinality} The support of $Z_{i}$ must contain $d$
distinct points, i.e., $\#\left({\cal Z}\right)\geq d = 1+d_{x}+d_{z}$. 
\end{condition}

Clearly, all of the above are necessary conditions for Assumption \ref{assu:FullRankW}:
\begin{lem}
\label{lem:A1impNLpi}(a) Assumption \ref{assu:FullRankW} implies Conditions \ref{cond:NoMulColZ} and \ref{cond:NLpi}; (b) Condition \ref{cond:NLpi} implies Conditions \ref{cond:Relevance} and \ref{cond:Cardinality}.
\end{lem}
The no multicollinearity condition and the relevance condition are standard for linear regression models. Condition \ref{cond:NLpi} requires $Z_{i}$ to
be relevant for $X_{i}$ in a nonlinear manner. 
This requirement of nonlinearity marks the departure of our approach from the standard IV approach which utilizes a linear projection of $X_i$ on $Z_i$. 

The requirement of nonlinearity also imposes a restriction on the cardinality
of the support of $Z_{i}$ as in Condition \ref{cond:Cardinality}. This is because it is always possible to fit a straight line between any two distinct points, and more generally, to fit a linear $d$-dimensional hyperplane across any $d$ distinct points in $\R^d$. Hence, our order condition is on the cardinality of the support of $Z_i$, rather than the number of variables. Of course, if ${\cal Z}$ is a continuum, then
the order condition is automatically satisfied.

~

When there is only one endogenous variable, then the converse of Lemma
\ref{lem:A1impNLpi}(a) is also true, effectively establishing the
sufficiency of nonlinearity for point identification.
\begin{lem}[Sufficient Condition with Scalar $X_{i}$]
\label{lem:SuffScalarX}Suppose that $X_{i}$ is scalar-valued, i.e.
$d_{x}=1$. Then, Conditions \ref{cond:NoMulColZ} and \ref{cond:NLpi}
$\imp$ Assumption \ref{assu:FullRankW}.
\end{lem}
Lemma \ref{lem:SuffScalarX} is particularly relevant when we are primarily worried about the endogeneity
of a single treatment status variable $X_{i}$, which is often
a discrete random variable. Then, if there exists some exogenous
shifter $Z_{i}$ that is relevant for $X_{i}$,  $\pi_{0}\left(z\right)$
is naturally nonlinear given the discreteness of $X_{i}$. 

\begin{example}[Linear Treatment Effect Model with Selection]
\label{exa:TE} Consider 
\begin{align*}
Y_{i} & =\a_{0}+Z_{i}^{'}\b_{0}+X_{i}\g_{0}+\e_{i},\\
X_{i} & =\ind\left\{ \varphi_{0}\left(Z_{i}\right)\geq u_{i}\right\}, 
\end{align*}
with $\E\left[\rest{\e_{i}}Z_{i}\right]=0$, $u_{i}\indep Z_{i}$, and
$u_{i}\sim F_{u}$. Then, the propensity
score function $\pi_{0}\left(z\right):=\E\left[\rest{X_{i}}Z_{i}=z\right]
$ is naturally nonlinear in $z$ when $\#(Z_i) \geq 3$, i.e., the support of $Z_i$ contains at least three points. 
As discussed in the introduction, the order condition $\#(Z_i) \geq 3$ can be satisfied even if $Z_i$ just consists of two dummy variables. Hence,
Condition \ref{cond:NLpi} can be thought as a  mild condition in this setting.
\end{example}
Lastly, we note that, when $d_{x}>1$, we not only need each $\pi_{0,k}$
to be nonlinear in $z$, but also need each $\pi_{0,k}$ to be linearly
independent (as a function) from 1, $z$, and all other $\left(\pi_{0,j}\right)_{j\neq k}$
as well. We consider this condition relatively mild and easy to verify. Heuristically, whenever ${\cal Z}$
is a continuum, the space
of functions on ${\cal Z}$ (under some regularity conditions) can be often viewed as an infinite-dimensional Hilbert space that admits a linear series representation under a certain orthonormal basis of functions
$\left(b_{k}\left(\cd\right)\right)_{k=1}^{\infty}$ on ${\cal Z}$:
\[
{\cal F}=\left\{ \sum_{k=1}^{\infty}c_{k}b_{k}\left(\cd\right):\sum_{k=1}^{\infty}c_{k}^{2}<\infty\right\} .
\]
Hence, linear independence among a finite number $\left(d=1+d_{x}+d_{z}\right)$
of ``generic" functions from $\cal{F}$ seems heuristically as a ``generic property".

\section{Discussion}
\label{sec:disc}

\subsection{Comparison: IV regression with Excluded Instrument}

The canonical IV approach utilizes \textit{a linear projection} of the endogenous regressor on the exogenous regressors. This approach requires the presence of an excluded instrument, as otherwise all regressors will exhibit perfect multicollinearity. In principle, this method can achieve nonparametric identification under the completeness condition and is robust to the misspecification of function forms.\footnote{In practice, however, nonparametric IV estimation is not commonly used, partly due to the computational difficulties and inference complexities.}

In contrast, our approach exploits \textit{a mean projection} of the endogenous regressor on the included exogenous regressor, which allows us to extract more information for identification through the nonlinear dependence between the exogenous variable and the endogenous variable. Our approach relies on a parametric (e.g., linear) assumption on the functional form, but it enables identification without exclusion restrictions. We believe our method could be a viable alternative to the standard IV approach in situations where there is a natural parametric specification, and where finding excluded IVs is challenging.

\subsection{Comparison: Heckman Correction Approach}

The conventional Heckman correction approach can also achieve identification without exclusion restrictions, under distributional assumptions or parametric functions. This approach typically focuses on a binary endogenous regressor and examines the following specification:
\[\begin{aligned}
Y_{i} & =\a_{0}+Z_{i}^{'}\b_{0}+X_{i}\g_{0}+\e_{i},\\
X_{i} & =\ind\left\{ Z_{i}'\eta_0 \geq u_{i}\right\}, \\
(\e_i, u_i)' & \sim \mathcal{N}\left((0, 0)', (1, \rho_0; \rho_0, 1) \right).
\end{aligned}
\]

Under the joint distribution of the two error terms $(\e_i, u_i)$, it yields the following conditional moment restriction:
\[ \E[\rest{Y_i}{X_i, Z_i}]=\a_0+Z_i'\b_0+X_i'\g_0+\rho_0 \frac{\phi(Z_{i}'\eta_0)}{\Phi(Z_{i}'\eta_0)}, \]
which can identify $(\t_0, \rho_0)$ without exclusion restrictions.
The Heckman correction approach can be extended to nonbinary and multi-dimensional endogenous regressor $X_i$. We can still look at the conditional expectation of $Y_i$ given all regressors $(X_i, Z_i)$:
\[ \E[\rest{Y_i}{X_i, Z_i}]=\a_0+Z_i'\b_0+X_i'\g_0+\E[\rest{\e_i}{X_i, Z_i}].
\]
If a parametric form on the selection bias term $\E[\rest{\e_i}{X_i, Z_i}]$ is imposed as follows:
\[\E[\rest{\e_i}{X_i, Z_i}]=s(X_i, Z_i, \eta_0),
\]
and this function $s$ is nonlinear in $(X_i, Z_i)$, then the coefficient $\t_0$ is identified.

The Heckman correction approach exploits the mean projection of the error $\e_i$ on all regressors $(X_i, Z_i)$. To achieve identification, this approach requires a parametric form (or parametric distributions of errors) for $s$ as well as the nonlinearity of $s$. However, since the function $s$ involves the unobserved error term $\e_i$, its nonlinearity cannot be directly tested.

In contrast, our approach applies the mean projection of the endogenous $X_i$ on the exogenous $Z_i$. The identification relies on the nonlinearity of the function $\pi_0(Z_i)=\E[\rest{X_i}{Z_i}]$, but does not require further functional form assumption on $\pi_0$. Moreover, the function $\pi_0$ only depends on observed variables $(X_i, Z_i)$, making its nonlinearity a testable condition.

\subsection{Relationship to the Instrumental Function Approach}

We establish identification of $\t_0$ from the viewpoint of a standard linear regression model, while treating $\pi_0(Z_i)$ as an exogenous regressor. Point identification is then obtained under the no-multicollinearity condition, which translates into a nonlinearity requirement on $\pi_0$. Another approach to address endogeneity is to use a (known) nonlinear function $g(Z_i)$ as an instrument for the endogenous regressor $X_i$. In this section, we will discuss the connections between our identification approach and this alternative approach.

To illustrate, consider the case where both $Z_i$ and $X_i$ are scalar variables. Using $(1, Z_i, g(Z_i))$ as IVs, we can obtain the following moment restrictions for $\t_0$:
\[ \E\left[(Y_i-\a_0-\b_0Z_i-\g_0X_i)\left(
\begin{array}{c}
1 \\
Z_i\\
g(Z_i)
\end{array} 
\right)
\right]=0.\]
The parameter $\t_0$ is identified from the above equation if 
\[H_g:=
\left[\begin{array}{ccc}
1 &\E[Z_i] &\E[X_i]\\
\E[Z_i] & \E[Z_i^2] & \E[X_iZ_i]\\
\E[g(Z_i)] &\E[Z_i g(Z_i)] &\E[X_i g(Z_i)]
\end{array}
\right]
=\left[\begin{array}{ccc}
1 &\E[Z_i] &\E[\pi_0(Z_i)]\\
\E[Z_i] & \E[Z_i^2] & \E[\pi_0(Z_i)Z_i]\\
\E[g(Z_i)] &\E[Z_i g(Z_i)] &\E[\pi_0(Z_i)g(Z_i)]
\end{array}
\right]
\]
has full rank, which depends on the functional form of $\pi_0$ and the choice of $g$. 

Clearly, a necessary condition for the full rank requirement is the nonlinearity of $\pi_0$. Otherwise, the third column of $H_g$ will be a linear combination of the first two columns,\footnote{Writing $\pi_0(Z_i) = a + bZ_i$, we have $\E[\pi_0(Z_i)] = a + b\E[Z_i]$, $\E[\pi_0(Z_i)Z_i] = a\E[Z_i] + b\E[Z_i^2]$, and $\E[\pi_0(Z_i)g(Z_i)] = a\E[g(Z_i)] + b\E[Z_i g(Z_i)]$.} and hence $H_g$ cannot have full rank, irrespective of the choice of $g$. Our identification results thus make explicit the dependence of the identifiability on the nonlinearity of the $\pi_0$ function.

Moreover,  it is worth noting that not all nonlinear functions can serve as valid IVs for $X_i$ in the sense of satisfying the full rank condition on $H_g$. For example, if $Z_i \sim \cN(0,1)$, $\pi_0(z)=z^3$, and $g(z)=z^2$, then $H_g$ has deficient rank, and thus $Z_i^2$ is not a valid IV.

Our identification results in Theorem \ref{thm:IDGen} can be interpreted as using $\pi_0(Z_i)$ as an instrument for the endogenous regressor $X_i$, which is an unknown function that can be identified from data. As shown in \cite{chamberlin1987asymptotic} and \cite{newey1990efficient}, $\pi_0$ is in fact the optimal instrument under homoskedasticity. The identification results in our paper, Lemma \ref{lem:SuffScalarX} in particular, further imply that, with $\pi_0(Z_i)$ used as the IV, the nonlinearity of $\pi_0$ becomes \textit{sufficient} for the full rank condition. Therefore, $\pi_0(Z_i)$ is not only the instrumental function that minimizes the asymptotic variance under homoskedasticity, but also the instrumental function that requires the minimum assumption for identification.

\section{\label{sec:Estimation}Estimation and Inference}

\subsection{Semiparametric Estimators $\hat{\t}$ and $\hat{\t}^*$ }\label{subsec:Est_Semipara}

Based on our identification result, we propose the following two semiparametric estimators:
\[
\begin{aligned}
\hat{\t} & =\left(\frac{1}{n}\sum_{i=1}^{n}\hat{W}_{i}\hat{W}_{i}^{'}\right)^{-1}\frac{1}{n}\sum_{i=1}^{n}\hat{W}_{i}Y_{i},\\
\hat{\t}^{*} & =\left(\frac{1}{n}\sum_{i=1}^{n}\hat{W}_{i}\hat{W}_{i}^{'}\right)^{-1}\frac{1}{n}\sum_{i=1}^{n}\hat{W}_{i}\hat{h}\left(Z_{i}\right).
\end{aligned}
\]
We now lay out the regularity conditions for the
$\sqrt{n}$-consistency and asymptotic normality of $\hat{\t}$ and $\hat{\t}^{*}$. The first one is a standard one on the existence of moments.\footnote{We impose this assumption on the fourth moment for subsequent variance estimation.}
\begin{assumption}[Finite Fourth Moments]
\label{assu:Moment2} $\E\left|\e_{i}\right|^{4},\E\norm{X_{i}}^{4}$,
and $\E\norm{Z_{i}}^{4}$ are finite.
\end{assumption}
\noindent Below we give some high-level conditions about the first-stage nonparametric regressions, which can be satisfied with a wide variety of lower-level
conditions and many types of nonparametric estimators. See, for example, \citet*{newey1994large} and \citet*{chen2007large} for more information.
\begin{assumption}[Smoothness and Nonparametric Convergence]
\label{assu:NPFS} Suppose that:
\begin{itemize}
\item[(a)]  $h_{0},\pi_{0}\in{\cal H}$, where ${\cal H}$ is a Sobolev function
space of order $s>\frac{d_{z}}{2}$ on ${\cal Z}$.
\item[(b)]  The nonparametric estimators $\hat{h}$ and $\hat{\pi}$ belong
to ${\cal H}$ (with probability approaching 1) and are asymptotically
linear.
\item[(c)]  The nonparametric estimators $\hat{h}$ for $h_{0}$ and $\text{\ensuremath{\hat{\pi}}}$
for $\pi_{0}$ converge in $L_{2}\left(Z\right)$-norm faster than
the $n^{-1/4}$ rate: $\norm{\hat{h}-h_{0}}_{L_{2}\left(Z\right)}=o_{p}\left(n^{-\frac{1}{4}}\right)$, $\norm{\hat{\pi}-\pi_{0}}_{L_{2}\left(Z\right)}=o_{p}\left(n^{-\frac{1}{4}}\right)$.
\end{itemize}
\end{assumption}
\begin{thm}[Asymptotic Normality]
\label{thm:AsymNorm} Under Assumptions \ref{assu:exog} - \ref{assu:NPFS},
we have:
\begin{align*}
\sqrt{n}\left(\hat{\t}-\t_{0}\right)\dto\cN\left({\bf 0},V_{0}\right) & ,\quad\sqrt{n}\left(\hat{\t}^{*}-\t_{0}\right)\dto\cN\left({\bf 0},V_{0}\right),
\end{align*}
with
$$
V_{0}:= \E\left[W_{i}W_{i}^{'}\right]^{-1} \E\left[\e_{i}^{2}W_{i}W_{i}^{'}\right] \E\left[W_{i}W_{i}^{'}\right]^{-1}. 
$$ 
The assumptions about $h_{0}$ and $\hat{h}$ can be
dropped for $\hat{\t}$ since it does not involve $\hat{h}$. Also, if $\E[\rest{\e_i^2} Z_i]\equiv\s^2_\e$ (errors are homoskedastic), $V_0$ simplifies to $ \s^2_\e  \E\left[W_{i}W_{i}^{'}\right]^{-1}$. 
\end{thm}
The asymptotic variance can then be easily estimated via standard plug-in methods as in the following theorem. Based on the standard error estimates, confidence intervals and various test statistics can be computed in the standard manner.
\begin{thm}[Variance Estimation]
\label{thm:VarEst} Let $\hat{\e}_{i}:=Y_{i}-\hat{\a}-Z_{i}^{'}\hat{\b}-X_{i}^{'}\hat{\g},$
and 
\begin{align*}
\hat{\Sigma}:=\frac{1}{n}\sum_{i=1}^{n}\hat{W}_{i}\hat{W}_{i}^{'},\quad & \hat{\O}:=\frac{1}{n}\sum_{i=1}^{n}\hat{\e}_{i}^{2}\hat{W}_{i}\hat{W}_{i}^{'},\quad \hat{V}:=\hat{\Sigma}^{-1}\hat{\O}\hat{\Sigma}^{-1}.
\end{align*}
With homoskedasticity, $\hat{V}:= \left(\frac{1}{n}\sum_{i=1}^{n}\hat{\e}_{i}^{2}\right)\hat{\Sigma}^{-1}$. Under Assumptions \ref{assu:exog} - \ref{assu:NPFS},
$\hat{V}\pto V_{0}.$ 
\end{thm}

\subsection{\label{subsec:DiscTSLS}Discretization-Based Estimator $\hat{\t}_{disc}$}

The two estimators $\hat{\t}$ and $\hat{\t}^{*}$ we proposed before
both involve nonparametric regressions in the first stage. Alternatively,
we propose a third estimator $\hat{\t}_{disc}$ that does \emph{not}
require \emph{any} nonparametric regression at all.

Specifically, let $\left({\cal Z}_{k}\right)_{k=1}^{K}$ be a partition
of ${\cal Z}$ with $K$ being a finite and fixed number such that $K\geq d=1+d_{z}+d_{x}$. To rule out redundant cells, we require each cell to have a positive probability.
\begin{assumption}[Positive Probabilities]
\label{assu:PositveiP} $p_k := \P\left(Z_{i}\in{\cal Z}_{k}\right)>0$ for $k=1,...,K$. 
\end{assumption}
Define the dummy variable for each of the $K$ partition cells as
$D_{i,k}:=\ind\left\{ Z_{i}\in{\cal Z}_{k}\right\}$,
and write $D_{i}:=\left(D_{i1},...,D_{iK}\right)^{'}$. We can then
use $D_{i}$ as IVs to identify and estimate $\t_0$
based on the following transformation of equation \eqref{eq:Model}:
\begin{equation*}
\E\left[D_{i}Y_{i}\right]=\a_{0}\E\left[D_{i}\right]+\E\left[D_{i}Z_{i}^{'}\right]\b_{0}+\E\left[D_{i}X_{i}^{'}\right]\g_{0}.
\end{equation*}
Since $Z_{i}$ is averaged out within each partition cell, there is
some information loss, and the no-multicollinearity condition for
identification of $\t_{0}$ in Assumption \ref{assu:FullRankW} needs
to be strengthened to a partitional version. To state the condition
in a ``lower-level'' form, write $\ol Z_{k}:=\E\left[\rest{Z_{i}}Z_{i}\in{\cal Z}_{k}\right]$,
 $\ol X_{k}:=\E\left[\rest{X_{i}}Z_{i}\in{\cal Z}_{k}\right]$,
and $\ol W_{k}:=\left(1,\ol Z_{k}^{'},\ol X_{k}^{'}\right)^{'}$.
\begin{assumption}[No Partitional Multicollinearity]
\label{assu:NoMC_disc} Suppose that $\left(1,\ol Z_{k}^{'},\ol X_{k}^{'}\right)$
are not multicollinear across $k=1,...,K$, or equivalently,
\[
\text{rank}\left(\begin{array}{ccc}
1 & \ol Z_{1}^{'} & \ol X_{1}^{'}\\
\vdots & \vdots & \vdots\\
1 & \ol Z_{K}^{'} & \ol X_{K}^{'}
\end{array}\right)=d.
\]
\end{assumption}
We note that Assumption \ref{assu:NoMC_disc} translates into the following standard full-rank condition written in terms of expectations (i.e., probability-weighted sums under discreteness), provided that each cell has a strictly positive probability.
\begin{lem}
    Suppose that Assumption \ref{assu:PositveiP} holds. Then Assumption \ref{assu:NoMC_disc} holds if and only if $\sum_{k=1}^{K}p_k\ol W_{k}\ol W_{k}^{'}$ is invertible.
\end{lem}

Note that a necessary condition for Assumption \ref{assu:NoMC_disc} is the order condition $K\geq d$ already mentioned above. It is
also easy to verify that Assumption \ref{assu:NoMC_disc}
implies Assumption \ref{assu:FullRankW}, but the converse is not
generally true. However, Assumption \ref{assu:NoMC_disc} still remains
as a condition nonparametrically identified from the observable distribution
of data.

We can then construct $\hat{\t}_{disc}$ as the standard two-stage
least square (2SLS) estimator with the $K$-dimensional vector $D_{i}$
as instruments. Formally, write $\tilde{W}_{i}:=\left(1,Z_{i}^{'},X_{i}^{'}\right)^{'}$, and let
 $Y,D,\tilde{W}$ denote the  vector/matrix concatenation of the variables
across all $i=1,...,n$, and each row of which contains $Y_{i},D_{i}^{'},\tilde{W}_{i}^{'}$, respectively. Then
\begin{align}
\hat{\t}_{disc} & :=\left(\tilde{W}^{'}P_{D}\tilde{W}\right)^{-1}\tilde{W}^{'}P_{D}Y,\label{eq:thetahat_disc}
\end{align}
where $P_{D}:=D\left(D^{'}D\right)^{-1}D^{'}$. Since $D$ consists
of partition cell dummies, the projection matrix $P_{D}$ is essentially
computing cell-wise averages, and thus $\hat{\t}_{disc}$ can be equivalently
written as
\begin{equation}
\hat{\t}_{disc}=\left(\sum_{k=1}^{K}\hat{p}_k\hat{\ol W}_{k}\hat{\ol W}_{k}^{'}\right)^{-1}\sum_{k=1}^{K}\hat{p}_k\hat{\ol W}_{k}\hat{\mu}_{y,k},\label{eq:thetahat_disc_K}
\end{equation}
where $\hat{p}_k := \frac{n_k}{n}$, $n_{k}:=\sum_{i=1}^{n}D_{ik}$, $\hat{\mu}_{y,k}:=\frac{1}{n_{k}}\sum_{i=1}^{n}D_{ik}Y_{i},$
$\hat{\mu}_{z,k}:=\frac{1}{n_{k}}\sum_{i=1}^{n}D_{ik}Z_{i}$, $\hat{\mu}_{x,k}:=\frac{1}{n_{k}}\sum_{i=1}^{n}D_{ik}X_{i}$,
and $\hat{\ol W}_{k}:=\left(1,\hat{\mu}_{z,k}^{'},\hat{\mu}_{x,k}^{'}\right)^{'}$.

Clearly, $\hat{\t}_{disc}$ is very easy to compute. Researchers
may use any standard 2SLS  command with
$D_{i}$ as IVs (or with
one of $D_{ik}$'s dropped if the constant is included), which yields equivalent results as \eqref{eq:thetahat_disc} and \eqref{eq:thetahat_disc_K}.
The asymptotic distribution of
$\hat{\t}_{disc}$ is derived as follows:
\begin{thm}[Asymptotic Normality of $\hat{\t}_{disc}$]
\label{thm:AsymDisc}Under Assumptions \ref{assu:exog},  \ref{assu:Moment2}, \ref{assu:PositveiP}, and
\ref{assu:NoMC_disc}, $\sqrt{n}\left(\hat{\t}_{disc}-\t_{0}\right)\dto\cN\left({\bf 0},V_{0,disc}\right)$
with 
\begin{equation}
V_{0,disc}:=\left(\sum_{k=1}^{K}p_{k}\ol W_{k}\ol W_{k}^{'}\right)^{-1}\sum_{k=1}^{K}p_{k}\ol{\s}_{\e,k}^{2}\ol W_{k}\ol W_{k}^{'}\left(\sum_{k=1}^{K}p_{k}\ol W_{k}\ol W_{k}^{'}\right)^{-1}\label{eq:V0_disc}
\end{equation}
where $\ol{\s}_{\e,k}^{2}:=\E\left[\rest{\e_{i}^{2}}Z_{i}\in{\cal Z}_{k}\right].$  
Furthermore, a consistent estimator $\hat{V}_{disc}$ for $V_{0,disc}$
can be constructed by plugging $\hat{p}_{k}:=\frac{n_{k}}{n}$ in
place of $p_{k}$, $\hat{\ol W_{k}}$ in place of $\ol W_{k}$, and
$\hat{\ol{\s}}_{\e,k}^{2}:=\frac{1}{n_{k}}\sum_{i:Z_{i}\in{\cal Z}_{k}}(Y_{i}-\tilde{W}_{i}^{'}\hat{\t}_{disc})^{2}$ 
in place of $\ol{\s}_{\e,k}^{2}$ in the formula \eqref{eq:V0_disc}
above. 

Under homoskedasticity $\E[\rest{\e_i^2} Z_{i}\in{\cal Z}_{k}]\equiv\s^2_\e$, the asymptotic variance simplifies to $V_{0,disc} = \s^2_\e \left(\sum_{k=1}^{K}p_{k}\ol W_{k}\ol W_{k}^{'}\right)^{-1}$ with $\hat{\s}^2_\e := \frac{1}{n}\sum_{i=1}^n(Y_{i}-\tilde{W}_{i}^{'}\hat{\t}_{disc})^{2}$ consistent for $\s^2_\e$.
\end{thm}

Since $\hat{\t}_{disc}$ is constructed based on averages over each partition cell ${\cal Z}_{k}$, there is in general some information loss, and thus $\hat{\t}_{disc}$ tends to be less efficient than $\hat{\t}$ and $\hat{\t}^*$. Our next result formalizes this efficiency loss in the setting where $\e_i$ is homoskedastic.

\begin{thm}\label{lem:DiscLessEff} Suppose that $\E[\rest{\e_i^2} Z_i]\equiv\s^2_\e$. Then $V_{0, disc} - V_0$ is positive semi-definite. 
\end{thm}

Despite the efficiency loss, the discretization-based estimator $\hat{\t}_{disc}$ is very simple and user-friendly. Applied researchers just need to create dummy variables for a chosen partition of ${\cal Z}$ and run a standard 2SLS command. Furthermore, in our simulations, we find that $\hat{\t}_{disc}$ performs surprisingly well in finite sample: the efficiency loss of $\hat{\t}_{disc}$ tends to be quite small and more than compensated by its smaller finite-sample bias as a 2SLS estimator that does not require nonparametric regressions.

Furthermore, in the special case where $Z_{i}$ are discrete variables with finite support, there is clearly no information loss from discretization. In this case, expectations simplify to weighted sums over the $K$ realizations of $Z_i$ (weighted by the probability mass $p_k$), and it can be easily shown that the asymptotic variance of $\hat{\t}_{disc}$ coincides with the one for $\hat{\t}$ and $\hat{\t}^{*}$ in Theorem \ref{thm:AsymNorm} (without the homoskedasticity assumption).
\begin{cor}
Suppose that ${\cal Z}=\left\{ z_{1},...,z_{K}\right\} $ for some
finite $K$. Then, under the element-by-element partition, i.e., ${\cal Z}_{k}:=\left\{ z_{k}\right\}$,
we have $V_{0,disc}=V_{0}.$
\end{cor}

While the established results for $\hat{\theta}_{disc}$ hold for any choice of partition $({\cal Z}_k)_{k=1}^K$ that satisfy Assumptions \ref{assu:PositveiP} and \ref{assu:NoMC_disc},  it is recommended in practice to choose ${\cal Z}_k$ in such a way that the cell probabilities $p_k$ are comparable in magnitude. To illustrate, consider a simple example where we partition ${\cal Z}$ into $K = 10$ cells with $p_1 = 0.91$ but $p_2=...=p_{10} = 0.01.$ Despite Assumption \ref{assu:NoMC_disc} holds (so that $\t_0$ is identified), the estimator 
$\hat{\t}_{disc}$ is likely to  perform badly, since there are only a few observations in cell $2,...,K$ and thus the sample average estimation in those cells could be highly imprecise. Moreover, since $p_2,...,p_{10}$ are close to zero, the smallest eigenvalue of $\sum_{k=1}^{K}p_{k}\ol W_{k}\ol W_{k}^{'}$ may be close to zero (unless the corresponding $\ol{W}_k$'s are very large in magnitude, so that the product terms $p_k\ol W_{k}\ol W_{k}^{'}$ stay comparable across $k$). Since the estimator $\hat{\t}_{disc}$ is based on the inverse of $\sum_{k=1}^{K}p_{k}\ol W_{k}\ol W_{k}^{'}$, its variance can be large because of the imbalance between $p_1$ and $p_2,...,p_{10}$.

If $Z_i$ is a scalar, a natural strategy would be to choose the partition $({\cal Z}_k)$ to be the $K$ equally-sized quantile ranges, which would ensure that $p_k \equiv 1/K$ (or at least asymptotically so when sample quantiles are used in finite samples). If $Z_i$ is vector, one could work with (empirical) vector quantiles as developed relatively recently in the literature based on the theory of optimal transport: see \cite{galichon2018optimal} for an introduction, and, e.g., \cite{chernozhukov2017monge}, \cite{hallin2021distribution}, and \cite{ghosal2022multivariate} for detailed discussions. Alternatively, one could start with a partition of the support of $Z_i$ obtained as products of partitions in each dimension of $Z_i$, and adjust and/or merge certain cells (if necessary) to ensure that the sample proportions of observations in each cell are comparable across $k = 1,...,K$.


We emphasize again that our results above apply for any choice of the partition $({\cal Z}_k)$ as long as Assumptions \ref{assu:PositveiP} and \ref{assu:NoMC_disc} are satisfied.
Hence, while we provide some suggestions for the choice of partitions above, there may be more appropriate partition choices depending on the specific applications and contexts.

\section{Extension: Nonlinear and Quantile Regressions}
\label{sec:nonlinear}
The identification strategy can be extended to analyze the following endogenous nonlinear regression models without exclusion restrictions:
\[
Y_{i}=f(Z_{i}, X_i, \t_{0})+\e_{i},
\]
where $X_{i}$ is a $d_{x}$-dimensional endogenous regressor, $Z_{i}$ is a $d_{z}$-dimensional exogenous regressor satisfying Assumption \ref{assu:exog}, and the function $f$ is known up to the $d$-dimensional parameter $\t_0$. 
The function $f$ can be nonlinear and nonseparable in the covariates $Z_i$ and $X_i$. 

To identify $\t_0$, we adopt a similar strategy by projecting the entire functional term $f(Z_{i}, X_i, \t_{0})$ onto the included instrument $Z_i$. Let the function $m_0$ be defined as
\[m_0(Z_i, \t_0):=\E[\rest{f(Z_{i}, X_i, \t_{0})}Z_i], \]
 which is identified up to the parameter $\t_0$.
Then under Assumption \ref{assu:exog} (exogeneity) of the included instrument $Z_i$, we have the following conditional moment condition:
\[\E[\rest{Y_i-m_0(Z_i, \t_0)}Z_i]=0.\]
The above moment condition can be viewed as the moment restriction of the standard nonlinear regression model without endogeneity, while treating $m_0(Z_i, \t_0)$ as the nonlinear regressor. The key distinction is that the function $m_0$ needs to be estimated. 
For standard nonlinear regression, local identification of $\t_0$ can be attained under the following condition.

\begin{thm}
\label{thm:nonlinear}
Suppose that Assumption \ref{assu:exog} holds, the function $m_0(z, \cdot)$ is continuously differentiable for any $z$, and $\E\left[\Dif_\t m_0(Z_i, \t_0)\Dif_{\t^{'}} m_0(Z_i, \t_0)\right]$ has full rank, then $\t_0$ is locally identified. 
\end{thm}

Local identification of $\t_0$ ensures that there exists a neighborhood $\T_0$ of $\t_0$ on which $\t_0$ is identified.
This result can be expanded to achieve global identification under additional assumptions, by invoking the global inversion theorem in \cite{ambrosetti1995primer} (Chapter 3, Theorem 1.8).
In the case of general nonlinear regressions, the interpretation of the full-rank condition depends on the specific functional form of $f(Z_{i}, X_i, \t_{0})$. This feature also applies to the standard IV regression with excluded instruments, where the identification conditions are contingent upon the specification of $f$ as well.

\begin{remark}
We focus on the nonlinear regression model, while the analysis also applies to a more general parametric model. Consider that we have the following moment condition:
\[ \tilde{m}_0(Z_i, \t_0):= \E[ \rest{g\left(Y_i, Z_i, X_i, \t_0\right) }Z_i ]=0, \]
 where the function $g$ is known up to the parameter $\t_0$ and $g$ can be nonlinear and nonseparable in all variables $(Y_i, Z_i, X_i)$. The moment function $g$ may be derived from structural models, which is naturally nonlinear in all variables. In terms of the nonlinear regression model, the function $g$ is given as $g\left(Y_i, Z_i, X_i, \t_0\right)=Y_i-f(Z_{i}, X_i, \t_{0})$. Following Theorem \ref{thm:nonlinear},
 the parameter $\t_0$ is locally identified if $\E\left[\Dif_\t \tilde{m}_0(Z_i, \t_0)\Dif_{\t^{'}} \tilde{m}_0(Z_i, \t_0)\right]$ has full rank.
Section \ref{subsec:quantile} explores the endogenous quantile regression model, where the moment condition is nonseparable in all observed variables $(Y_i, Z_i, X_i)$. 

\end{remark}

Similar to linear regression models, we propose two semiparametric two-step nonlinear regression estimators. In the first step, nonparametrically regress $f(X_i, Z_i, \t)$ on $Z_i$ for each $\t$ and get the predicted value $\hat{m}(Z_i, \t)$; nonparametrically regress $Y_i$ on $Z_i$ and get the predicted value $\hat{h}(Z_i)$. In the second step, run the standard nonlinear regression using $Y_i$ and $\hat{h}(Z_i)$ as the dependent variable, respectively:
\[\begin{aligned}
\hat{\t}_{nl}&=\arg\min_{\t \in \T_0} \frac{1}{n}\sum_i \left(Y_i- \hat{m}(Z_i, \t) \right)^2, \\
\hat{\t}^{*}_{nl}&=\arg\min_{\t \in \T_0 } \frac{1}{n}\sum_i \left(\hat{h}(Z_i)- \hat{m}(Z_i, \t) \right)^2.
\end{aligned}
\]
Similar to Section \ref{subsec:Est_Semipara}, the two estimators $\hat{\t}_{nl}, \hat{\t}^{*}_{nl}$ are $\sqrt{n}$-consistent and have the same asymptotic variance:
\[\sqrt{n} \left( \hat{\t}_{nl}-\t_0 \right) \dto \cN\left({\bf 0},V_{0}\right), \quad \sqrt{n} \left( \hat{\t}_{nl}^{*}-\t_0 \right) \dto \cN\left({\bf 0},V_{0}\right),
\]
with 
\[ V_0:=M_0^{-1} \E\left[ \e_i^2 \Dif_{\t} m_0(Z_i, \t_0)  \Dif_{\t'} m_0(Z_i, \t_0) \right] M_0^{-1}, 
\]
where $M_0=\E\left[\Dif_\t m_0(Z_i, \t_0)\Dif_{\t^{'}} m_0(Z_i, \t_0)\right]$.
A consistent estimator $\hat{V}$ for the variance matrix $V_0$ can be developed by replacing $\e_i$ with $Y_i-f(Z_i, X_i, \hat{\t}_{nl})$, $\Dif_{\t} m_0(Z_i, \t_0)$ with its estimator $\Dif_{\t} \hat{m}(Z_i, \hat{\t}_{nl})$, and expectation with the sample mean.


Next, we investigates endogenous quantile regression as an illustration of Theorem \ref{thm:nonlinear}.

\subsection{Endogenous Quantile Regressions}\label{subsec:quantile} 


We study the following endogenous quantile regression model:
\[\begin{aligned}
&Y_{i}=\a_{0}+Z_{i}^{'}\b_{0}+X_{i}^{'}\g_{0}+\e_{i},
&\ \text{Quan}_{\tau}(\rest{\e_i} Z_i)=0,
\end{aligned}
\]
where $\text{Quan}_{\tau}(\rest{\e_i} Z_i)$ denotes the $\tau$-th quantile of the conditional distribution of $\e_i$ given $Z_i$. In this example, $X_i$ is the potentially endogenous regressor and $Z_i$ is the exogenous regressor that satisfies the conditional quantile restriction. 
We still study the identification and estimation of the coefficient $\t_0$ using only the included instrument $Z_i$.

By the quantile exogeneity of $Z_i$, it yields the conditional moment restriction as follows:
\[ \E\left[\rest{\ind\{Y_i\leq \a_{0}+Z_{i}^{'}\b_{0}+X_{i}^{'}\g_{0}\}-\tau} Z_i \right]=0. \]
The moment condition above is naturally nonlinear and nonseparable in all variables $(Y_i, Z_i, X_i)$. We project the whole indicator term on $Z_i$ and 
define the function $m_0$  as
\[m_0(Z_i, \t_0):=\E\left[\rest{\ind\{Y_i\leq \a_{0}+Z_{i}^{'}\b_{0}+X_{i}^{'}\g_{0}\}} Z_i  \right]. \]
According to Theorem \ref{thm:nonlinear}, the coefficient $\t_0$ is locally identified if $\E\left[\Dif_\t m_0(Z_i, \t_0)\Dif_{\t^{'}} m_0(Z_i, \t_0)\right]$ has full rank. Lemma \ref{lem:quantile}  presents an alternative condition that is equivalent to this full rank condition, making it easier to interpret. 

\begin{assumption}[Continuous Errors]
\label{assu:eps_pdf}The error term $\e_{i}$ conditional on $(x, z)$ is continuously distributed
with the density function $f_{\e\mid X, Z}\left(\rest{\e}x,z\right)$. 
\end{assumption}
Assumption \ref{assu:eps_pdf} is a standard assumption that simplifies the calculation of
 $\Dif_{\t}m_0 \left(Z_i,\t_0\right)$.
 

\begin{lem}
\label{lem:quantile}
Suppose that Assumption \ref{assu:eps_pdf} holds and  $f_{\rest{\e}{Z}}(\rest{0}z)>0$ for any $z\in \cal Z$, then $\E\left[\Dif_\t m_0(Z_i, \t_0)\Dif_{\t^{'}} m_0(Z_i, \t_0)\right]$ has full rank if and only if $1, Z_i$, and $\tilde{\pi}_0\left(Z_i \right):=\E\left[\rest{X_i} Z_i, \e_i=0\right]$ are not multicollinear. 
\end{lem}

In the case where the endogenous regressor $X_i$ is a scalar, the full-rank condition is equivalent to the nonlinearity of $\tilde{\pi}_0\left(Z_i \right)$. This nonlinearity condition is analogous to the nonlinearity requirement of $\pi_0\left(Z_i \right)$ in the linear regression model, except it is also conditional on $\e_i=0$. Similarly, this nonlinearity relationship naturally arises when the endogenous regressor $X_i$ is binary or discrete.

Based on the identification results, a natural two-step quantile regression estimator $\hat{\t}_{q}$ can be obtained: in the first step, nonparametrically regress $\ind\{Y_i\leq \a+Z_{i}^{'}\b +X_{i}^{'}\g\}$ on $Z_i$ for each $\t$ and compute $\hat{m}(Z_i, \t)$; in the second step, obtain the quantile estimator $\hat{\t}_q$ as
\[
\hat{\t}_{q}:=\arg\min_{\t\in \T_0} \frac{1}{n}\sum_i \left(\hat{m}(Z_i, \t)-\tau \right)^2.
\]

Writing $S_i:=f_{\e\mid Z}\left(\rest 0 Z_i\right)(1, Z_i', \tilde{\pi}_0\left(Z_i \right)' )'$, the asymptotic distribution of the two-step quantile estimator $\hat{\t}_q$ is given by
\[\sqrt{n} \left( \hat{\t}_{q}-\t_0 \right) \dto \cN\left({\bf 0},V_{0, q}\right),  \]
with
$
V_{0, q}:=\left( \E\left[S_i S_i' \right] \right)^{-1} \E\left[(\ind\{\e_i\leq 0 \}-\tau )^2  S_i S_i'  \right] \left( \E\left[S_i S_i' \right] \right)^{-1}
=\tau (1-\tau) \left( \E\left[  S_i S_i' \right] \right)^{-1}.$

In contrast to the standard quantile regression without endogeneity, our approach allows for potential endogeneity in covariate $X_i$. 
While \cite{chernozhukov2005iv} examines endogenous quantile models with excluded instruments, the key distinction is that our method establishes identification of $\t_0$ \textit{using solely included regressors}. Our approach can be viewed as leveraging the derivative term $\Dif_{\t} m_0(Z_i, \t_0)$ as an instrumental function, which is more informative than using $Z_i$ as an instrument since it exploits the dependence between the indicator term $\ind\{Y_i\leq \a_{0}+Z_{i}^{'}\b_{0}+X_{i}^{'}\g_{0}\}$ and the included regressor $Z_i$. Thus, our approach enables identification without exclusion restrictions. On the other hand, our method does require a parametric specification for $Y_i$, whereas \cite{chernozhukov2005iv} allows for nonparametric identification with excluded instruments.

\section{\label{sec:Simulation}Simulation} 
This section examines the finite sample performances of $\hat{\t}, \hat{\t}^*$, and $\hat{\t}_{disc}$, the three estimators proposed in Section \ref{sec:Estimation}. We compare their performance with both the standard 2SLS estimator $\hat{\t}_{2sls}$, which treats the included regressor $Z_i$ as an excluded instrument, and the OLS estimator $\hat{\t}_{ols}$ obtained by regressing $Y_i$ on $(1, Z_i', X_i')$. We report four finite-sample performance measures for every estimator: ``Bias'', ``SD" (standard deviation), ``RMSE" (root mean squared error), and ``CP" (coverage probability of 95\% confidence interval). The confidence intervals are constructed using the standard $\pm 1.96\times\text{SE}$ formula, where the standard error estimates SE are obtained based on the asymptotic variance estimators proposed in Section \ref{sec:Estimation}, all of which allow for heteroskedasticity.\footnote{The standard errors of the OLS estimator are also calculated under heteroskedasticity.} The four performance measures are computed based on $B=2000$ simulations, and we table the performance measures under three sample sizes: $n = 250, 500, 1000$.

\subsection{Binary $X_i$ with Two Binary $Z_{i1}, Z_{i2}$} \label{subsec: Sim1}

Our first simulation setup is as follows. In this setup, there are two binary included IVs $Z_{i1}, Z_{i2}$, randomly generated from $Bernoulli(0.5)$ independently.  The endogenous regressor $X_i$ and the outcome variable $Y_i$ are generated by
\[\begin{aligned}
X_i&=\mathbbm{1}\{2Z_{i1}Z_{i2}+2(1-Z_{i1})(1-Z_{i2})-1 \geq u_i\}, \\
Y_i&=\a_0+\b_{01}Z_{i1}+\b_{02}Z_{i2}+\g_0 X_i+\epsilon_i,
\end{aligned}
\]
where $\a_0=\g_0=1$. 
To compare with the 2SLS estimator, which treats $(Z_{i1}, Z_{i2})$ as excluded instruments, we examine different values of the coefficients of the included regressors $(Z_{i1}, Z_{i2})$: $\b_{01}=\b_{02}= \{1, 0.5, 0\}$. 
The values of the coefficients $(\b_{01}, \b_{02})$ represent the degree of violation of the exclusion restriction, and the 2SLS estimator is only consistent when 
 $\b_{01}=\b_{02}=0$. 

The two error terms $(\epsilon_i, u_i)$ are drawn, independently from $(Z_{i1}, Z_{i2})$, from the joint normal distribution with mean $(0, 0)$, variance $(1, 1)$, and correlation parameter $\rho$, which captures the extent of endogeneity between $X_i$ and $\e_i$. We also consider different levels of endogeneity $\rho\in \{0.5, 0,  -0.5\}$, with $\rho=0$ corresponding to the case with no endogeneity issue (where OLS becomes unbiased and consistent).

Since the instrument $Z_i=(Z_{i1}, Z_{i2})$ is discrete, we use the sample averages to estimate the conditional expectations:
\[
\hat{\pi}(z)=\frac{\sum_{i=1}^n X_i \ind\{Z_i=z\} } {\sum_{i=1}^n\ind\{Z_i=z\} }, \quad
\hat{h}(z)=\frac{\sum_{i=1}^n Y_i \ind\{Z_i=z\} } {\sum_{i=1}^n\ind\{Z_i=z\} }.
\]
In this case, the three estimators $\hat{\t}, \hat{\t}^{*}, \hat{\t}_{disc}$ are numerically equivalent. 


Tables \ref{table:Sim1_gamma_exc} and \ref{table:Sim1_gamma} report the performance of the five different estimators for $\g_0$ under various degrees of exclusion violations (with endogeneity $\rho=0.5$) and different levels of endogeneity (with $\b_{01}=\b_{02}=1$), respectively.\footnote{Tables \ref{table:Sim1_alpha_exc}-\ref{table:Sim1_beta2} in the Online Appendix report the performances of the estimators for the remaining coefficients $\a_0, \b_{01}$, and $\b_{02}$.}
The results demonstrate the robust performance of our estimators  in the presence of violations of the exclusion restriction and endogeneity. The root mean squared error (RMSE) of the three estimators are reasonably small, and the coverage probabilities of the 95\% confidence intervals are close to the nominal level.

In contrast, the 2SLS estimator has a very large standard deviation and bias in this simulation setup, even when the exclusion is satisfied $\b_{01}=\b_{02}=0$, due to the small determinant of the matrix $X'P_ZX$.
Additionally, the OLS estimator has a very small (close to zero) coverage probability with the presence of endogeneity.  
Our estimators' advantages become more significant as the sample size $n$ increases due to the fast reduction in both bias and standard deviation, but the 2SLS and OLS estimators remain biased regardless of sample size. Also, the $\sqrt{n}$ convergence rate of our three estimators, are strongly demonstrated by the almost exact 50\% reduction in SD and RMSE from $n=250$ to $n=1000$.

\begin{sidewaystable}[!htbp] 
\caption{Bin $X$ with Bin $Z_1, Z_2$: Performance of $\hat{\g}$ (Coef. of $X$)  \\
Different Degrees of Exclusion Violations}
 \label{table:Sim1_gamma_exc}
\begin{tabular}{c|cccc|cccc|cccc}
\hline
\hline
& \multicolumn{4}{c|}{$n=250$ } & \multicolumn{4}{c|}{$n=500$ }& \multicolumn{4}{c}{$n=1000$ }  \\
\hline
 Est& Bias & SD & RMSE & CP &Bias & SD & RMSE & CP & Bias & SD & RMSE & CP \\
\hline
&\multicolumn{12}{c}{ $\b_{01}=\b_{02}=1$}  \\
\hline
$\hat{\theta}$ &   -0.003 &  0.182 &  0.182 &  0.956&    0.002&   0.137&   0.137&   0.939  & 0.005&   0.094  & 0.094  & 0.952
  \\[1.5ex]
 $\hat{\theta}^{*}$ &  -0.003 &  0.182 &  0.182 &  0.956&    0.002&   0.137&   0.137&   0.939  & 0.005&   0.094  & 0.094  & 0.952
  \\[1.5ex]
 $\hat{\theta}_{disc}$ &  -0.003 &  0.182 &  0.182 &  0.956&    0.002&   0.137&   0.137&   0.939  & 0.005&   0.094  & 0.094  & 0.952
  \\[1.5ex]
   $\hat{\theta}_{2sls}$ &  -0.826 &  34.302 &  34.312 &  0.889 & -3.105 &  116.242 & 116.284 &  0.893 &  2.563 & 65.802 & 65.852 & 0.878
  \\[1.5ex]
  $\hat{\theta}_{ols}$ &  -0.485&   0.121  & 0.500    & 0.024 &  -0.485  & 0.088&   0.493 &  0.000 &  -0.485  & 0.062 &  0.489  & 0.000 
       \\[1.5ex]
 \hline 
&\multicolumn{12}{c}{ $\b_{01}=\b_{02}=0.5$} \\
\hline
$\hat{\theta}$ &  -0.003  & 0.182 &  0.182 & 0.956&   0.002 &  0.137 &  0.137  &0.939&   0.005  & 0.094  & 0.094&  0.952
  \\[1.5ex]
 $\hat{\theta}^{*}$ &   -0.003  & 0.182 &  0.182 & 0.956&   0.002 &  0.137 &  0.137  &0.939&0.005  & 0.094  & 0.094&  0.952
  \\[1.5ex]
 $\hat{\theta}_{disc}$ &   -0.003  & 0.182 &  0.182 & 0.956&   0.002 &  0.137 &  0.137  &0.939&0.005  & 0.094  & 0.094&  0.952
  \\[1.5ex]
     $\hat{\theta}_{2sls}$ &   -0.678 & 17.872 & 17.885 & 0.932&-1.759&  59.329 & 59.355 & 0.917& 0.984 & 32.927 & 32.942 & 0.896
  \\[1.5ex]
  $\hat{\theta}_{ols}$ &   -0.485  & 0.121 &  0.500 &  0.024& -0.485  & 0.088   &0.493 & 0.000  &  -0.485  & 0.062 &  0.489  &0.000 
     \\[1.5ex]
 \hline
 &\multicolumn{12}{c}{ $\b_{01}=\b_{02}=0$} \\
\hline
$\hat{\theta}$ &  -0.003&  0.182 & 0.182&  0.956&  0.002&   0.137&   0.137 &  0.939 & 0.005 & 0.094  &0.094  &0.952
  \\[1.5ex]
 $\hat{\theta}^{*}$ &  -0.003&  0.182 & 0.182&  0.956&  0.002&   0.137&   0.137 &  0.939& 0.005 & 0.094  &0.094  &0.952
   \\[1.5ex]
 $\hat{\theta}_{disc}$ & -0.003&  0.182 & 0.182&  0.956&  0.002&   0.137&   0.137 &  0.939& 0.005 & 0.094  &0.094  &0.952
   \\[1.5ex]
      $\hat{\theta}_{2sls}$ & -0.530 &  3.734&  3.771 & 0.991&-0.413 &  4.745&   4.763  & 0.996&  -0.594 & 3.571 & 3.620  & 0.993
  \\[1.5ex]
  $\hat{\theta}_{ols}$ &   -0.485 & 0.121&  0.500 &   0.024& -0.485&   0.088  & 0.493&   0.000  & -0.485 & 0.062 & 0.489 & 0.000
     \\[1.5ex]
 \hline
\end{tabular}
\end{sidewaystable}

\begin{sidewaystable}[!htbp] 
\caption{Bin $X$ with Bin $Z_1, Z_2$: Performance of $\hat{\g}$ (Coef. of $X$)\\ Different Degrees of Endogeneity}
 \label{table:Sim1_gamma}
\begin{tabular}{c|cccc|cccc|cccc}
\hline
\hline
& \multicolumn{4}{c|}{$n=250$ } & \multicolumn{4}{c|}{$n=500$ }& \multicolumn{4}{c}{$n=1000$ }  \\
\hline
 Est& Bias & SD & RMSE & CP &Bias & SD & RMSE & CP & Bias & SD & RMSE & CP \\
\hline
&\multicolumn{12}{c}{$\rho=0.5$}  \\
\hline
$\hat{\theta}$ &   -0.003 &  0.182 &  0.182 &  0.956&    0.002&   0.137&   0.137&   0.939  & 0.005&   0.094  & 0.094  & 0.952
  \\[1.5ex]
 $\hat{\theta}^{*}$ &  -0.003 &  0.182 &  0.182 &  0.956&    0.002&   0.137&   0.137&   0.939  & 0.005&   0.094  & 0.094  & 0.952
  \\[1.5ex]
 $\hat{\theta}_{disc}$ &  -0.003 &  0.182 &  0.182 &  0.956&    0.002&   0.137&   0.137&   0.939  & 0.005&   0.094  & 0.094  & 0.952
  \\[1.5ex]
   $\hat{\theta}_{2sls}$ &  -0.826 &  34.302 &  34.312 &  0.889 & -3.105 &  116.242 & 116.284 &  0.893 &  2.563 & 65.802 & 65.852 & 0.878
  \\[1.5ex]
  $\hat{\theta}_{ols}$ &  -0.485&   0.121  & 0.500    & 0.024 &  -0.485  & 0.088&   0.493 &  0.000 &  -0.485  & 0.062 &  0.489  & 0.000 
       \\[1.5ex]
 \hline 
&\multicolumn{12}{c}{$\rho=0$} \\
\hline
$\hat{\theta}$ &  -0.007 &  0.181&   0.181&   0.956&  -0.000   &  0.137&   0.137 &  0.938&  0.004&   0.093 &  0.093  & 0.952
  \\[1.5ex]
 $\hat{\theta}^{*}$ &     -0.007 &  0.181&   0.181&   0.956&  -0.000   &  0.137&   0.137 &  0.938&  0.004&   0.093 &  0.093  & 0.952
  \\[1.5ex]
 $\hat{\theta}_{disc}$ &   -0.007 &  0.181&   0.181&   0.956&  -0.000   &  0.137&   0.137 &  0.938&  0.004&   0.093 &  0.093  & 0.952
  \\[1.5ex]
     $\hat{\theta}_{2sls}$ &   -1.315&  32.234 & 32.261 & 0.904 &  -0.301 & 38.531  &38.532 & 0.894 & 0.362 & 74.036 & 74.036 &  0.877
  \\[1.5ex]
  $\hat{\theta}_{ols}$ &  -0.002 &  0.127  &0.127  &0.948 & -0.001&  0.091  &0.091 & 0.946 & 0.002  &0.064  &0.064  &0.942
     \\[1.5ex]
 \hline
 &\multicolumn{12}{c}{$\rho=-0.5$} \\
\hline
$\hat{\theta}$ &  -0.011 & 0.182  &0.183&  0.954&  -0.003&  0.138&  0.138 & 0.937 &0.003 & 0.093 & 0.093  &0.950
  \\[1.5ex]
 $\hat{\theta}^{*}$ &  -0.011 & 0.182  &0.183&  0.954&  -0.003&  0.138&  0.138 & 0.937 &0.003 & 0.093 & 0.093  &0.950
  \\[1.5ex]
 $\hat{\theta}_{disc}$ &   -0.011 & 0.182  &0.183&  0.954&  -0.003&  0.138&  0.138 & 0.937 &0.003 & 0.093 & 0.093  &0.950
   \\[1.5ex]
      $\hat{\theta}_{2sls}$ & 0.913 & 59.738 & 59.745 & 0.900 & 1.477&  58.951 & 58.969 & 0.882 &   0.576 &  74.167 &  74.169 &  0.872
  \\[1.5ex]
  $\hat{\theta}_{ols}$ &   0.483 &  0.124&   0.499  & 0.023  &  0.485 &  0.087 &  0.492&   0.002  & 0.485 &  0.062 &  0.489  & 0.000
     \\[1.5ex]
 \hline
\end{tabular}
\end{sidewaystable}

\subsection{Binary $X_i$ with Continuous $Z_i$} \label{subsec: Sim2}

In this subsection, we consider a different DGP in which there is a continuous IV $Z_i$ drawn from $\mathcal{N}(0, 2)$. The variables $X_i$ and $Y_i$ are generated by
\[
\begin{aligned}
X_i =\mathbbm{1}\{2Z_i \geq u_i\}, \quad
Y_i =\a_0+\b_0Z_i+\g_0 X_i+\epsilon_i,
\end{aligned}
 \]
 where $\a_0=\g_0=1$, and we consider three values for $\b_0=\{1, 0.5, 0\}$. The error terms $(u_i,\e_i)$ are again drawn from the joint normal distribution as in Section \ref{subsec: Sim1}, independently from $Z_i$, with $\rho=\{0.5, 0, -0.5\}$. 

For the two semiparametric estimators $\hat{\t}$ and $\hat{\t}^*$, we use the Nadaraya-Watson kernel estimator to nonparametrically estimate $\pi_0$ and $h_0$ in the first stage. We use the standard Gaussian kernel and set the bandwidth based on least square cross validation. We also find that the performances of the final estimators do not change much with other choices of kernels (e.g., an Epanechnikov kernel). For the discretization-based estimator $\hat{\t}_{disc}$, we partition the support of $Z_i$ into $K=10$ cells defined by the (empirical) decile ranges. Our results stay similar if $K$ is set to be larger, say, $30$.

As shown in Tables \ref{table:Sim2_gamma_exc} and \ref{table:Sim2_gamma} (and Tables \ref{table:Sim2_alpha_exc}-\ref{table:Sim2_beta} in the Online Appendix), all our three estimators perform uniformly well across different values of $\b_0$ and $\rho$. Although the two semiparametric estimators $\hat{\t}, \hat{\t}^{*}$ involve nonparametric regressions in the first stage, they have reasonably good performances even with a small sample size ($n=250$), with the corresponding CI coverage probabilities for $\g_0$ close to their nominal level 95\%. In contrast, the 2SLS and OLS estimators are significantly biased under exclusion restriction violation and endogeneity, and their CI coverage probabilities are close to zero for all the sample sizes.

 We also find that, 
the discretization-based estimator $\hat{\t}_{disc}$ performs (surprisingly) well in finite samples. While the two semiparametric estimators $\hat{\t}$ and $\hat{\t}^*$ perform well overall, their small-sample biases induced by the first-stage nonparametric regressions are fairly noticeable when compared to that of $\hat{\t}_{disc}$, especially in the estimation of $\g_0$ under $n = 250$. In contrast, the loss of asymptotic efficiency in $\hat{\t}_{disc}$ seems to be fairly small and more than compensated by its smaller finite-sample bias.


\begin{sidewaystable}[!htbp] 
\centering
\caption{Bin $X$ with Cts $Z$: Performance of $\hat{\g}$ (Coef. of $X$)  \\
Different Degrees of Exclusion Violations}
 \label{table:Sim2_gamma_exc}
\begin{tabular}{c|cccc|cccc|cccc}
\hline
\hline
& \multicolumn{4}{c|}{$n=250$ } & \multicolumn{4}{c|}{$n=500$ }& \multicolumn{4}{c}{$n=1000$ }  \\
\hline
 Est& Bias & SD & RMSE & CP &Bias & SD & RMSE & CP & Bias & SD & RMSE & CP \\
\hline
&\multicolumn{12}{c}{$\b_0=1$}  \\
\hline
$\hat{\theta}$ &   0.044 & 0.326 & 0.329 & 0.942 & 0.036  &0.220 &  0.223 & 0.954 &  0.024 &   0.155  &  0.156  &  0.948
  \\[1.5ex]
 $\hat{\theta}^{*}$ &   -0.099 & 0.306 & 0.321 & 0.942 & -0.072 & 0.209 & 0.221 & 0.948 & -0.059   & 0.148  &  0.159  &  0.942
  \\[1.5ex]
 $\hat{\theta}_{disc}$ &   -0.031 &  0.321 &  0.323 & 0.950 & -0.016 & 0.223 & 0.223 & 0.955 & -0.014   & 0.161  &  0.161 &   0.951
  \\[1.5ex]
    $\hat{\theta}_{2sls}$ &   5.165  &0.298 & 5.174 & 0.000 &  5.159  &0.200 &   5.163&  0.000 &  5.165  & 0.147&  5.167&  0.000 
  \\[1.5ex]
  $\hat{\theta}_{ols}$ &  -0.477 &  0.198 &  0.516 & 0.322 & -0.485 & 0.140  & 0.505 & 0.060 & -0.486   & 0.097  &  0.495 &   0.000 
       \\[1.5ex]
 \hline 
&\multicolumn{12}{c}{$\b_0=0.5$} \\
\hline
$\hat{\theta}$ &   0.044 & 0.326  &0.329 & 0.942&   0.036  &0.220 &  0.223  &0.954 & 0.024&  0.155 & 0.156 & 0.948
\\[1.5ex]
 $\hat{\theta}^{*}$ &   -0.150 &  0.293 & 0.329 & 0.930& -0.114 & 0.204 & 0.234  &0.936& -0.092&  0.146 & 0.173 & 0.906
  \\[1.5ex]
 $\hat{\theta}_{disc}$ &   -0.031 & 0.321&  0.323&  0.950& -0.016 & 0.223 & 0.223 & 0.955& -0.014 & 0.161 & 0.161 & 0.951
  \\[1.5ex]
      $\hat{\theta}_{2sls}$ &   2.584 & 0.206 & 2.592 & 0.000  &  2.578&  0.140   &2.582 & 0.000 &  2.582 & 0.102 & 2.584 & 0.000  
  \\[1.5ex]
  $\hat{\theta}_{ols}$ &   -0.477 & 0.198&  0.516&  0.322& -0.485& 0.140 &  0.505  &0.060 & -0.486&  0.097&  0.495 & 0.000
     \\[1.5ex]
 \hline
 &\multicolumn{12}{c}{$\b_0=0$} \\
\hline
$\hat{\theta}$ &    0.044   &0.326  & 0.329  & 0.942 & 0.036  &0.220&   0.223 & 0.954&  0.024 & 0.155 & 0.156 & 0.948
  \\[1.5ex]
 $\hat{\theta}^{*}$ & -0.287 &  0.310  &  0.422  & 0.828 &-0.220 &  0.223  &0.313&  0.802& -0.168&  0.160 &  0.232&  0.782
  \\[1.5ex]
 $\hat{\theta}_{disc}$ &  -0.031 &  0.321 &  0.323 &  0.950 & -0.016 & 0.223 & 0.223 & 0.955& -0.014&  0.161 & 0.161&  0.951
   \\[1.5ex]
       $\hat{\theta}_{2sls}$ & 0.003  & 0.164  & 0.164   &0.948 & -0.002 & 0.113  &0.113  &0.954& -0.001 & 0.081&  0.081 & 0.950 
  \\[1.5ex]
  $\hat{\theta}_{ols}$ &  -0.477   &0.198 &  0.516 &  0.322 &-0.485 & 0.140  & 0.505  &0.060 & -0.486 & 0.097 & 0.495 & 0.000  
     \\[1.5ex]
 \hline
\end{tabular}
\end{sidewaystable}

\begin{sidewaystable}[!htbp] 
\centering
\caption{Bin $X$ with Cts $Z$: Performance of $\hat{\g}$ (Coef. of $X$) \\ Different Degrees of Endogeneity}
 \label{table:Sim2_gamma}
\begin{tabular}{c|cccc|cccc|cccc}
\hline
\hline
& \multicolumn{4}{c|}{$n=250$ } & \multicolumn{4}{c|}{$n=500$ }& \multicolumn{4}{c}{$n=1000$ }  \\
\hline
 Est& Bias & SD & RMSE & CP &Bias & SD & RMSE & CP & Bias & SD & RMSE & CP \\
\hline
&\multicolumn{12}{c}{$\rho=0.5$}  \\
\hline
$\hat{\theta}$ &   0.044 & 0.326 & 0.329 & 0.942 & 0.036  &0.220 &  0.223 & 0.954 &  0.024 &   0.155  &  0.156  &  0.948
  \\[1.5ex]
 $\hat{\theta}^{*}$ &   -0.099 & 0.306 & 0.321 & 0.942 & -0.072 & 0.209 & 0.221 & 0.948 & -0.059   & 0.148  &  0.159  &  0.942
  \\[1.5ex]
 $\hat{\theta}_{disc}$ &   -0.031 &  0.321 &  0.323 & 0.950 & -0.016 & 0.223 & 0.223 & 0.955 & -0.014   & 0.161  &  0.161 &   0.951
  \\[1.5ex]
    $\hat{\theta}_{2sls}$ &   5.165  &0.298 & 5.174 & 0.000 &  5.159  &0.200 &   5.163&  0.000 &  5.165  & 0.147&  5.167&  0.000 
  \\[1.5ex]
  $\hat{\theta}_{ols}$ &  -0.477 &  0.198 &  0.516 & 0.322 & -0.485 & 0.140  & 0.505 & 0.060 & -0.486   & 0.097  &  0.495 &   0.000 
       \\[1.5ex]
 \hline 
&\multicolumn{12}{c}{$\rho=0$} \\
\hline
$\hat{\theta}$ &   0.074  &0.325 & 0.334 & 0.933 & 0.056 & 0.218 & 0.225  &0.948&  0.035 & 0.153 & 0.157 & 0.943
  \\[1.5ex]
 $\hat{\theta}^{*}$ &    -0.093 & 0.305 & 0.319 & 0.946 & -0.065  &0.207 & 0.218 & 0.953 & -0.056 & 0.147 & 0.158  &0.943
  \\[1.5ex]
 $\hat{\theta}_{disc}$ &   0.001 & 0.325&  0.325 & 0.953 &  0.002 & 0.222 & 0.222 & 0.958& -0.005 & 0.161&  0.161 & 0.954
  \\[1.5ex]
      $\hat{\theta}_{2sls}$ &  5.165  & 0.298  & 5.173 &  0.000  &   5.158 & 0.198&  5.161 & 0.000  & 5.165 & 0.147 & 5.167 & 0.000  
  \\[1.5ex]
  $\hat{\theta}_{ols}$ &  -0.004 & 0.204  &0.204 & 0.938 & -0.003 & 0.146 & 0.146 & 0.940 &-0.003 & 0.101 & 0.101&  0.948
     \\[1.5ex]
 \hline
 &\multicolumn{12}{c}{$\rho=-0.5$} \\
\hline
$\hat{\theta}$ &   0.108&  0.323 & 0.340   &0.917 &  0.074 & 0.217 & 0.230&   0.937& 0.046  &0.153 & 0.160  & 0.936
  \\[1.5ex]
 $\hat{\theta}^{*}$ &  -0.081&  0.302 & 0.312 & 0.955 & -0.060 &  0.206 & 0.214 & 0.959 & -0.054 & 0.147 & 0.156 & 0.947
  \\[1.5ex]
 $\hat{\theta}_{disc}$ &  0.038 & 0.321&  0.324 & 0.952 &  0.019 & 0.222 & 0.223 & 0.960 &   0.004 & 0.161 & 0.161 & 0.952
   \\[1.5ex]
       $\hat{\theta}_{2sls}$ &  5.163 & 0.294 & 5.172  &0.000 & 5.157 & 0.196 & 5.161&  0.000  &5.165  & 0.147 & 5.167  & 0.000 
  \\[1.5ex]
  $\hat{\theta}_{ols}$ &   0.481 & 0.197 & 0.519 & 0.298 &  0.477 & 0.139 & 0.497 & 0.068 &  0.479 & 0.098 & 0.489 & 0.002
     \\[1.5ex]
 \hline
\end{tabular}
\end{sidewaystable}

\subsection{Continuous $X_i$ with Continuous IV $Z_i$} \label{subsec:Sim3}

In this subsection, we consider another DGP setup, where the endogenous covariate $X_i$ is a continuous random variable generated as
$$X_i = \cos(Z_i) + \sqrt{0.5\left|Z_i\right| + 0.5} \cd u_i,\quad \text{ with } Z_i \sim U\left[-\pi,\pi\right],$$
where $u_i$ and $\e_i$ are jointly normal as before. As before, $(\a_0, \g_0) = (1,1)$ and we study three values of $\b_0$: $\b_0=\{1, 0.5, 0\}$. To also illustrate the point that our method works well under different choices of the first-stage nonparametric estimation methods, here we run cubic spline regressions with cross-validated choices of degrees of freedom.\footnote{Our results do not change substantially if the Nadaraya-Watson estimator is used instead.}



Tables \ref{table:Sim3_gamma_exc} and \ref{table:Sim3_gamma} (along with Tables \ref{table:Sim3_alpha_exc}-\ref{table:Sim3_beta} in the Online Appendix) show that our estimators continue to perform well under different nonparametric regression methods and DGP designs. It is worth noting that the 2SLS estimator yields very large standard errors (even when the exclusion restriction is satisfied). This is because, even though $Z_i$ is by construction relevant for $X_i$ (in a nonlinear manner), $Z_i$ is only ``weak IV" for $X_i$ via linear projection. In contrast, our three proposed estimators are able to capture the nonlinear relevance of $Z_i$, and deliver small standard errors across all the simulation configurations.

\begin{sidewaystable}[!htbp] 
\centering
\caption{Cts $X$ with Cts $Z$: Performance of $\hat{\g}$ (Coef. of $X$)  \\
Different Degrees of Exclusion Violations}
 \label{table:Sim3_gamma_exc}
\begin{tabular}{c|cccc|cccc|cccc}
\hline
\hline
& \multicolumn{4}{c|}{$n=250$ } & \multicolumn{4}{c|}{$n=500$ }& \multicolumn{4}{c}{$n=1000$ }  \\
\hline
 Est& Bias & SD & RMSE & CP &Bias & SD & RMSE & CP & Bias & SD & RMSE & CP \\
\hline
&\multicolumn{12}{c}{$\b_0=1$}  \\
\hline
$\hat{\theta}$ & 0.061 & 0.094 & 0.112 & 0.871 & 0.035 & 0.064 & 0.072 & 0.901 & 0.020 & 0.045 & 0.049 & 0.911 \\[1.5ex] 
$\hat{\theta}^{*}$ & -0.044 & 0.112 & 0.120 & 0.909 &   -0.029 & 0.073 & 0.079 & 0.922 & -0.018 & 0.049 & 0.052 & 0.925\\[1.5ex] 
$\hat{\theta}_{disc}$ & 0.024 & 0.089 & 0.092 & 0.932 & 0.013 & 0.062 & 0.064 & 0.947 & 0.005 & 0.045 & 0.045 & 0.936 \\[1.5ex] 
$\hat{\theta}_{2sls}$ & 39.48 & 1367 & 1368 & 0.946 & 35.02 & 1609 & 1609 & 0.943 & -21.83 & 2032 & 2032 & 0.956\\[1.5ex]
$\hat{\theta}_{ols}$ & 0.311 & 0.042 & 0.314 & 0.000 & 0.312 & 0.030 & 0.313 & 0.000 & 0.312 & 0.022 & 0.313 & 0.000\\[1.5ex]

 \hline 
&\multicolumn{12}{c}{$\b_0=0.5$} \\
\hline
$\hat{\theta}$ & 0.061 & 0.094 & 0.112 & 0.871 &  0.035 & 0.064 & 0.072 & 0.901 & 0.020 & 0.045 & 0.049 & 0.911
\\[1.5ex]
 $\hat{\theta}^{*}$ & -0.044 & 0.111 & 0.120 & 0.909 & -0.029 & 0.073 & 0.079 & 0.922 & -0.018 & 0.049 & 0.052 & 0.925
  \\[1.5ex]
 $\hat{\theta}_{disc}$ & 0.024 & 0.089 & 0.092 & 0.932 & 0.013 & 0.062 & 0.064 & 0.947 & 0.005 & 0.045 & 0.045 & 0.936
  \\[1.5ex]
      $\hat{\theta}_{2sls}$ & 19.50 & 670.8 & 670.9 & 0.945 &  17.06 & 785.3 & 785.3 & 0.943 & -10.21 & 982.7 & 982.5 & 0.957
  \\[1.5ex]
  $\hat{\theta}_{ols}$ & 0.311 & 0.042 & 0.314 & 0.000 & 0.312 & 0.030 & 0.313 & 0.000 & 0.312 & 0.022 & 0.313 & 0.000
     \\[1.5ex]
 \hline
 &\multicolumn{12}{c}{$\b_0=0$} \\
\hline
$\hat{\theta}$ & 0.061 & 0.094 & 0.112 & 0.871 & 0.035 & 0.064 & 0.072 & 0.901 & 0.020 & 0.045 & 0.049 & 0.911
  \\[1.5ex]
 $\hat{\theta}^{*}$ & -0.044 & 0.112 & 0.120 & 0.909 & -0.029 & 0.073 & 0.079 & 0.922 & -0.018 & 0.049 & 0.052 & 0.925
  \\[1.5ex]
 $\hat{\theta}_{disc}$ & 0.024 & 0.089 & 0.092 & 0.932 & 0.013 & 0.062 & 0.064 & 0.947 & 0.005 & 0.045 & 0.045 & 0.936
   \\[1.5ex]
       $\hat{\theta}_{2sls}$ & -0.487 & 49.20 & 49.19 & 0.990 & -0.899 & 54.68 & 54.67 & 0.987 & 1.403 & 72.53 & 72.52 & 0.996
  \\[1.5ex]
  $\hat{\theta}_{ols}$ & 0.311 & 0.042 & 0.314 & 0.000 & 0.312 & 0.030 & 0.313 & 0.000 & 0.312 & 0.022 & 0.313 & 0.000
     \\[1.5ex]
 \hline
\end{tabular}
\end{sidewaystable}

\begin{sidewaystable}[!htbp] 
\centering
\caption{Cts $X$ with Cts $Z$: Performance of $\hat{\g}$ (Coef. of $X$) \\ Different Degrees of Endogeneity}
 \label{table:Sim3_gamma}
\begin{tabular}{c|cccc|cccc|cccc}
\hline
\hline
& \multicolumn{4}{c|}{$n=250$ } & \multicolumn{4}{c|}{$n=500$ }& \multicolumn{4}{c}{$n=1000$ }  \\
\hline
 Est& Bias & SD & RMSE & CP &Bias & SD & RMSE & CP & Bias & SD & RMSE & CP \\
\hline
&\multicolumn{12}{c}{$\rho=0.5$}  \\
\hline
$\hat{\theta}$ & 0.061 & 0.094 & 0.112 & 0.871 & 0.035 & 0.064 & 0.072 & 0.901 & 0.020 & 0.045 & 0.049 & 0.911 \\[1.5ex] 
$\hat{\theta}^{*}$ & -0.044 & 0.112 & 0.120 & 0.909 &   -0.029 & 0.073 & 0.079 & 0.922 & -0.018 & 0.049 & 0.052 & 0.925\\[1.5ex] 
$\hat{\theta}_{disc}$ & 0.024 & 0.089 & 0.092 & 0.932 & 0.013 & 0.062 & 0.064 & 0.947 & 0.005 & 0.045 & 0.045 & 0.936 \\[1.5ex] 
$\hat{\theta}_{2sls}$ & 39.48 & 1367 & 1368 & 0.946 & 35.02 & 1609 & 1609 & 0.943 & -21.83 & 2032 & 2032 & 0.956\\[1.5ex]
$\hat{\theta}_{ols}$ & 0.311 & 0.042 & 0.314 & 0.000 & 0.312 & 0.030 & 0.313 & 0.000 & 0.312 & 0.022 & 0.313 & 0.000\\[1.5ex]
 \hline 
&\multicolumn{12}{c}{$\rho=0$} \\
\hline
$\hat{\theta}$ & 0.052 & 0.099 & 0.112 & 0.913 &0.030 & 0.068 & 0.074 & 0.918 & 0.016 & 0.047 & 0.050 & 0.928
  \\[1.5ex]
 $\hat{\theta}^{*}$ & -0.025 & 0.105 & 0.108 & 0.916 & -0.017 & 0.072 & 0.074 & 0.916 & -0.011 & 0.048 & 0.049 & 0.929
  \\[1.5ex]
 $\hat{\theta}_{disc}$ & -0.002 & 0.090 & 0.090 & 0.950 & -0.000 & 0.065 & 0.065 & 0.942 & -0.002 & 0.046 & 0.046 & 0.945
  \\[1.5ex]
 $\hat{\theta}_{2sls}$ & -17.63 & 1306 & 1306 & 0.935 & -74.79 & 2165 & 2166 & 0.946 & -37.09 & 882.2 & 882.8 & 0.951\\[1.5ex]
  $\hat{\theta}_{ols}$ & -0.001 & 0.047 & 0.047 & 0.942 & 0.001 & 0.034 & 0.034 & 0.940 & -0.001 & 0.024 & 0.024 & 0.947\\[1.5ex]
 \hline
 &\multicolumn{12}{c}{$\rho=-0.5$} \\
\hline
$\hat{\theta}$ & 0.044 & 0.104 & 0.113 & 0.936 &0.026 & 0.070 & 0.074 & 0.934 & 0.014 & 0.047 & 0.049 & 0.941
  \\[1.5ex]
 $\hat{\theta}^{*}$ & -0.003 & 0.106 & 0.106 & 0.918 &-0.001 & 0.070 & 0.070 & 0.918 &-0.002 & 0.048 & 0.048 & 0.937
  \\[1.5ex]
 $\hat{\theta}_{disc}$ & -0.026 & 0.090 & 0.094 & 0.919 &-0.013 & 0.065 & 0.066 & 0.926 &-0.007 & 0.046 & 0.046 & 0.942
   \\[1.5ex]
 $\hat{\theta}_{2sls}$ & -14.81 & 597.6 & 597.6 & 0.946 &35.31 & 1988 & 1988 & 0.946 & -68.94 & 1710 & 1711 & 0.952
     \\[1.5ex]
 $\hat{\theta}_{ols}$ & -0.313 & 0.043 & 0.315 & 0.000 &-0.311 & 0.031 & 0.313 & 0.000 & -0.313 & 0.021 & 0.314 & 0.000
     \\[1.5ex]
 \hline
\end{tabular}
\end{sidewaystable}


\section{Empirical Applications}\label{sec:Application}

We apply our methodology to examine the returns to education, a topic of substantial attention in the literature (see, e.g., \cite{card2001estimating} for a review of various studies on this topic). A key concern in investigating the causal impact of education is its potential endogeneity, and it is challenging to find valid instruments that are excluded from the model. Our approach allows us to include all potential instruments in the regression model and test their validity. 

We conduct two applications and test the direct effects of different instruments for education. The first application explores the college proximity indicators, which are proposed in \cite{card1993using}. We find that after controlling for regional characteristics, the presence of a nearby college does not significantly affect income. However, the 2SLS estimator varies substantially with different instruments and can become insignificant, while our estimators remain more robust regardless of the choice of the instruments. 
In the second application, we examine the validity of family background variables as instruments. Our findings show that the number of siblings has no significant effect on wages rates, while parents' education significantly increases wages. 

\subsection{Application I: College Proximity Indicators}

We use the same dataset as in \cite{card1993using}, drawn from the National Longitudinal Survey of Young Men (NLSYM). This data contains information of $n=3010$ male observations in 1976, documenting their educational attainment, wage, race, age, and assorted demographic characteristics. \cite{card1993using} proposes to use the presence of a 4-year college  as an instrument for education, which is likely to affect an individual's educational attainment but may not have direct effects on their earnings. However, \cite{card1993using} also raises a potential concern with this instrument, as the presence of a college might be correlated with superior school quality and, consequently, could lead to higher earnings. We study two specifications that investigate one of the two indicators of college proximity respectively: the presence of a nearby 2-year college (nearc2) and the presence of a nearby 4-year college (nearc4).

Following \cite{card1993using}, the dependent variable is the log of hourly wage in 1976, the endogenous variable is education, and the control variables include experience, experience squared, a black indicator, indicators for southern residence and residence in an SMSA in 1976, and indicators for region in 1966 and living in an SMSA in 1966. Distinct from \cite{card1993using}, our approach also includes the college proximity instrument in the model and allows for testing its direct effect on the outcome by examining the significance of the coefficient.
Table \ref{table:app1_sum} presents the summary statistics of the primary variables.
\begin{table}[!htbp] 
\centering
\renewcommand{\arraystretch}{1.3} 
\caption{Summary Statistics with College Proximity Indicators}
\label{table:app1_sum}
\centering
\begin{tabular}{c|c@{\hspace{1cm}}c@{\hspace{1cm}}c@{\hspace{1cm}}c@{\hspace{1cm}}}
\hline
\hline
 & mean & s.d.  & minimum & maximum  \\
\hline
log(wage) in 1976 &     6.262  &  0.444  & 4.605   &  7.785 \\
education &13.263  &  2.677 &  1.000  &  18.000  \\
experience &  8.856  &  4.142  & 0.000  &  23.000 \\
experience squared &  95.579 &  84.618  & 0.000  & 529.000\\
black &  0.234  &  0.423  & 0.000  &   1.000 \\
nearc2 &0.441 &   0.497 &  0.000  &   1.000 \\
nearc4 &  0.682  &  0.466 &  0.000   &  1.000\\
live in SMSA in 1966 &0.650  & 0.477 & 0.000 &   1.000 \\
live in SMSA in 1976 &0.713  &  0.452  &  0.000  &   1.000  \\
live in South in 1976  &0.404  &  0.491   &  0.000  &   1.000 \\
\hline
\end{tabular}
\\[5pt]
\footnotesize{Notes: the experience variable is constructed using the conventional meansure: experience=age-education-6.}
\end{table}

We present the results of six different estimators. The first two estimators $\hat{\t}, \hat{\t}^{*}$ are introduced in Section \ref{subsec:Est_Semipara}.
For the estimation of $\hat{\pi}(Z_i), \hat{h}(Z_i)$, we employ the 
Support Vector Machine (SVM) method, a broadly applied machine learning technique for high-dimensional regressors.  The neural network approach is also implemented, yielding similar results and the same significance of all coefficients.\footnote{We adopt the function `svm' from the e1071 package and the function `neuralnet' from the neuralnet package in R to implement the SVM approach and the neural network method.}
For the discretization estimator $\hat{\t}_{disc}$, we divide the experience variable into three partitions using empirical quantiles and generate dummy variables for each partition. Then we construct the instrument for education using the product of any two indicator variables. 
The estimator $\hat{\t}_{ols}$ is the OLS estimator that includes the instrument in the regression, while $\tilde{\t}_{ols}$ represents the OLS estimator that does not include the instrument. The last one $\hat{\t}_{2sls}$ is the 2SLS estimator using the college proximity indicator as the excluded instrument for education.

Table \ref{table:app1_svm} and Table \ref{table:app1_nn} display the outcomes of the coefficients for education and the college proximity instruments using SVM and neural network methods. The results from our three estimators $\hat{\t}, \hat{\t}^{*}, \hat{\t}_{disc}$ demonstrate that, after controlling for all the regional factors in 1966 and 1976,  having a nearby 2-year college or 4-year college has no significant effects on wages. This finding supports the validity of using college proximity indicators as excluded instruments.  
The standard deviations of the instrument coefficients with $\hat{\t},\hat{\t}^{*}, \hat{\t}_{disc}$ are very close to that of  $\hat{\t}_{ols}$, reinforcing the good performance of these  estimators.

The estimated returns to education from $\hat{\t}, \hat{\t}^{*}, \hat{\t}_{disc}$ are uniformly positive and significant under various specifications and nonparametric estimation methods. Moreover, the standard deviations of the education coefficients, derived from the three estimators, are also reasonably small across different specifications and are smaller than the one obtained from the 2SLS estimator. The coefficients on education from $\hat{\t}, \hat{\t}_{disc}$ are all higher than those from the two OLS estimators, suggesting that the OLS estimators may underestimate education's impact. The coefficient from $\hat{\t}^{*}$ can be lower than $\hat{\t}_{ols}$, as it uses a different dependent variable. Overall, the three estimators $\hat{\t}, \hat{\t}^{*}, \hat{\t}_{disc}$ yield very similar results, which corroborates our theoretical results on their asymptotic variances in Section \ref{sec:Estimation}. 

For the 2SLS estimator $\hat{\t}_{2sls}$, the coefficients on education vary substantially when using the two different instruments. It becomes insignificant when using the presence of a nearby 2-year college as an instrument, due to the large standard deviation. In addition, the estimated coefficients on education from $\hat{\t}, \hat{\t}^{*}, \hat{\t}_{disc}$ are uniformly smaller than the one obtained from the 2SLS estimator, since our three estimators allow for the direct effect of the instrument. 

\begin{table}[!htbp]
\centering
\caption{Returns to Education: College Proximity Indicators (SVM)}
\label{table:app1_svm}
\begin{tabular}{c|c@{\hspace{1cm}}c@{\hspace{1cm}}|c@{\hspace{1cm}}c@{\hspace{1cm}}}
\hline
\hline
& education & nearc2  & education & nearc4  \\
\hline
& \multicolumn{2}{c}{nearc2} & \multicolumn{2}{c}{nearc4} \\
\hline
\multirow{2}{*}{$\hat{\t}$} & 0.083** &  0.028 & 0.078** & 0.019  \\
& (0.012) & (0.015) &  (0.012) & (0.017)  \\[1.5ex]
\multirow{2}{*}{$\hat{\t}^{*} $ }& 0.067** & 0.024 & 0.065** & 0.015 \\
& (0.012) & (0.015) & (0.012) & (0.017)  \\[1.5ex]
\multirow{2}{*}{$\hat{\t}_{disc}$} & 0.092** & 0.029 & 0.082** & 0.019   \\
& (0.014) & (0.015) & (0.013)& (0.017) \\[1.5ex]
\multirow{2}{*}{$\hat{\t}_{ols}$} & 0.075** & 0.027 &0.074**  &0.018   \\
& (0.004) & (0.015) & (0.004) & (0.017)\\[1.5ex]
\multirow{2}{*}{ $\tilde{\t}_{ols}$ } &  0.075** & \multirow{2}{*}{-} &  0.075** & \multirow{2}{*}{-}   \\
& (0.004) &  & (0.004) &  \\[1.5ex]
\multirow{2}{*}{ $\hat{\t}_{2sls}$} & 0.293 & \multirow{2}{*}{-} &0.132** & \multirow{2}{*}{-}    \\
& (0.186) & & (0.054) &  \\
\hline
\hline
\end{tabular}
\\[5pt]
\footnotesize{Notes: the symbol ** denotes significant coefficients from zero at 95\% level.}
\end{table}

\begin{table}[!htbp]
\centering
\caption{Returns to Education: College Proximity Indicators (Neural Network)}
\label{table:app1_nn}
\begin{tabular}{c|c@{\hspace{1cm}}c@{\hspace{1cm}}|c@{\hspace{1cm}}c@{\hspace{1cm}}}
\hline
\hline
& education & nearc2  & education & nearc4  \\
\hline
& \multicolumn{2}{c}{nearc2} & \multicolumn{2}{c}{nearc4} \\
\hline
\multirow{2}{*}{$\hat{\t}$} & 0.087** &  0.024 &  0.099** &0.012  \\
& (0.018) & (0.015) &  (0.017) & (0.017)  \\[1.5ex]
\multirow{2}{*}{$\hat{\t}^{*} $ }& 0.081** & 0.030 & 0.099** & 0.012  \\
& (0.018) & (0.015) & (0.017) & (0.018)  \\[1.5ex]
\multirow{2}{*}{$\hat{\t}_{disc}$} & 0.092** & 0.029 & 0.082** & 0.019   \\
& (0.014) & (0.015) & (0.013)& (0.017) \\[1.5ex]
\multirow{2}{*}{$\hat{\t}_{ols}$} & 0.075** & 0.027 &0.074**  &0.018   \\
& (0.004) & (0.015) & (0.004) & (0.017)\\[1.5ex]
\multirow{2}{*}{ $\tilde{\t}_{ols}$ } &  0.075** & \multirow{2}{*}{-} &  0.075** & \multirow{2}{*}{-}   \\
& (0.004) &  & (0.004) &  \\[1.5ex]
\multirow{2}{*}{ $\hat{\t}_{2sls}$} & 0.293 & \multirow{2}{*}{-} &0.132** & \multirow{2}{*}{-}    \\
& (0.186) & & (0.054) &  \\
\hline
\hline
\end{tabular}
\\[5pt]
\footnotesize{Notes: the symbol ** denotes significant coefficients from zero at 95\% level.}
\end{table}

\subsection{Application II: Family Background Variables}

As data about nearby colleges might not always be available, family background variables are often used as instruments for education. 
In this application, we explore two family background variables as instruments: parents' average education and number of siblings. To conduct this study, we utilize the dataset `NLSY79', which conducts interviews of $n = 10800$ young individuals, both male and female, ranging in age from 14 to 21 in 1979. This survey records various characteristics of the individuals, such as gender, marriage status, work-related factors, region indicators, as well as family background variables. 
Table \ref{table:app2_sum} displays the summary statistics of the key variables.

\begin{table}[!htbp] 
\centering
\renewcommand{\arraystretch}{1.3} 
\caption{Summary Statistics with Family Background Variables}
\label{table:app2_sum}
\centering
\begin{tabular}{c|c@{\hspace{1cm}}c@{\hspace{1cm}}c@{\hspace{1cm}}c@{\hspace{1cm}}}
\hline
\hline
 & mean & s.d.  & minimum & maximum  \\
\hline
log(wage) &  2.780 &0.604 & 0.756 & 5.284    \\
education & 13.678 &2.475  &0.000 &20.000  \\
female &   0.500 &0.500 & 0.000 & 1.000  \\
black &  0.100 &0.300 & 0.000  &1.000 \\
marriage  & 0.652 &0.476  &0.000  &1.000 \\
experience &  16.977 &4.373&  0.827 &23.808\\
hour &   40.831& 8.925& 10.000& 60.000 \\
live in  North-Central&  0.325 &0.468 & 0.000 & 1.000 \\
live in  North-Eastern &   0.162 &0.368  &0.000  &1.000 \\
live in Southern & 0.360 &0.480  &0.000 & 1.000  \\
parents' average education & 11.703& 2.738 & 0.000 &20.000  \\
number of siblings  &3.165 &2.139  &0.000 &17.000    \\
\hline
\end{tabular}
\\[5pt]
\footnotesize{Notes: parents' average education is computed by (mother's education+father's education)/2.}
\end{table}

We compare our three estimators $\hat{\t}, \hat{\t}^{*}, \hat{\t}_{disc}$ with two OLS estimators and three 2SLS estimators. We still apply both the SVM and the neural network methods to estimate $\hat{\pi}(Z_i)$ and $ \hat{h}(Z_i)$. For the discretization estimator, we construct dummy variables for experience, hour, parents' average education, and number of siblings, based on whether each variable is above the median. The included instruments for education are then constructed using the product of any two variables. 
For the two OLS estimators, $\hat{\t}_{ols}$ includes the two family background variables, while $\tilde{\t}_{ols}$ does not. Additionally, we evaluate three 2SLS estimators. The first one $\hat{\t}_{2sls}^{both}$ uses both parents' education and number of siblings as excluded instruments for education. 
The second estimator $\hat{\t}_{2sls}^{edu}$ employs only parents' education as an instrument, while the third one $\hat{\t}_{2sls}^{sib}$ utilizes solely the number of siblings.

Table \ref{table:app2_svm} and Table \ref{table:app2_nn} display the results of the eight estimators. The findings from the three estimators $\hat{\t}, \hat{\t}^{*}, \hat{\t}_{disc}$ show that the number of siblings does not have significant effects on wages, whereas parents' education significantly increases income. 
This result is consistent with our intuitive reasoning, as parents' education could influence an individual's wage by creating a more favorable educational environment.
The three 2SLS estimators appear to overestimate the returns to education, especially the two estimators  $\hat{\t}_{2sls}^{both}, \hat{\t}_{2sls}^{edu}$ which involve using parents' education as instruments. The three estimators  $\hat{\t}, \hat{\t}^{*}, \hat{\t}_{disc}$ all have significantly positive coefficients on education, and their results are smaller than those of the three 2SLS estimators, given that they control for direct effects of parents' education.

\begin{table}[!htbp] 
\centering
\caption{Returns to Education: Family Background Variables (SVM)}
\label{table:app2_svm}
\centering
\begin{tabular}{c|c@{\hspace{1cm}}c@{\hspace{1cm}}c@{\hspace{1cm}} }
\hline
\hline
 &education & parents' education  &  number of siblings \\
 \hline
 \hline
\multirow{2}{*}{  $\hat{\t}$ } &   0.116** &  0.014** &  0.001  \\
&  (0.004) & (0.002) & (0.002)  \\[1.5ex]
\multirow{2}{*}{$ \hat{\t}^{*} $} & 0.101** & 0.019** &  -0.001  \\
& (0.004) & (0.002) & (0.002) \\[1.5ex]
\multirow{2}{*}{ $\hat{\t}_{disc}$} & 0.109** & 0.019** & -0.004   \\
&  (0.012) & (0.008) &  (0.005)   \\[1.5ex]
\multirow{2}{*}{ $\hat{\t}_{ols}$} & 0.103**  & 0.019** & 0.002   \\
&  (0.002) & (0.002) &(0.002)  \\[1.5ex]
\multirow{2}{*}{ $\tilde{\t}_{ols}$} &  0.113** &  \multirow{2}{*}{-} & \multirow{2}{*}{-}   \\
&  (0.002) &   \\[1.5ex]
\multirow{2}{*}{ $\hat{\t}_{2sls}^{both}$ } & 0.147** &\multirow{2}{*}{-} &\multirow{2}{*}{-}    \\
&  (0.005) &  \\[1.5ex]
\multirow{2}{*}{ $\hat{\t}_{2sls}^{edu}$ } & 0.149** &\multirow{2}{*}{-} &\multirow{2}{*}{-}    \\
&  (0.005) &  \\[1.5ex]
\multirow{2}{*}{ $\hat{\t}_{2sls}^{sib}$ } & 0.123** &\multirow{2}{*}{-} &\multirow{2}{*}{-}    \\
&  (0.009) &  \\
\hline
\hline
\end{tabular}
\\[5pt]
\footnotesize{Notes: the symbol ** denotes significant coefficients from zero at 95\% level.}
\end{table}

\begin{table}[!htbp] 
\centering
\caption{Returns to Education: Family Background Variables  (Neural Network)}
\label{table:app2_nn}
\centering
\begin{tabular}{c|c@{\hspace{1cm}}c@{\hspace{1cm}}c@{\hspace{1cm}} }
\hline
\hline
 &education & parents' education  &  number of siblings \\
 \hline
 \hline
\multirow{2}{*}{  $\hat{\t}$ } &  0.109** & 0.017**  & 0.004     \\
&  (0.006) & (0.003) & (0.002)  \\[1.5ex]
\multirow{2}{*}{$ \hat{\t}^{*} $} & 0.081** &    0.029**  &  -0.001   \\
&   (0.006)  & (0.003) & (0.002) \\[1.5ex]
\multirow{2}{*}{ $\hat{\t}_{disc}$} & 0.109** & 0.019** & -0.004   \\
&  (0.012) & (0.008) &  (0.005)   \\[1.5ex]
\multirow{2}{*}{ $\hat{\t}_{ols}$} & 0.103**  & 0.019** & 0.002   \\
&  (0.002) & (0.002) &(0.002)  \\[1.5ex]
\multirow{2}{*}{ $\tilde{\t}_{ols}$} &  0.113** &  \multirow{2}{*}{-} & \multirow{2}{*}{-}   \\
&  (0.002) &   \\[1.5ex]
\multirow{2}{*}{ $\hat{\t}_{2sls}^{both}$ } & 0.147** &\multirow{2}{*}{-} &\multirow{2}{*}{-}    \\
&  (0.005) &  \\[1.5ex]
\multirow{2}{*}{ $\hat{\t}_{2sls}^{edu}$ } & 0.149** &\multirow{2}{*}{-} &\multirow{2}{*}{-}    \\
&  (0.005) &  \\[1.5ex]
\multirow{2}{*}{ $\hat{\t}_{2sls}^{sib}$ } & 0.123** &\multirow{2}{*}{-} &\multirow{2}{*}{-}    \\
&  (0.009) &  \\
\hline
\hline
\end{tabular}
\\[5pt]
\footnotesize{Notes: the symbol ** denotes significant coefficients from zero at 95\% level.}
\end{table}

\section{\label{sec:Conclusion}Conclusion}

This paper offers an alternative approach to achieve point identification of endogenous regression models in the absence of excluded instruments. The key idea of this approach is to leverage the nonlinear dependence between the included exogenous regressor and the endogenous variable. For estimation, we introduce two semiparametric estimators and a easy-to-compute discretization-based estimator. The asymptotic properties of all three estimators are derived and their robust finite sample performances are demonstrated through Monte Carlo simulations. We apply the approach to study returns to education, and to test the direct effects of college proximity indicators as well as family background variables.

\bibliographystyle{ecta}
\bibliography{IncIV}

\appendix

\section{\label{sec:App_Proof}Proofs}

\subsection{Proof of Lemma \ref{lem:d_values}}
\begin{proof}
We first prove that Condition \ref{cond:rd_supp} $\Longrightarrow$ Assumption \ref{assu:FullRankW} by contradiction.
Suppose
that Assumption \ref{assu:FullRankW} fails. Then there exists $c\in\R^{d}\backslash\left\{ {\bf 0}\right\} $
s.t. $\left(1,\pi_{0}^{'}\left(z\right),z^{'}\right)c=0,\ \forall z\in{\cal Z}.$
For any distinct $z_{1},...,z_{d}\in{\cal Z}$, define 
$$A(z_1,...,z_d) :=\left(\begin{array}{ccc}
1 & z_{1}^{'} & \pi_{0}\left(z_{1}\right)^{'}\\
1 & z_{2}^{'} & \pi_{0}\left(z_{2}\right)^{'}\\
\vdots & \vdots & \vdots\\
1 & z_{d}^{'} & \pi_{0}\left(z_{d}\right)^{'}
\end{array}\right),\quad r(z_1,...,z_d) := \text{rank}\left(A(z_1,...,z_d)\right).$$
We
have
$A(z_1,...,z_d) c={\bf 0}\imp r(z_1,...,z_d) <d$.

We now prove that Assumption \ref{assu:FullRankW} $\Longrightarrow$ Condition \ref{cond:rd_supp}.
Suppose that Assumption \ref{assu:FullRankW}
holds, i.e., $\left(1,Z_{i}^{'},\pi_{0}\left(Z_{i}\right)^{'}\right)$
are not multicollinear. This means that for any $c\in\R^{d}\backslash\left\{ {\bf 0}\right\} $,
there must exist some $z\in{\cal Z}$ s.t. 
\begin{equation}
\left(1,z^{'},,\pi_{0}\left(z\right)^{'}\right)c\neq0.\label{eq:NoMul_cz}
\end{equation}
If $\#\left({\cal Z}\right)=K<d$, then \eqref{eq:NoMul_cz} cannot true. Hence $\#\left({\cal Z}\right)\geq d$. For any $d$ distinct points
$z_{1},...,z_{d}$, if $r(z_1,...,z_d) = d$, then we are done. If $r(z_1,...,z_d)<d$,
then there exists some $c\in\R^{d}\backslash\left\{ {\bf 0}\right\} $
s.t. $A(z_1,...,z_d)c={\bf 0}$.
Now, by \eqref{eq:NoMul_cz} there must exists $z_{d+1}\in{\cal Z}$
s.t. 
$\left(1,z_{d+1}^{'},\pi_{0}\left(z_{d+1}\right)^{'}\right)c\neq0$,
which implies that $\left(1,z_{d+1}^{'},\pi_{0}\left(z_{d+1}\right)^{'}\right)$
is linearly independent from 
$\left\{ \left(1,z_{k}^{'},\pi_{0}\left(z_{k}\right)^{'}\right):  k=1,..,d\right\}$
and thus
$r(z_1,...,z_d,z_{d+1})=r(z_1,...,z_d)+1$.
If $r(z_1,...,z_d,z_{d+1})=d,$ we stop; otherwise we can repeat the argument above and find some $z_{d+2}\in{\cal Z}$
such that $r\left(z_{1},...,z_{d+2}\right)=r\left(z_{1},...,z_{d+1}\right)+1$. This recursion must stop at most $k^{*}\leq d-r\left(z_{1},...,z_{d}\right)$
steps, with
$r\left(z_{1},...,z_{d+k^{*}}\right)=d$.
Then we pick $d$ distinct points from $\{z_1,...,z_{d+k^{*}}\}$ such that its rank is $d$,
which is precisely Condition \ref{cond:rd_supp}.
\end{proof}

\subsection{\label{subsec:Notation}Notation for Asymptotic Theory}

We first formally set up our notation. For any $\t=\left(\a,\b^{'},\g^{'}\right)^{'}\in\R^{d}$
and any functions $h:\R^{d_{z}}\to\R$ and $\pi:\R^{d_{z}}\to\R^{d_{x}}$,
define
\begin{align*}
g^{*}\left(z;\t,h,\pi\right) & :=h\left(z\right)-w\left(z,\pi\right)^{'}\t,\quad
g\left(y,z;\t,\pi\right) :=y-w\left(z,\pi\right)^{'}\t
\end{align*}
with $w\left(z,\pi\right):=\left(1,z^{'},\pi\left(z\right)^{'}\right)^{'}$
so that 
\begin{align*}
g^{*}\left(Z_{i};\t_{0},h_{0},\pi_{0}\right) & =h_{0}\left(Z_{i}\right)-w\left(Z_{i},\pi_{0}\right)^{'}\t_{0}\equiv0,\\
g\left(Y_{i},Z_{i};\t_{0},\pi_{0}\right) & =Y_{i}-w\left(Z_{i},\pi_{0}\right)^{'}\t_{0}=\e_{i}+u_{i}^{'}\g_{0},
\end{align*}
where $u_{i}:=X_{i}-\E\left[\rest{X_{i}}Z_{i}\right]$ and
\[
\E\left[\rest{g\left(Y_{i},Z_{i};\t_{0},h_{0},\pi_{0}\right)}Z_{i}\right]=\E\left[\rest{\e_{i}+u_{i}^{'}\g_{0}}Z_{i}\right]=0.
\]
We construct the following quadratic population criterion function:
\begin{align*}
Q^{*}\left(\t,h,\pi\right) & :=\frac{1}{2}\E\left[g^{*2}\left(Z_{i};\t,h,\pi\right)\right],\quad
Q\left(\t,\pi\right) :=\frac{1}{2}\E\left[g^{2}\left(Y_{i},Z_{i};\t,\pi\right)\right],
\end{align*}
so that
$Q^{*}\left(\t_{0};h_{0},\pi_{0}\right) =0,$ and $
Q\left(\t_{0};\pi_{0}\right) =\E\left[\left(\e_{i}+u_{i}^{'}\g_{0}\right)^{2}\right]$.

\begin{cor}
\label{prop:minQ0}Under Assumptions \ref{assu:exog}-\ref{assu:FullRankW}, $\t_{0}$
is the unique minimizer of $Q^{*}\left(\cd,h_{0},\pi_{0}\right)$ and  $Q\left(\cd,\pi_{0}\right)$, i.e.,
$\t_{0}=\arg\min_{\t\in\R^{d}}Q^{*}\left(\t,h_{0},\pi_{0}\right)=\arg\min_{\t\in\R^{d}}Q\left(\t,\pi_{0}\right)$.
\end{cor}
\begin{proof}
Note that $Q^{*}\left(\t,h_{0},\pi_{0}\right)=\frac{1}{2}\E\left[g^{*2}\left(Z_{i};\t,h_{0},\pi_{0}\right)\right]=0$
implies $g^{*}\left(Z_{i};\t,h_{0},\pi_{0}\right)=0$ almost surely,
and thus $\E\left[w\left(Z_{i},\pi_{0}\right)g^{*}\left(Z_{i};\t,h_{0},\pi_{0}\right)\right]=\E\left[W_{i}W_{i}^{'}\right]\t-\E\left[W_{i}h_{0}\left(Z_{i}\right)\right]={\bf 0}$,
which implies that $\t=\t_{0}$. In the meanwhile, the first-order condition for the minimization of
$Q\left(\t,\pi_{0}\right)=\frac{1}{2}\E\left[g^{2}\left(Y_{i},Z_{i};\t,\pi_{0}\right)\right]$
is given by $\E\left[w\left(Z_{i},\pi_{0}\right)g\left(Y_{i},Z_{i};\t_{0},\pi_{0}\right)\right]=\E\left[W_{i}\left(Y_{i}-W_{i}^{'}\t_{0}\right)\right]={\bf 0}$,
which is equivalent to $\E\left[W_{i}W_{i}^{'}\right]\t_{0}-\E\left[W_{i}Y_{i}\right]={\bf 0}$.
\end{proof}

\subsection{\label{subsec:Pf_AsymNorm}Proof of Theorem \ref{thm:AsymNorm}}
\begin{proof}
It is well known that a Sobolev space of order $s>\frac{d_{z}}{2}$
is a Donsker class of functions. Since the residual functions in our setup are given by
\[
w\left(Z_{i},\pi\right)g^{*}\left(Z_{i},\t,h,\pi\right)=\left(h\left(Z_{i}\right)-\a-Z_{i}^{'}\b-\pi\left(Z_{i}\right)^{'}\g\right)\left(\begin{array}{c}
1\\
Z_{i}\\
\pi\left(Z_{i}\right)
\end{array}\right),
\]
\[
w\left(Z_{i},\pi\right)g\left(Y_{i},Z_{i},\t,\pi\right)=\left(Y_{i}-\a-Z_{i}^{'}\b-\pi\left(Z_{i}\right)^{'}\g\right)\left(\begin{array}{c}
1\\
Z_{i}\\
\pi\left(Z_{i}\right)
\end{array}\right),
\]
which are smooth functions of $\left(\t,h,\pi\right)$,
the function classes 
\begin{align*}
{\cal F}^{*} & :=\left\{ w\left(\cd,\pi\right)g^{*} \left(\cd,\t,h,\pi\right) - w\left(\cd,\pi_0\right)g^{*} \left(\cd,\t,h_0,\pi_0\right):\t\in\R^{d},h\in{\cal H},\pi\in{\cal H}\right\} ,\\
{\cal F}^{*} & :=\left\{ w\left(\cd,\pi\right)g\left(\cd,\cd,\t,\pi\right) - w\left(\cd,\pi_0\right)g\left(\cd,\cd,\t,\pi_0\right):\t\in\R^{d},\pi\in{\cal H}\right\} ,
\end{align*}
are also Donsker, and thus satisfy the stochastic equicontinuity
condition.

We then proceed to derive the influence functions for $\hat{\t}^{*}$
and $\hat{\t}$ separately.
\begin{itemize}
\item[(a)] For $\hat{\t}^{*}$, recall that $g^{*}\left(z;\t,h,\pi\right)=h\left(z\right)-\a-z^{'}\b-\pi\left(z\right)^{'}\g$
with $g^{*}\left(Z_{i};\t_{0},h_{0},\pi_{0}\right)\equiv0.$
Hence,
$\Dif_{\t}Q^{*}\left(\t,h_{0},\pi_{0}\right) =-\E\left[w\left(Z_{i},\pi_{0}\right)g^{*}\left(Z_{i};\t,h_{0},\pi_{0}\right)\right],$ 
with $\Dif_{\t}Q^{*}\left(\t_{0},h_{0},\pi_{0}\right)  =-\E\left[W_{i}0\right]={\bf 0}$
and $\Dif_{\t\t}Q^{*}\left(\t,h_{0},\pi_{0}\right) =\E\left[w\left(Z_{i},\pi_{0}\right)w\left(Z_{i},\pi_{0}\right)^{'}\right] = \E\left[W_iW_i^{'}\right] =\Sigma_{0}.$
Furthermore,
\begin{align*}
 & D_{\left(h,\pi\right)}\left[\Dif_{\t}Q^{*}\left(\t_{0},h_{0},\pi_{0}\right),h-h_{0},\pi-\pi_{0}\right]\\
:=\  & \lim_{t\downto0}\frac{1}{t}\left(\Dif_{\t}Q^{*}\left(\t_{0},h_{0}+t\left(h-h_{0}\right),\pi_{0}+t\left(\pi-\pi_{0}\right)\right)-\Dif_{\t}Q^{*}\left(\t,h_{0},\pi_{0}\right)\right)\\
=\  & -\lim_{t\downto0}\frac{1}{t}\E\left[t\left(\begin{array}{c}
h-h_{0}-\left(\pi-\pi_{0}\right)^{'}\g_{0}\\
Z_{i}\left(h-h_{0}-\left(\pi-\pi_{0}\right)^{'}\g_{0}\right)\\
\pi_{0}\left(h-h_{0}-\left(\pi-\pi_{0}\right)^{'}\g_{0}\right)+\left(\pi-\pi_{0}\right)g^{*}\left(\cd,\t_{0},h_{0},\pi_{0}\right)
\end{array}\right)_{Z_{i}}\right]\\
= & -\E\left[\left(\begin{array}{c}
h-h_{0}-\left(\pi-\pi_{0}\right)^{'}\g_{0}\\
Z_{i}\left(h-h_{0}-\left(\pi-\pi_{0}\right)^{'}\g_{0}\right)\\
\pi_{0}\left(h-h_{0}-\left(\pi-\pi_{0}\right)^{'}\g_{0}\right)
\end{array}\right)_{Z_{i}}\right],
\end{align*}
where the subscript $Z_{i}$ means that all the functions $h,h_{0},\pi,\pi_{0},g\left(\cd,\t_{0},h_{0},\pi_{0}\right)$
are evaluated at $Z_{i}$. and the last equality uses the observation
that $g^{*}\left(z,\t_{0},h_{0},\pi_{0}\right)\equiv0$. Define
\begin{align*}
\psi^{*}\left(Y_{i},X_{i},Z_{i}\right) & :=-\left(\begin{array}{c}
Y_{i}-h_{0}\left(Z_{i}\right)-\left(X_{i}-\pi_{0}\left(Z_{i}\right)\right)^{'}\g_{0}\\
Z_{i}\left(Y_{i}-h_{0}\left(Z_{i}\right)-\left(X_{i}-\pi_{0}\left(Z_{i}\right)\right)^{'}\g_{0}\right)\\
\pi_{0}\left(Z_{i}\right)\left(Y_{i}-h_{0}\left(Z_{i}\right)-\left(X_{i}-\pi_{0}\left(Z_{i}\right)\right)^{'}\g_{0}\right)
\end{array}\right)\\
 & =-\left(\begin{array}{c}
\e_{i}\\
Z_{i}\e_{i}\\
\pi_{0}\left(Z_{i}\right)\e_{i}
\end{array}\right)=-\e_{i}W_{i}
\end{align*}
since, by \eqref{eq:Model} and \eqref{eq:E_YZ},
$Y_{i}-h_{0}\left(Z_{i}\right)-\left(X_{i}-\pi_{0}\left(Z_{i}\right)\right)^{'}\g_{0}=\e_{i}$.
Hence,
\[
\E\left[\psi^{*}\left(Y_{i},X_{i},Z_{i}\right)\right]=-\E\left[\E\left[\rest{\e_{i}}Z_{i}\right]W_{i}\right]={\bf 0},
\]
and $\E\left[\norm{\psi^{*}\left(Y_{i},X_{i},Z_{i}\right)}^{2}\right]<\infty$.
Noting that $g^{*}\left(Z_{i},\t_{0},h_{0},\pi_{0}\right)\equiv0$,
then by the standard theory for semiparametric two-stage estimation, e.g. Theorems 8.1 \& 8.2 of \citet*{newey1994large}, we have
\begin{align*}
\sqrt{n}\left(\hat{\t}^{*}-\t_{0}\right) & =-\Sigma_{0}^{-1}\frac{1}{\sqrt{n}}\sum_{i=1}^{n}\left(-g^{*}\left(Z_{i},\t_{0},h_{0},\pi_{0}\right)W_i+\psi^{*}\left(Y_{i},X_{i},Z_{i}\right)\right)+o_{p}\left(1\right)\\
 & =-\Sigma_{0}^{-1}\frac{1}{\sqrt{n}}\sum_{i=1}^{n}\psi^{*}\left(Y_{i},X_{i},Z_{i}\right)+o_{p}\left(1\right)
  \dto\cN\left({\bf 0},V_{0}=\Sigma_{0}^{-1}\O_{0}\Sigma_{0}^{-1}\right)
\end{align*}
with
$\O_{0} =\E\left[\psi^{*}\left(Y_{i},X_{i},Z_{i}\right)\psi^{*}\left(Y_{i},X_{i},Z_{i}\right)^{'}\right]=\E\left[\e_{i}^{2}W_{i}W_{i}^{'}\right]$.

\item[(b)] For $\text{\ensuremath{\hat{\t}}},$ recall $g\left(y,z;\t,\pi\right):=y-w\left(z,\pi\right)^{'}\t$
and 
$g\left(Y_{i},Z_{i};\t_{0},\pi_{0}\right)=\e_{i}+u_{i}^{'}\g_{0}.$
Hence,
\begin{align*}
\Dif_{\t}Q\left(\t,\pi_{0}\right) & =-\E\left[w\left(Z_{i},\pi_{0}\right)g\left(Y_{i},Z_{i};\t,\pi_{0}\right)\right],
\end{align*}
with
$\Dif_{\t}Q\left(\t_{0},\pi_{0}\right)  =-\E\left[w\left(Z_{i},\pi_{0}\right)\left(\e_{i}+u_{i}^{'}\g_{0}\right)\right] ={\bf 0}$
and $$\Dif_{\t\t}Q\left(\t,\pi_{0}\right) =\E\left[w\left(Z_{i},\pi_{0}\right)w\left(Z_{i},\pi_{0}\right)^{'}\right] = \E\left[W_{i}W_{i}^{'}\right]=\Sigma_{0}.$$
Furthermore,
\begin{align*}
 & D_{\pi}\left[\Dif_{\t}Q\left(\t_{0},\pi_{0}\right),\pi-\pi_{0}\right]\\
:=\  & \lim_{t\downto0}\frac{1}{t}\left(\Dif_{\t}Q\left(\t_{0},\pi_{0}+t\left(\pi-\pi_{0}\right)\right)-\Dif_{\t}Q\left(\t,\pi_{0}\right)\right)\\
=\  & -\lim_{t\downto0}\frac{1}{t}\E\left[t\left(\begin{array}{c}
-\left(\pi-\pi_{0}\right)^{'}\g_{0}\\
-Z_{i}\left(\pi-\pi_{0}\right)^{'}\g_{0}\\
-\pi_{0}\left(\pi-\pi_{0}\right)^{'}\g_{0}+\left(\pi-\pi_{0}\right)g\left(Y_{i},Z_{i},\t_{0},\pi_{0}\right)
\end{array}\right)_{Z_{i}}\right]\\
=\ & \E\left[\left(\begin{array}{c}
\left(\pi-\pi_{0}\right)^{'}\g_{0}\\
Z_{i}\left(\pi-\pi_{0}\right)^{'}\g_{0}\\
\pi_{0}\left(\pi-\pi_{0}\right)^{'}\g_{0}
\end{array}\right)_{Z_{i}}\right]
\end{align*}
where the last equality follows from the Law of Iterated Expectations
and
\[
\E\left[\rest{g\left(Y_{i},Z_{i},\t_{0},\pi_{0}\right)}Z_{i}\right]=\E\left[\rest{\e_{i}+u_{i}^{'}\g_{0}}Z_{i}\right]=0.
\]
Defining
\begin{align*}
\psi\left(Y_{i},X_{i},Z_{i}\right) & :=\left(\begin{array}{c}
\begin{array}{c}
\left(X_{i}-\pi_{0}\left(Z_{i}\right)\right)^{'}\g_{0}\\
Z_{i}\left(X_{i}-\pi_{0}\left(Z_{i}\right)\right)^{'}\g_{0}\\
\pi_{0}\left(Z_{i}\right)\left(X_{i}-\pi_{0}\left(Z_{i}\right)\right)^{'}\g_{0}
\end{array}\end{array}\right)=u_{i}^{'}\g_{0}W_{i},
\end{align*}

we have $\E\left[\psi\left(Y_{i},X_{i},Z_{i}\right)\right]=\E\left[W_{i}\E\left[\rest{u_{i}^{'}}Z_{i}\right]\g_{0}\right]={\bf 0}.$
Again, based on the standard results for the asymptotic theory of
semiparametric two-stage estimators, such as Theorems 8.1 \& 8.2 of
\citet*{newey1994large}, we have
\begin{align*}
\sqrt{n}\left(\hat{\t}-\t_{0}\right) & =-\Sigma_{0}^{-1}\frac{1}{\sqrt{n}}\sum_{i=1}^{n}\left(-g\left(Z_{i},\t_{0},\pi_{0}\right)W_i+\psi\left(Y_{i},X_{i},Z_{i}\right)\right)+o_{p}\left(1\right)\\
 & =-\Sigma_{0}^{-1}\frac{1}{\sqrt{n}}\sum_{i=1}^{n}\left(-\left(\e_{i}+u_{i}^{'}\g_{0}\right)W_{i}+u_{i}^{'}\g_{0}W_{i}\right)+o_{p}\left(1\right)\\
 & =\Sigma_{0}^{-1}\frac{1}{\sqrt{n}}\sum_{i=1}^{n}\e_{i}W_{i}+o_{p}\left(1\right)
  \dto\cN\left({\bf 0},V_{0}=\Sigma_{0}^{-1}\O_{0}\Sigma_{0}^{-1}\right).
\end{align*}
\end{itemize}
\end{proof}

\subsection{\label{subsec:pf_VarEst}Proof of Theorem \ref{thm:VarEst}}
\begin{proof}
Given finite fourth moment in Assumption \ref{assu:Moment2}, we have
\[
\frac{1}{n}\sum_{i=1}^{n}W_{i}W_{i}^{'}\pto\Sigma_{0},\quad\frac{1}{n}\sum_{i=1}^{n}\e_{i}^{2}W_{i}W_{i}^{'}\pto\O_{0}.
\]
Moreover, given the consistency of the first-stage nonparametric estimator
$\hat{\pi}$ in Assumption \ref{assu:NPFS} and the consistency of estimators
$\hat{\t}$ in Theorem \ref{thm:AsymNorm}, we have
\begin{align*}
\hat{W}_{i}-W_{i} & =\left(0, {\bf 0}', \hat{\pi}\left(Z_{i}\right)'-\pi_{0}\left(Z_{i}\right)'\right)^{'} \pto{\bf 0}, \\
\hat{\e}_{i}-\e_{i} & =\a_{0}-\hat{\a}-Z_{i}^{'}\left(\hat{\b}-\b_{0}\right)-X_{i}^{'}\left(\hat{\g}-\g_{0}\right)\pto0,
\end{align*}
and thus
\begin{align*}
\hat{\Sigma}-\Sigma_{0} & =\frac{1}{n}\sum_{i=1}^{n}\left(\hat{W}_{i}\hat{W}_{i}^{'}-W_{i}W_{i}^{'}\right)+\frac{1}{n}\sum_{i=1}^{n}W_{i}W_{i}^{'}-\Sigma_{0}\pto{\bf 0},\\
\hat{\O}-\O_{0} & =\frac{1}{n}\sum_{i=1}^{n}\left(\hat{\e}_{i}^{2}\hat{W}_{i}\hat{W}_{i}^{'}-\e_{i}W_{i}W_{i}^{'}\right)+\frac{1}{n}\sum_{i=1}^{n}\e_{i}^{2}W_{i}W_{i}^{'}-\O_{0}\pto{\bf 0}.
\end{align*}
Hence, $\hat{V}:=\hat{\Sigma}^{-1}\hat{\O}\hat{\Sigma}^{-1}\pto V_{0}=\Sigma_{0}^{-1}\O_{0}\Sigma_{0}^{-1}$.
\end{proof}

\subsection{\label{subsec:pf_AsymDisc}Proof of Theorem \ref{thm:AsymDisc}}

\begin{proof}
Since $\hat{\t}_{disc}$ is a 2SLS estimator, it is $\sqrt{n}$-consistent
and asymptotic normal, with the asymptotic variance given by the formula
$V_{0,disc}:=\Sigma_{0,disc}^{-1}\O_{0,disc}\Sigma_{0,disc}^{-1}$ 
with
\begin{align*}
\Sigma_{0,disc} & :=\E\left[W_{i}^{'}D_{i}\right]\left(\E\left[D_{i}D_{i}^{'}\right]\right)^{-1}\E\left[D_{i}^{'}W_{i}\right]\\
 & =\sum_{k=1}^{K}\left(p_{k}\ol W_{k}\cd\frac{1}{p_{k}}\cd p_{k}\ol W_{k}^{'}\right)=\sum_{k=1}^{K}p_{k}\ol W_{k}\ol W_{k}^{'},
\end{align*}
and 
\begin{align*}
\O_{0,disc} & :=\E\left[W_{i}^{'}D_{i}\right]\left(\E\left[D_{i}D_{i}^{'}\right]\right)^{-1}\E\left[\e_{i}^{2}D_{i}D_{i}^{'}\right]\left(\E\left[D_{i}D_{i}^{'}\right]\right)^{-1}\E\left[D_{i}W_{i}^{'}\right]\\
 & =\sum_{k=1}^{K}\left(p_{k}\ol W_{k}\cd\frac{1}{p_{k}}\cd p_{k}\ol{\s}_{k}\cd\frac{1}{p_{k}}\cd p_{k}\ol W_{k}'\right)=\sum_{k=1}^{K}p_{k}\ol{\s}_{k}^{2}\ol W_{k}\ol W_{k}^{'}.
\end{align*}
Note that, in the proofs above, we exploited the fact that each $D_{ik}$
is a partition cell dummy along with the associated properties such as, for
all $k$ and $j\neq k$,
$\E\left[D_{ik}\right]=p_{k},\ D_{ik}^{2}=D_{ik},\ D_{ik}D_{ij}=0.$
\end{proof}

\subsection{Proof of Theorem \ref{lem:DiscLessEff}}
\begin{proof}
Under homoskedasticity, the formulas for $V_0$ and $V_{0,disc}$ simplify to 
\[\begin{aligned}
V_0 = \s^2_\e \left(\E\left[W_{i}W_{i}^{'}\right] \right)^{-1},\quad
V_{0,disc} = \s^2_\e \left(\sum_{k=1}^{K} p_k \ol W_{k}\ol W_{k}^{'}\right)^{-1}.
\end{aligned}
\]
Recalling that $\ol{W}_k = \E[\rest{W_i} Z_i \in {\cal Z}_k]$, we have
\[
\begin{aligned}
V_0^{-1}  - V_{0,disc}^{-1} & = \frac{1}{\s^2_\e} \left(\E\left[W_{i}W_{i}^{'}\right] -\sum_{k=1}^{K} p_k\ol W_{k}\ol W_{k}^{'} \right) \\ 
& = \frac{1}{\s^2_\e} \sum_{k=1}^K p_k\left( \E\left[\rest{W_{i}W_{i}^{'}}{Z_i \in \cal Z}_{k} \right]-\ol W_{k}\ol W_{k}^{'}\right) 
 = \frac{1}{\s^2_\e}\sum_{k=1}^K p_k \var(\rest{W_{i}}Z_i \in {\cal Z}_{k}),
\end{aligned}
\]
which is positive semi-definite. Hence, $V_{0, disc} - V_0$ is positive semi-definite.
\end{proof}

\subsection{Proof of Theorem \ref{thm:nonlinear}}

\begin{proof}
Under Assumption \ref{assu:exog}, the parameter $\t_0$ should satisfy 
$g(z, \t_0):=\E[\rest{Y_i}Z_i=z]-m_0(z, \t_0)=0,$
for any $z\in {\cal Z}$. The function $g(z, \cdot)$ is continuously differentiable since the function $m_0$ is continuously differentiable by assumption. Then by the local inverse theorem in \cite{ambrosetti1995primer} (Chapter 2, Theorem 2), $\t_0$ is locally identified if
 there exists $d=dims(\t_0)$ distinct points $z_1,..., z_d\in {\cal Z}$ such that the following condition holds:
\[ G_0:=\left[
\begin{array}{cccc}
 \partial_{\t_1}g(z_1, \t_0), & \partial_{\t_2}g(z_1, \t_0),& ..., &\partial_{\t_d}g(z_1, \t_0)\\
   \partial_{\t_1}g(z_2, \t_0), & \partial_{\t_2}g(z_2, \t_0),& ..., &\partial_{\t_d}g(z_2, \t_0)\\
  \vdots & \vdots & \vdots & \vdots \\
   \partial_{\t_1}g(z_d, \t_0), & \partial_{\t_2}g(z_d, \t_0),& ..., &\partial_{\t_d}g(z_d, \t_0)\\
 \end{array}
 \right] \text{ has full rank.}
 \]
As shown in Lemma  \ref{assu:FullRankW}, the above full rank condition is equivalent to the requirement that
$\partial_{\t_1}g(Z_i, \t_0),  \partial_{\t_2}g(Z_i, \t_0),...,  \partial_{\t_d}g(Z_i, \t_0) \text{ are not multicollinear.}$ The no multicollinearity condition is also equivalent to the requirement that
$\E\left[ \Dif_\t g(Z_i, \t_0)\Dif_{\t^{'}} g(Z_i, \t_0) \right]=\E\left[ \Dif_\t m_0(Z_i, \t_0)\Dif_{\t^{'}} m_0(Z_i, \t_0)  \right] \text{ has full rank.} $

\end{proof}

\subsection{Proof of Lemma \ref{lem:quantile}}
\begin{proof}
Under Assumption
\ref{assu:eps_pdf}, the moment function can be expressed as follows:
\begin{align*}
g\left(z,\t \right) &=  \E\left[\rest{\ind\{Y_i\leq \a +Z_{i}^{'}\b+X_{i}^{'}\g \} } Z_i=z \right] -\tau \\
 & =\E\left[\rest{\E\left[\rest{\ind\left\{ Y_{i}\leq\a+Z_{i}^{'}\b+X_{i}^{'}\g\right\} }X_{i}\right]}Z_{i}=z\right]-\tau \\
 & =\E\left[\rest{\int\ind\left\{ \e \leq\a-\a_{0}+z^{'}\left(\b-\b_{0}\right)+X_{i}^{'}\left(\g-\g_{0}\right)\right\} f_{\e\mid X, Z}\left(\rest{\e}X_{i},z\right)d\e}Z_{i}=z\right]-\tau. 
\end{align*}
The derivative of $\Dif_{\t} g\left(z,\t \right)=\Dif_{\t}m_0\left(z,\t\right)$ with respect to $\t$ is given as
\begin{align*}
&\Dif_{\t} g\left(z,\t \right)=\Dif_{\t}m_0\left(z,\t\right) \\
=& \E\left[\rest{\Dif_{\t}\int\ind\left\{ \e\leq\a-\a_{0}+z^{'}\left(\b-\b_{0}\right)+X_{i}^{'}\left(\g-\g_{0}\right)\right\} f_{\e \mid X, Z}\left(\rest{\e}X_{i},z\right)d\e}Z_{i}=z\right]\\
 =& \E\left[\rest{f_{\e \mid X, Z}\left(\rest{\a-\a_{0}+z^{'}\left(\b-\b_{0}\right)+X_{i}^{'}\left(\g-\g_{0}\right)}X_{i},z\right)\left(1,z^{'},X_{i}^{'}\right)^{'}}Z_{i}=z\right].
\end{align*}
Evaluating at $\t_0$, the derivative $m_{0}\left(z,\t \right)$ is simplified as
\begin{align*}
\Dif_{\t}m_0\left(z,\t_{0}\right) & =\E\left[\rest{f_{\e \mid X, Z}\left(\rest 0X_{i},z\right)\left(\begin{array}{c}
1\\
z\\
X_{i}
\end{array}\right)}Z_{i}=z\right] 
  =f_{\e \mid Z}\left(\rest 0z\right)\left(\begin{array}{c}
1\\
z\\
\E\left[\rest{\frac{f_{\e \mid X, Z}\left(\rest 0X_{i},z\right)}{f_{\e \mid Z}\left(\rest 0z\right)}X_{i}}Z_{i}=z\right]
\end{array}\right).
\end{align*}
Applying Bayes' rule, we have
\[ \begin{aligned}
\E\left[\rest{\frac{f_{\e \mid X, Z}\left(\rest 0X_{i},z\right)}{f_{\e \mid Z}\left(\rest 0z\right)}X_{i}}Z_{i}=z\right]
&=\int\frac{f_{\e \mid X, Z}\left(\rest {0} x,z\right)}{f_{\e \mid Z}\left(\rest 0z\right)}xf_{X\mid Z}\left(\rest xz\right)dx \\
&=\int xf_{X\mid \e, Z}\left(\rest {x}z, \e=0\right)dx =\E\left[\rest{X_{i}} Z_{i}=z, \e_{i}=0\right].
\end{aligned}
\]
Therefore,
\[
\Dif_{\t}m_0\left(Z_i,\t_{0}\right) 
=f_{\e\mid Z}\left(\rest 0 Z_i\right)
\left(\begin{array}{c}
1\\
Z_i\\
\E\left[\rest{X_{i}}Z_{i},\e_{i}=0\right]
\end{array}\right).
\]
Since $f_{\e\mid Z}\left(\rest 0 Z_i=z\right)>0$ for any $z\in \cal Z$, the full rank condition of $\E\left[\Dif_\t m_0(Z_i, \t_0)\Dif_{\t^{'}} m_0(Z_i, \t_0) \right]$ is equivalent to the no multicollinearity of $1$, $Z_i$, $\tilde{\pi}_0(Z_i):=\E\left[\rest{X_{i}}Z_{i},\e_{i}=0\right]$.

\end{proof}

\newpage
\section{Online Appendix}\label{sec:App_OL}

\subsection{Variance Comparison with the Infeasible Estimator}\label{subsec:VInfeasible}

We now provide a more detailed discussion about the asymptotic variance
of our estimators. We note that $\O_{0}$ and $V_{0}$ are in general different from
the asymptotic variance matrices that would appear in the \emph{infeasible}
regression of $Y_{i}$ on $1,$ $Z_{i}$, and the \emph{true value
of} $\pi_{0}\left(Z_{i}\right)$, which has no endogeneity issue:
\begin{equation}
Y_{i}=\a_{0}+Z_{i}^{'}\b_{0}+\pi_{0}\left(Z_{i}\right)^{'}\g_{0}+v_{i},\label{eq:RegInfeasible}
\end{equation}
where $v_{i}:=\e_{i}+u_{i}^{'}\g_{0},\  u_i=X_i-\pi_0(Z_i),\ \E\left[\rest{v_{i}}Z_{i}\right]=0$.
The infeasible OLS estimator based on \eqref{eq:RegInfeasible} is
then given by
\[
\tilde{\t}_{infeasible}:=\left(\frac{1}{n}\sum_{i=1}^{n}W_{i}W_{i}^{'}\right)^{-1}\frac{1}{n}\sum_{i=1}^{n}W_{i}Y_{i},
\]
which is $\sqrt{n}$-consistent and asymptotically normal with asymptotic
variance matrix $V_{infeasible}:=\Sigma_{0}^{-1}\O_{infeasible}\Sigma_{0}^{-1}$,
where
\[
\O_{infeasible}:=\E\left[v_{i}^{2}W_{i}W_{i}^{'}\right]=\E\left[\left(\e_{i}+u_{i}^{'}\g_{0}\right)^{2}W_{i}W_{i}^{'}\right].
\]
Recalling that $\E\left[\rest{\e_{i}}Z_{i}\right]=0$ and $\E\left[\rest{u_{i}}Z_{i}\right]=0$,
we have
\begin{align}
\O_{infeasible}-\O_{0} & =\E\left[\left(\left(\e_{i}+u_{i}^{'}\g_{0}\right)^{2}-\e_{i}^{2}\right)W_{i}W_{i}^{'}\right]\nonumber \\
 & =\E\left[\left(2\E\left[\rest{\e_{i}u_{i}^{'}}Z_{i}\right]\g_{0}+\g_{0}^{'}\E\left[\rest{u_{i}u_{i}^{'}}Z_{i}\right]\g_{0}\right)W_{i}W_{i}^{'}\right]\nonumber \\
 & =\E\left[\left(2\text{Cov}\left(\rest{\e_{i},u_{i}}Z_{i}\right)\g_{0}+\g_{0}^{'}\text{Var}\left(\rest{u_{i}}Z_{i}\right)\g_{0}\right)W_{i}W_{i}^{'}\right],\label{eq:VarComp}
\end{align}
which can be positive or negative definite (or zero), depending on
the value of $\g_{0}$ and the conditional covariance between $\e_{i}$
and $u_{i}$ given $Z_{i}$.
Hence, our estimators $\hat{\t}^{*},\hat{\t}$ do not dominate, or
are dominated by, the infeasible estimator in terms of asymptotic
efficiency.

As a further illustration, consider the simple case with scalar-valued
$X_{i}$ and $Z_{i}$, and furthermore suppose that $\text{Var}\left(\rest{\e_{i}}Z_{i}\right)=\text{Var}\left(\rest{u_{i}}Z_{i}\right)=1$
and 
\[
\text{Cor}\left(\rest{\e_{i},u_{i}}Z_{i}\right)=\text{Cor}\left(\e_{i},u_{i}\right)=\rho_{0}\in\left[0,1\right].
\]
The parameter $\rho_{0}$ can be viewed as a measure of the extent
of the endogeneity issue in the linear regression \eqref{eq:Model}.
It is then clear from \eqref{eq:VarComp} above that:
\begin{itemize}
\item $\O_{infeasible}-\O_{0}$ is positive definite if and only if $2\rho_{0}\g_{0}+\g_{0}^{2}>0$,
i.e.,
\[
\g_{0}>0\text{ and }\rho_{0}>-\frac{1}{2}\g_{0}\quad\text{or }\quad\g_{0}<0\text{ and }\rho_{0}<-\frac{1}{2}\g_{0}.
\]
In words, this says that our estimators $\hat{\t}^{*}$ and $\hat{\t}$
are asymptotically \emph{more efficient} than the infeasible estimator
obtained with the true values of $\pi_{0}\left(Z_{i}\right)$, when
the endogeneity issue is \emph{not too large in the opposite direction}
of $\g_{0}$. Note that this is trivially satisfied when $\g_{0}\neq0$
and $\rho_{0}=0$, i.e., when $X_{i}$ has some effect on $Y_{i}$
and $X_{i}$ is exogenous. 
\item $\O_{infeasible}-\O_{0}$ is negative definite if and only if 
\[
\g_{0}>0\text{ and }\rho_{0}<-\frac{1}{2}\g_{0}\quad\text{or }\quad\g_{0}<0\text{ and }\rho_{0}>-\frac{1}{2}\g_{0}.
\]
In words, this says that our estimators $\hat{\t}^{*}$ and $\hat{\t}$
are asymptotically \emph{less efficient} than the infeasible estimator
when the endogeneity issue is \emph{sufficiently large in the opposite
direction} of $\b_{0}$.
\item $\O_{infeasible}=\O_{0}$ if and only if
\[
\g_{0}=0\quad\text{or}\quad\rho_{0}=-\frac{1}{2}\g_{0}.
\]
In words, our estimators $\hat{\t}^{*}$ and $\hat{\t}$ are asymptotically
\emph{equally efficient as} the infeasible estimator when $X_{i}$
has no effect on $Y_{i}$, or if the endogeneity works into the opposition
direction of $\g_{0}$ exactly to a certain extent.
\end{itemize}
While it is true that our feasible estimators may be more or less
efficient than the infeasible estimator with known values of $\pi_{0}\left(Z_{i}\right)$,
the analysis in the illustration above suggests that our feasible
estimators are likely to be \emph{more efficient} in scenarios where
applied researchers are somewhat confident that the endogeneity issue
won't overwhelm the true effect of $X_{i}$ on $Y_{i}$ in the opposite
direction.

\subsection{Additional Simulation Results}\label{subsec:Sim_Additional}

Below we report the simulation results for the estimators of the intercept $\a_0$ and the coefficients on the exogenous regressors $\b_0$. The title of each table below designates the underlying DGP and the parameter being estimated. Specifically, ``Bin $X$ with Bin $Z_1, Z_2$" refers to the DGP in Section \ref{subsec: Sim1}, ``Bin $X$ with Cts $Z$" refers that in Section \ref{subsec: Sim2}, while ``Cts $X$ with Cts $Z$" refers that in Section \ref{subsec:Sim3}.

\begin{sidewaystable}[!htbp] 
\centering
\caption{Bin $X$ with Bin $Z_1, Z_2$: Performance of $\hat{\alpha}$ (Intercept)  \\
Different Degrees of Exclusion Violations}
 \label{table:Sim1_alpha_exc}
\begin{tabular}{c|cccc|cccc|cccc}
\hline
\hline
& \multicolumn{4}{c|}{$n=250$ } & \multicolumn{4}{c|}{$n=500$ }& \multicolumn{4}{c}{$n=1000$ }  \\
\hline
 Est& Bias & SD & RMSE & CP &Bias & SD & RMSE & CP & Bias & SD & RMSE & CP \\
\hline
&\multicolumn{12}{c}{$\b_{01}=\b_{02}=1$}  \\
\hline
$\hat{\theta}$ &   0.001 & 0.141&  0.141&  0.957 & -0.001&  0.104 & 0.104&  0.951& -0.003 & 0.074&  0.074  &0.944
  \\[1.5ex]
 $\hat{\theta}^{*}$ &   0.001 & 0.141&  0.141&  0.957 & -0.001&  0.104 & 0.104&  0.951& -0.003 & 0.074&  0.074  &0.944
  \\[1.5ex]
 $\hat{\theta}_{disc}$ &   0.001 & 0.141&  0.141&  0.957 & -0.001&  0.104 & 0.104&  0.951& -0.003 & 0.074&  0.074  &0.944
  \\[1.5ex]
    $\hat{\theta}_{2sls}$ & 1.399 & 17.307 & 17.364 & 0.846 & 2.602 &  59.012 &  59.069 & 0.859 &  -0.254 & 33.300  &  33.301 & 0.854
  \\[1.5ex]
  $\hat{\theta}_{ols}$ &  0.241 & 0.119  &0.269  &0.486 &  0.242 & 0.088&  0.258 & 0.204  & 0.242 & 0.062 & 0.250  & 0.026
       \\[1.5ex]
 \hline 
&\multicolumn{12}{c}{$\b_{01}=\b_{02}=0.5$} \\
\hline
$\hat{\theta}$ &   0.001 &  0.141 &  0.141 &  0.957&  -0.001 &   0.104 &   0.104 &  0.951& -0.003  & 0.074 &  0.074 & 0.944
  \\[1.5ex]
 $\hat{\theta}^{*}$ &  0.001 &  0.141 &  0.141 &  0.957&  -0.001 &   0.104 &   0.104 &  0.951&  -0.003  & 0.074 &  0.074 & 0.944
  \\[1.5ex]
 $\hat{\theta}_{disc}$ &  0.001 &  0.141 &  0.141 &  0.957&  -0.001 &   0.104 &   0.104 &  0.951&  -0.003  & 0.074 &  0.074 & 0.944
  \\[1.5ex]
      $\hat{\theta}_{2sls}$ &  0.829 &  9.015 &  9.054 &  0.876 &  1.405 & 30.121&  30.154 &  0.881 & 0.022 & 16.656 & 16.656 & 0.878
        \\[1.5ex]
  $\hat{\theta}_{ols}$ &  0.241&   0.119&   0.269 &  0.486&   0.242&    0.088 &   0.258&   0.204& 0.242 &  0.062 &  0.250  & 0.026
     \\[1.5ex]
 \hline
 &\multicolumn{12}{c}{$\b_{01}=\b_{02}=0$} \\
\hline
$\hat{\theta}$ &   0.001 &  0.141  & 0.141   &0.957 &-0.001 & 0.104  &0.104  &0.951 & -0.003 & 0.074 & 0.074 & 0.944
  \\[1.5ex]
 $\hat{\theta}^{*}$ & 0.001 &  0.141  & 0.141   &0.957 &-0.001 & 0.104  &0.104  &0.951 & -0.003 & 0.074 & 0.074 & 0.944
   \\[1.5ex]
 $\hat{\theta}_{disc}$ &  0.001 &  0.141  & 0.141   &0.957 &-0.001 & 0.104  &0.104  &0.951& -0.003 & 0.074 & 0.074 & 0.944
   \\[1.5ex]
       $\hat{\theta}_{2sls}$ & 0.259  & 1.883  & 1.901  & 0.992 &0.209 & 2.399  &2.408 & 0.996 & 0.297  &1.787  &1.811 & 0.994
  \\[1.5ex]
  $\hat{\theta}_{ols}$ & 0.241 &  0.119  & 0.269  & 0.486 & 0.242 & 0.088 & 0.258&  0.204& 0.242 & 0.062 & 0.250  & 0.026
     \\[1.5ex]
 \hline
\end{tabular}
\end{sidewaystable}

\begin{sidewaystable}[!htbp] 
\caption{Bin $X$ with Bin $Z_1, Z_2$: Performance of $\hat{\beta}_1$ (Coef. of $Z_{1}$) \\
Different Degrees of Exclusion Violations}
 \label{table:Sim1_beta1_exc}
\begin{tabular}{c|cccc|cccc|cccc}
\hline
\hline
& \multicolumn{4}{c|}{$n=250$ } & \multicolumn{4}{c|}{$n=500$ }& \multicolumn{4}{c}{$n=1000$ }  \\
\hline
 Est& Bias & SD & RMSE & CP &Bias & SD & RMSE & CP & Bias & SD & RMSE & CP \\
\hline
&\multicolumn{12}{c}{$\b_{01}=\b_{02}=1$}  \\
\hline
$\hat{\theta}$ &   -0.000 &   0.126 & 0.126 & 0.948 & -0.002 & 0.091 & 0.091 & 0.950&  -0.001 & 0.064 & 0.064  &0.947
  \\[1.5ex]
 $\hat{\theta}^{*}$ &  -0.000 &   0.126 & 0.126 & 0.948 & -0.002 & 0.091 & 0.091 & 0.950&  -0.001 & 0.064 & 0.064  &0.947
  \\[1.5ex]
 $\hat{\theta}_{disc}$ &  -0.000 &   0.126 & 0.126 & 0.948 & -0.002 & 0.091 & 0.091 & 0.950&  -0.001 & 0.064 & 0.064  &0.947
  \\[1.5ex]
   $\hat{\theta}_{2sls}$ & - & - & - & - & - & - & - & -  &  - & - & - & - \\[1.5ex]
  $\hat{\theta}_{ols}$ &   0.001 &  0.122  & 0.122 &  0.950 &   -0.002&   0.088 &  0.088 &  0.954&   -0.002&   0.062 &  0.062 &  0.945
       \\[1.5ex]
 \hline 
&\multicolumn{12}{c}{$\b_{01}=\b_{02}=0.5$} \\
\hline
$\hat{\theta}$ &  -0.000  &  0.126 & 0.126 & 0.948&-0.002  & 0.091  & 0.091 &  0.950  &  -0.001 & 0.064 & 0.064 & 0.947
  \\[1.5ex]
 $\hat{\theta}^{*}$ &  -0.000  &  0.126 & 0.126 & 0.948&-0.002  & 0.091  & 0.091 &  0.950  &   -0.001 & 0.064 & 0.064 & 0.947
  \\[1.5ex]
 $\hat{\theta}_{disc}$ & -0.000  &  0.126 & 0.126 & 0.948&-0.002  & 0.091  & 0.091 &  0.950  &   -0.001 & 0.064 & 0.064 & 0.947
  \\[1.5ex]
   $\hat{\theta}_{2sls}$ & - & - & - & - & - & - & - & -  &  - & - & - & - \\[1.5ex]
  $\hat{\theta}_{ols}$ &   0.001 & 0.122 & 0.122 & 0.950 &  -0.002  & 0.088  & 0.088  & 0.954 & -0.002 & 0.062  &0.062 & 0.945
     \\[1.5ex]
 \hline
 &\multicolumn{12}{c}{$\b_{01}=\b_{02}=0$} \\
\hline
$\hat{\theta}$ &  -0.000 &   0.126 & 0.126 & 0.948 & -0.002 & 0.091 & 0.091&  0.950 &  -0.001 & 0.064 & 0.064&  0.947
  \\[1.5ex]
 $\hat{\theta}^{*}$ &  -0.000 &   0.126 & 0.126 & 0.948& -0.002 & 0.091 & 0.091&  0.950 &  -0.001 & 0.064 & 0.064&  0.947
  \\[1.5ex]
 $\hat{\theta}_{disc}$ &  -0.000 &   0.126 & 0.126 & 0.948& -0.002 & 0.091 & 0.091&  0.950 &  -0.001 & 0.064 & 0.064&  0.947
   \\[1.5ex]
    $\hat{\theta}_{2sls}$ & - & - & - & - & - & - & - & -  &  - & - & - & - \\[1.5ex]
  $\hat{\theta}_{ols}$ &   0.001 & 0.122 & 0.122  &0.950 & -0.002 & 0.088&  0.088&  0.954 & -0.002&  0.062 & 0.062  &0.945
     \\[1.5ex]
 \hline
\end{tabular}
\end{sidewaystable}

\begin{sidewaystable}[!htbp] 
\caption{Bin $X$ with Bin $Z_1, Z_2$: Performance of $\hat{\beta}_2$ (Coef. of $Z_{2}$) \\
Different Degrees of Exclusion Violations}
 \label{table:Sim1_beta2_exc}
\begin{tabular}{c|cccc|cccc|cccc}
\hline
\hline
 & \multicolumn{4}{c|}{$n=250$ } & \multicolumn{4}{c|}{$n=500$ }& \multicolumn{4}{c}{$n=1000$ }  \\
\hline
 Est& Bias & SD & RMSE & CP &Bias & SD & RMSE & CP & Bias & SD & RMSE & CP \\
\hline
&\multicolumn{12}{c}{$\b_{01}=\b_{02}=1$}  \\
\hline
$\hat{\theta}$ &  -0.002&   0.126&   0.126 &  0.951&  0.004&   0.092  & 0.092 &  0.946&  0.001 &  0.064 &  0.064  & 0.956
  \\[1.5ex]
 $\hat{\theta}^{*}$ &   -0.002&   0.126&   0.126 &  0.951&  0.004&   0.092  & 0.092 &  0.946&  0.001 &  0.064 &  0.064  & 0.956
  \\[1.5ex]
 $\hat{\theta}_{disc}$ &   -0.002&   0.126&   0.126 &  0.951&  0.004&   0.092  & 0.092 &  0.946&  0.001 &  0.064 &  0.064  & 0.956
  \\[1.5ex]
   $\hat{\theta}_{2sls}$ & - & - & - & - & - & - & - & -  &  - & - & - & - \\[1.5ex]
  $\hat{\theta}_{ols}$ &  -0.001&   0.121&   0.121 &  0.951&  0.004&   0.089 &  0.089&   0.946&   0.001&   0.062&   0.062  & 0.954
       \\[1.5ex]
 \hline 
&\multicolumn{12}{c}{$\b_{01}=\b_{02}=0.5$} \\
\hline
$\hat{\theta}$ &   -0.002 & 0.126  &0.126 & 0.951& 0.004&   0.092 &  0.092  & 0.946&  0.001  &0.064 & 0.064  &0.956
  \\[1.5ex]
 $\hat{\theta}^{*}$ &  -0.002 & 0.126  &0.126 & 0.951&0.004&   0.092 &  0.092  & 0.946&   0.001  &0.064 & 0.064  &0.956
  \\[1.5ex]
 $\hat{\theta}_{disc}$ &   -0.002 & 0.126  &0.126 & 0.951&0.004&   0.092 &  0.092  & 0.946&   0.001  &0.064 & 0.064  &0.956
  \\[1.5ex]
   $\hat{\theta}_{2sls}$ & - & - & - & - & - & - & - & -  &  - & - & - & - \\[1.5ex]
  $\hat{\theta}_{ols}$ & -0.001 & 0.121 & 0.121 & 0.951 & 0.004  & 0.089 &  0.089&   0.946&  0.001  &0.062 & 0.062 & 0.954
     \\[1.5ex]
 \hline
 &\multicolumn{12}{c}{$\b_{01}=\b_{02}=0$} \\
\hline
$\hat{\theta}$ &-0.002 & 0.126 & 0.126 & 0.951& 0.004 & 0.092 & 0.092 & 0.946& 0.001 & 0.064 & 0.064 & 0.956
  \\[1.5ex]
 $\hat{\theta}^{*}$ &-0.002 & 0.126 & 0.126 & 0.951& 0.004 & 0.092 & 0.092 & 0.946& 0.001 & 0.064 & 0.064 & 0.956
  \\[1.5ex]
 $\hat{\theta}_{disc}$ & -0.002 & 0.126 & 0.126 & 0.951& 0.004 & 0.092 & 0.092 & 0.946& 0.001 & 0.064 & 0.064 & 0.956
    \\[1.5ex]
     $\hat{\theta}_{2sls}$ & - & - & - & - & - & - & - & -  &  - & - & - & - \\[1.5ex]
  $\hat{\theta}_{ols}$ &  -0.001 & 0.121 & 0.121 & 0.951& 0.004  &0.089&  0.089 & 0.946 &  0.001&  0.062  &0.062  &0.954
     \\[1.5ex]
 \hline
\end{tabular}
\end{sidewaystable}

\begin{sidewaystable}[!htbp] 
\centering
\caption{Bin $X$ with Bin $Z_1, Z_2$: Performance of $\hat{\alpha}$ (Intercept) \\Different Degrees of Endogeneity}
 \label{table:Sim1_alpha}
\begin{tabular}{c|cccc|cccc|cccc}
\hline
\hline
& \multicolumn{4}{c|}{$n=250$ } & \multicolumn{4}{c|}{$n=500$ }& \multicolumn{4}{c}{$n=1000$ }  \\
\hline
 Est& Bias & SD & RMSE & CP &Bias & SD & RMSE & CP & Bias & SD & RMSE & CP \\
\hline
&\multicolumn{12}{c}{$\rho=0.5$}  \\
\hline
$\hat{\theta}$ &   0.001 & 0.141&  0.141&  0.957 & -0.001&  0.104 & 0.104&  0.951& -0.003 & 0.074&  0.074  &0.944
  \\[1.5ex]
 $\hat{\theta}^{*}$ &   0.001 & 0.141&  0.141&  0.957 & -0.001&  0.104 & 0.104&  0.951& -0.003 & 0.074&  0.074  &0.944
  \\[1.5ex]
 $\hat{\theta}_{disc}$ &   0.001 & 0.141&  0.141&  0.957 & -0.001&  0.104 & 0.104&  0.951& -0.003 & 0.074&  0.074  &0.944
  \\[1.5ex]
    $\hat{\theta}_{2sls}$ & 1.399 & 17.307 & 17.364 & 0.846 & 2.602 &  59.012 &  59.069 & 0.859 &  -0.254 & 33.300  &  33.301 & 0.854
  \\[1.5ex]
  $\hat{\theta}_{ols}$ &  0.241 & 0.119  &0.269  &0.486 &  0.242 & 0.088&  0.258 & 0.204  & 0.242 & 0.062 & 0.250  & 0.026
       \\[1.5ex]
 \hline 
&\multicolumn{12}{c}{$\rho=0$} \\
\hline
$\hat{\theta}$ &   0.003 & 0.140 &  0.140  & 0.958&  -0.000 & 0.104&  0.104  &0.950 &  -0.002&  0.073  &0.073  &0.940 
  \\[1.5ex]
 $\hat{\theta}^{*}$ &     0.003 & 0.140 &  0.140  & 0.958&  -0.000 & 0.104&  0.104  &0.950 &  -0.002&  0.073  &0.073  &0.940 
  \\[1.5ex]
 $\hat{\theta}_{disc}$ &   0.003 & 0.140 &  0.140  & 0.958&  -0.000 & 0.104&  0.104  &0.950 &  -0.002&  0.073  &0.073  &0.940 
  \\[1.5ex]
      $\hat{\theta}_{2sls}$ & 1.663 & 16.202  & 16.287 &  0.881 &1.121 & 19.280 &   19.312 & 0.872 &  0.829  &36.898  &36.907  & 0.863
  \\[1.5ex]
  $\hat{\theta}_{ols}$ &  -0.000  &  0.125  &0.125 & 0.952 & -0.000 & 0.090   &0.090 &  0.946&  -0.001 & 0.065 & 0.065&  0.942
     \\[1.5ex]
 \hline
 &\multicolumn{12}{c}{$\rho=-0.5$} \\
\hline
$\hat{\theta}$ & 0.005 &  0.141 &  0.141  & 0.957 &  0.001&   0.104 &  0.104 &  0.952 & -0.002  & 0.073 &  0.073 &  0.940 
  \\[1.5ex]
 $\hat{\theta}^{*}$ &  0.005 &  0.141 &  0.141  & 0.957 &  0.001&   0.104 &  0.104 &  0.952 & -0.002  & 0.073 &  0.073 &  0.940 
  \\[1.5ex]
 $\hat{\theta}_{disc}$ &   0.005 &  0.141 &  0.141  & 0.957 &  0.001&   0.104 &  0.104 &  0.952 & -0.002  & 0.073 &  0.073 &  0.940 
   \\[1.5ex]
       $\hat{\theta}_{2sls}$ & 0.593 & 29.434  &29.440&  0.887&  0.283 &  29.787  & 29.789 &  0.877&   0.707 & 37.594 & 37.601 & 0.872
  \\[1.5ex]
  $\hat{\theta}_{ols}$ &  -0.242 & 0.122  &0.271&  0.496&  -0.242&  0.087 & 0.257  &0.202& -0.243 & 0.062&  0.250&   0.026
     \\[1.5ex]
 \hline
\end{tabular}
\end{sidewaystable}

\begin{sidewaystable}[!htbp] 
\caption{Bin $X$ with Bin $Z_1, Z_2$: Performance of $\hat{\beta}_1$ (Coef. of $Z_{1}$) \\Different Degrees of Endogeneity}
 \label{table:Sim1_beta1}
\begin{tabular}{c|cccc|cccc|cccc}
\hline
\hline
& \multicolumn{4}{c|}{$n=250$ } & \multicolumn{4}{c|}{$n=500$ }& \multicolumn{4}{c}{$n=1000$ }  \\
\hline
 Est& Bias & SD & RMSE & CP &Bias & SD & RMSE & CP & Bias & SD & RMSE & CP \\
\hline
&\multicolumn{12}{c}{$\rho=0.5$}  \\
\hline
$\hat{\theta}$ &   -0.000 &   0.126 & 0.126 & 0.948 & -0.002 & 0.091 & 0.091 & 0.950&  -0.001 & 0.064 & 0.064  &0.947
  \\[1.5ex]
 $\hat{\theta}^{*}$ &  -0.000 &   0.126 & 0.126 & 0.948 & -0.002 & 0.091 & 0.091 & 0.950&  -0.001 & 0.064 & 0.064  &0.947
  \\[1.5ex]
 $\hat{\theta}_{disc}$ &  -0.000 &   0.126 & 0.126 & 0.948 & -0.002 & 0.091 & 0.091 & 0.950&  -0.001 & 0.064 & 0.064  &0.947
  \\[1.5ex]
   $\hat{\theta}_{2sls}$ & - & - & - & - & - & - & - & -  &  - & - & - & - \\[1.5ex]
  $\hat{\theta}_{ols}$ &   0.001 &  0.122  & 0.122 &  0.950 &   -0.002&   0.088 &  0.088 &  0.954&   -0.002&   0.062 &  0.062 &  0.945
       \\[1.5ex]
 \hline 
&\multicolumn{12}{c}{$\rho=0$} \\
\hline
$\hat{\theta}$ & -0.000 & 0.126 & 0.126  &0.946 &-0.002  &0.090  &0.090 & 0.948 &-0.001 & 0.063 & 0.063  &0.947
  \\[1.5ex]
 $\hat{\theta}^{*}$ & -0.000 & 0.126 & 0.126  &0.946 &-0.002  &0.090  &0.090 & 0.948 &-0.001 & 0.063 & 0.063  &0.947
  \\[1.5ex]
 $\hat{\theta}_{disc}$ &-0.000 & 0.126 & 0.126  &0.946 &-0.002  &0.090  &0.090 & 0.948 &-0.001 & 0.063 & 0.063  &0.947
  \\[1.5ex]
   $\hat{\theta}_{2sls}$ & - & - & - & - & - & - & - & -  &  - & - & - & - \\[1.5ex]
  $\hat{\theta}_{ols}$ &   0.000 & 0.126 & 0.126  &0.944 & -0.002 & 0.090 & 0.090  &0.949 & -0.001 & 0.063&  0.063&  0.946
     \\[1.5ex]
 \hline
 &\multicolumn{12}{c}{$\rho=-0.5$} \\
\hline
$\hat{\theta}$ &   0.000 & 0.126 & 0.126 & 0.948&  -0.002 & 0.090&   0.090&   0.950&  -0.001&  0.063  &0.063 & 0.947
  \\[1.5ex]
 $\hat{\theta}^{*}$ & 0.000 & 0.126 & 0.126 & 0.948&  -0.002 & 0.090&   0.090&   0.950&  -0.001&  0.063  &0.063 & 0.947
  \\[1.5ex]
 $\hat{\theta}_{disc}$ &  0.000 & 0.126 & 0.126 & 0.948&  -0.002 & 0.090&   0.090&   0.950&  -0.001&  0.063  &0.063 & 0.947
   \\[1.5ex]
    $\hat{\theta}_{2sls}$ & - & - & - & - & - & - & - & -  &  - & - & - & - \\[1.5ex]
  $\hat{\theta}_{ols}$ &  -0.000  &0.123&  0.123 & 0.948&  -0.002 & 0.087 & 0.087 & 0.948& -0.001 & 0.061 & 0.061 & 0.950
     \\[1.5ex]
 \hline
\end{tabular}
\end{sidewaystable}

\begin{sidewaystable}[!htbp] 
\caption{Bin $X$ with Bin $Z_1, Z_2$: Performance of $\hat{\beta}_2$ (Coef. of $Z_{2}$) \\Different Degrees of Endogeneity}
 \label{table:Sim1_beta2}
\begin{tabular}{c|cccc|cccc|cccc}
\hline
\hline
& \multicolumn{4}{c|}{$n=250$ } & \multicolumn{4}{c|}{$n=500$ }& \multicolumn{4}{c}{$n=1000$ }  \\
\hline
 Est& Bias & SD & RMSE & CP &Bias & SD & RMSE & CP & Bias & SD & RMSE & CP \\
\hline
&\multicolumn{12}{c}{$\rho=0.5$}  \\
\hline
$\hat{\theta}$ &  -0.002&   0.126&   0.126 &  0.951&  0.004&   0.092  & 0.092 &  0.946&  0.001 &  0.064 &  0.064  & 0.956
  \\[1.5ex]
 $\hat{\theta}^{*}$ &   -0.002&   0.126&   0.126 &  0.951&  0.004&   0.092  & 0.092 &  0.946&  0.001 &  0.064 &  0.064  & 0.956
  \\[1.5ex]
 $\hat{\theta}_{disc}$ &   -0.002&   0.126&   0.126 &  0.951&  0.004&   0.092  & 0.092 &  0.946&  0.001 &  0.064 &  0.064  & 0.956
  \\[1.5ex]
   $\hat{\theta}_{2sls}$ & - & - & - & - & - & - & - & -  &  - & - & - & - \\[1.5ex]
  $\hat{\theta}_{ols}$ &  -0.001&   0.121&   0.121 &  0.951&  0.004&   0.089 &  0.089&   0.946&   0.001&   0.062&   0.062  & 0.954
       \\[1.5ex]
 \hline 
&\multicolumn{12}{c}{$\rho=0$} \\
\hline
$\hat{\theta}$ &  -0.002 & 0.126 & 0.126 & 0.950 & 0.004  &0.092&  0.092&  0.944&
 0.001  &0.064  &0.064  &0.956
  \\[1.5ex]
 $\hat{\theta}^{*}$ &   -0.002 & 0.126 & 0.126 & 0.950 & 0.004  &0.092&  0.092&  0.944&
 0.001  &0.064  &0.064  &0.956
  \\[1.5ex]
 $\hat{\theta}_{disc}$ &   -0.002 & 0.126 & 0.126 & 0.950 & 0.004  &0.092&  0.092&  0.944&
 0.001  &0.064  &0.064  &0.956
  \\[1.5ex]
   $\hat{\theta}_{2sls}$ & - & - & - & - & - & - & - & -  &  - & - & - & - \\[1.5ex]
  $\hat{\theta}_{ols}$ &  -0.001 & 0.126 & 0.126 & 0.946 & 0.004 & 0.092  &0.092 & 0.946 & 0.001&  0.064 & 0.064 & 0.956
     \\[1.5ex]
 \hline
 &\multicolumn{12}{c}{$\rho=-0.5$} \\
\hline
$\hat{\theta}$ & -0.002 & 0.126 & 0.126 & 0.948& 0.004&  0.092  &0.092&  0.944 &0.001 & 0.064 & 0.064  &0.958
  \\[1.5ex]
 $\hat{\theta}^{*}$ & -0.002 & 0.126 & 0.126 & 0.948& 0.004&  0.092  &0.092&  0.944 &0.001 & 0.064 & 0.064  &0.958
  \\[1.5ex]
 $\hat{\theta}_{disc}$ &  -0.002 & 0.126 & 0.126 & 0.948& 0.004&  0.092  &0.092&  0.944 &0.001 & 0.064 & 0.064  &0.958
   \\[1.5ex]
    $\hat{\theta}_{2sls}$ & - & - & - & - & - & - & - & -  &  - & - & - & - \\[1.5ex]
  $\hat{\theta}_{ols}$ &  -0.001 & 0.124 & 0.124  &0.946 &0.003 & 0.089&  0.089 & 0.942 &0.001  &0.061 & 0.061 & 0.952
     \\[1.5ex]
 \hline
\end{tabular}
\end{sidewaystable}

\begin{sidewaystable}[!htbp] 
\caption{Bin $X$ with Cts $Z$: Performance of $\hat{\alpha}$ (Intercept) \\
Different Degrees of Exclusion Violations}
 \label{table:Sim2_alpha_exc}
\begin{tabular}{c|cccc|cccc|cccc}
\hline
\hline
& \multicolumn{4}{c|}{$n=250$ } & \multicolumn{4}{c|}{$n=500$ }& \multicolumn{4}{c}{$n=1000$ }  \\
\hline
 Est& Bias & SD & RMSE & CP &Bias & SD & RMSE & CP & Bias & SD & RMSE & CP \\
\hline
&\multicolumn{12}{c}{$\b_0=1$}  \\
\hline
$\hat{\theta}$ &   -0.021 & 0.175  &0.176 & 0.948 & -0.018 &  0.118 & 0.119  &  0.956 &   -0.012  &0.083 & 0.084 & 0.950 
  \\[1.5ex]
 $\hat{\theta}^{*}$ &  0.052 & 0.165 & 0.173 & 0.942 &  0.036 &  0.113 &  0.118  & 0.951&   0.029 & 0.080 &  0.085&  0.936
  \\[1.5ex]
 $\hat{\theta}_{disc}$ &  0.017 & 0.174 & 0.175 & 0.948 &  0.009 &  0.119 &  0.119  & 0.961&   0.007 & 0.086 & 0.086 & 0.950
  \\[1.5ex]
 $\hat{\theta}_{2sls}$ & -2.582 & 0.192 & 2.589 & 0.000 &-2.578 & 0.136  &2.581  &0.000 &-2.583 & 0.098 & 2.585 & 0.000  
   \\[1.5ex]
  $\hat{\theta}_{ols}$ &   0.240 &  0.118 & 0.267 & 0.447 &  0.243 &  0.083 &  0.257&   0.170 &   0.243  &0.058  &0.250 &  0.012
       \\[1.5ex]
 \hline 
&\multicolumn{12}{c}{$\b_0=0.5$} \\
\hline
$\hat{\theta}$ &  -0.021&  0.175 & 0.176 & 0.948&-0.018 & 0.118&  0.119&  0.956& -0.012&  0.083 &  0.084 & 0.950 
  \\[1.5ex]
 $\hat{\theta}^{*}$ &   0.077 & 0.160 &  0.177 & 0.935& 0.057  &0.111&  0.124 & 0.940 &0.046 & 0.079 & 0.092 & 0.916
   \\[1.5ex]
 $\hat{\theta}_{disc}$ &   0.017&  0.174&  0.175&  0.948&  0.009 & 0.119 & 0.119 & 0.961& 0.007 & 0.086&  0.086  &0.950
  \\[1.5ex]
   $\hat{\theta}_{2sls}$ &  -1.291  &0.132 & 1.298 & 0.000  & -1.288  &0.094 & 1.292 & 0.000  &-1.291 & 0.068  &1.293 & 0.000  
   \\[1.5ex]
  $\hat{\theta}_{ols}$ &   0.240  & 0.118&  0.267 & 0.447&  0.243  &0.083 & 0.257&  0.170 & 0.243 & 0.058 & 0.250  & 0.012
     \\[1.5ex]
 \hline
 &\multicolumn{12}{c}{$\b_0=0$} \\
\hline
$\hat{\theta}$ &  -0.021&  0.175 & 0.176 & 0.948& -0.018 & 0.118 & 0.119 &  0.956& -0.012 & 0.083 & 0.084 & 0.950 
  \\[1.5ex]
 $\hat{\theta}^{*}$ &   0.146 & 0.168  &0.222&  0.840 &  0.110&   0.119  &0.162 & 0.836 & 0.084 & 0.086 & 0.120 &  0.808
  \\[1.5ex]
 $\hat{\theta}_{disc}$ &  0.017 & 0.174&  0.175 & 0.948&  0.009&  0.119 & 0.119 & 0.961&  0.007&  0.086 & 0.086 & 0.950 
   \\[1.5ex]
      $\hat{\theta}_{2sls}$ &  0.000 & 0.103 & 0.103 & 0.948 & 0.002 & 0.072 & 0.072&  0.960 &   0.000  &  0.052 & 0.052 & 0.950 
   \\[1.5ex]
  $\hat{\theta}_{ols}$ &0.240 & 0.118 & 0.267 & 0.447 &   0.243  &0.083 & 0.257&  0.170 &   0.243 & 0.058 & 0.250  & 0.012
     \\[1.5ex]
 \hline
\end{tabular}
\end{sidewaystable}

\begin{sidewaystable}[!htbp] 
\caption{Bin $X$ with Cts $Z$: Performance of $\hat{\beta}$ (Coef. of $Z$)  \\
Different Degrees of Exclusion Violations}
 \label{table:Sim2_beta_exc}
\begin{tabular}{c|cccc|cccc|cccc}
\hline
\hline
& \multicolumn{4}{c|}{$n=250$ } & \multicolumn{4}{c|}{$n=500$ }& \multicolumn{4}{c}{$n=1000$ }  \\
\hline
 Est& Bias & SD & RMSE & CP &Bias & SD & RMSE & CP & Bias & SD & RMSE & CP \\
\hline
&\multicolumn{12}{c}{$\b_0=1$}  \\
\hline
$\hat{\theta}$ &  -0.006 & 0.070 &  0.070&   0.936 & -0.006 & 0.048 & 0.048 & 0.952&-0.004  &0.034 & 0.035 & 0.937
  \\[1.5ex]
 $\hat{\theta}^{*}$ &   -0.016 & 0.066 & 0.068 &  0.943 & -0.016 & 0.045 & 0.048 & 0.948& -0.012  &0.033  &0.035&  0.939
  \\[1.5ex]
 $\hat{\theta}_{disc}$ &  0.007 & 0.073 & 0.073  & 0.944 &  0.003 & 0.050&   0.050&   0.958&  0.003 & 0.037 & 0.037 & 0.943
  \\[1.5ex]
   $\hat{\theta}_{2sls}$ & - & - & - & - & - & - & - & -  &  - & - & - & - \\[1.5ex]
  $\hat{\theta}_{ols}$ &   0.093 & 0.051  & 0.106  & 0.534 &  0.094 & 0.036  &0.100   & 0.247&  0.094  &0.025 & 0.097 & 0.041
       \\[1.5ex]
 \hline 
&\multicolumn{12}{c}{$\b_0=0.5$} \\
\hline
$\hat{\theta}$ &  -0.006 & 0.070&   0.070&   0.936&-0.006 & 0.048 & 0.048 & 0.952& -0.004 & 0.034 & 0.035 & 0.937
 \\[1.5ex]
 $\hat{\theta}^{*}$ &  -0.002 & 0.062 & 0.062 & 0.963& -0.003 & 0.043&  0.043  &0.966& -0.002 & 0.032 & 0.032 & 0.966
  \\[1.5ex]
 $\hat{\theta}_{disc}$ &   0.007&  0.073 & 0.073  &0.944&  0.003&  0.050   &0.050 &  0.958& 0.003 & 0.037 & 0.037 & 0.943
  \\[1.5ex]
   $\hat{\theta}_{2sls}$ & - & - & - & - & - & - & - & -  &  - & - & - & - \\[1.5ex]
  $\hat{\theta}_{ols}$ &   0.093 & 0.051&  0.106 & 0.534&   0.094 & 0.036 & 0.100&  0.247&  0.094 & 0.025 & 0.097 & 0.041
     \\[1.5ex]
 \hline
 &\multicolumn{12}{c}{$\b_0=0$} \\
\hline
$\hat{\theta}$ &  -0.006 & 0.070 &  0.070&   0.936& -0.006 & 0.048 & 0.048 & 0.952& -0.004&   0.034 &  0.035 &  0.937
 \\[1.5ex]
 $\hat{\theta}^{*}$ &  0.042 & 0.059 & 0.073 & 0.938&  0.033 & 0.043  &0.054 & 0.918&  0.026 &  0.033 &  0.042 &  0.880 
  \\[1.5ex]
 $\hat{\theta}_{disc}$ & 0.007 & 0.073&  0.073 & 0.944 &   0.003 & 0.050 &  0.050 &  0.958 &  0.003&   0.037 &  0.037&   0.943
   \\[1.5ex]
    $\hat{\theta}_{2sls}$ & - & - & - & - & - & - & - & -  &  - & - & - & - \\[1.5ex]
  $\hat{\theta}_{ols}$ & 0.093 & 0.051 & 0.106 & 0.534 &  0.094 & 0.036&  0.100 &   0.247&  0.094 &  0.025 &  0.097 &  0.041
     \\[1.5ex]
 \hline
\end{tabular}
\end{sidewaystable}

\begin{sidewaystable}[!htbp] 
\caption{Bin $X$ with Cts $Z$: Performance of $\hat{\alpha}$ (Intercept) \\Different Degrees of Endogeneity}
 \label{table:Sim2_alpha}
\begin{tabular}{c|cccc|cccc|cccc}
\hline
\hline
& \multicolumn{4}{c|}{$n=250$ } & \multicolumn{4}{c|}{$n=500$ }& \multicolumn{4}{c}{$n=1000$ }  \\
\hline
 Est& Bias & SD & RMSE & CP &Bias & SD & RMSE & CP & Bias & SD & RMSE & CP \\
\hline
&\multicolumn{12}{c}{$\rho=0.5$}  \\
\hline
$\hat{\theta}$ &   -0.021 & 0.175  &0.176 & 0.948 & -0.018 &  0.118 & 0.119  &  0.956 &   -0.012  &0.083 & 0.084 & 0.950 
  \\[1.5ex]
 $\hat{\theta}^{*}$ &  0.052 & 0.165 & 0.173 & 0.942 &  0.036 &  0.113 &  0.118  & 0.951&   0.029 & 0.080 &  0.085&  0.936
  \\[1.5ex]
 $\hat{\theta}_{disc}$ &  0.017 & 0.174 & 0.175 & 0.948 &  0.009 &  0.119 &  0.119  & 0.961&   0.007 & 0.086 & 0.086 & 0.950
  \\[1.5ex]
 $\hat{\theta}_{2sls}$ & -2.582 & 0.192 & 2.589 & 0.000 &-2.578 & 0.136  &2.581  &0.000 &-2.583 & 0.098 & 2.585 & 0.000  
   \\[1.5ex]
  $\hat{\theta}_{ols}$ &   0.240 &  0.118 & 0.267 & 0.447 &  0.243 &  0.083 &  0.257&   0.170 &   0.243  &0.058  &0.250 &  0.012
       \\[1.5ex]
 \hline 
&\multicolumn{12}{c}{$\rho=0$} \\
\hline
$\hat{\theta}$ & -0.036 & 0.174 & 0.178  &0.936 & -0.027  &0.117 & 0.120 &  0.954 & -0.017 &  0.083  &0.085&  0.946
  \\[1.5ex]
 $\hat{\theta}^{*}$ &   0.048 & 0.165&  0.172 & 0.948 &  0.033 &  0.112 &  0.117 & 0.954 &  0.028  &0.080  &0.085 & 0.943
  \\[1.5ex]
 $\hat{\theta}_{disc}$ &  0.001 & 0.175 & 0.175 & 0.950 & -0.001 &  0.119 & 0.119 & 0.958 &  0.003  &0.086 & 0.086 & 0.952
  \\[1.5ex]
   $\hat{\theta}_{2sls}$ &  -2.581 & 0.187 & 2.588 & 0.000 & -2.577&  0.134&  2.581&  0.000 &  -2.581&  0.096 & 2.583 & 0.000 
   \\[1.5ex]
  $\hat{\theta}_{ols}$ &   0.004 & 0.121&  0.121  & 0.940 &  0.002 &  0.086 & 0.086 & 0.948& 0.001 & 0.059 & 0.059 & 0.952
     \\[1.5ex]
 \hline
 &\multicolumn{12}{c}{$\rho=-0.5$} \\
\hline
$\hat{\theta}$ &  -0.053&  0.173 & 0.181 & 0.918 & -0.037 & 0.116 & 0.122&  0.946 &
  -0.023&  0.083&  0.086&  0.940 
  \\[1.5ex]
 $\hat{\theta}^{*}$ &   0.043 & 0.164 &  0.169 & 0.955 & 0.030 &  0.112 & 0.116  & 0.962&  0.027 & 0.080 &  0.084 & 0.944
  \\[1.5ex]
 $\hat{\theta}_{disc}$ &   -0.017&  0.174 & 0.174 & 0.951 & -0.009 & 0.119 & 0.119&  0.958& -0.002 & 0.086 & 0.086 & 0.956
   \\[1.5ex]
      $\hat{\theta}_{2sls}$ &  -2.582 &  0.180 &   2.588&  0.000 & -2.577 & 0.129 & 2.580&  0.000  & -2.581&  0.093  &2.582&  0.000  
   \\[1.5ex]
  $\hat{\theta}_{ols}$ &  -0.239&  0.118 & 0.266 & 0.452 & -0.238 & 0.083 & 0.252 & 0.174& -0.240&   0.058 & 0.247 & 0.017
     \\[1.5ex]
 \hline
\end{tabular}
\end{sidewaystable}

\begin{sidewaystable}[!htbp] 
\caption{Bin $X$ with Cts $Z$: Performance of $\hat{\beta}$ (Coef. of $Z$) \\Different Degrees of Endogeneity}
 \label{table:Sim2_beta}
\begin{tabular}{c|cccc|cccc|cccc}
\hline
\hline
& \multicolumn{4}{c|}{$n=250$ } & \multicolumn{4}{c|}{$n=500$ }& \multicolumn{4}{c}{$n=1000$ }  \\
\hline
 Est& Bias & SD & RMSE & CP &Bias & SD & RMSE & CP & Bias & SD & RMSE & CP \\
\hline
&\multicolumn{12}{c}{$\rho=0.5$}  \\
\hline
$\hat{\theta}$ &  -0.006 & 0.070 &  0.070&   0.936 & -0.006 & 0.048 & 0.048 & 0.952&-0.004  &0.034 & 0.035 & 0.937
  \\[1.5ex]
 $\hat{\theta}^{*}$ &   -0.016 & 0.066 & 0.068 &  0.943 & -0.016 & 0.045 & 0.048 & 0.948& -0.012  &0.033  &0.035&  0.939
  \\[1.5ex]
 $\hat{\theta}_{disc}$ &  0.007 & 0.073 & 0.073  & 0.944 &  0.003 & 0.050&   0.050&   0.958&  0.003 & 0.037 & 0.037 & 0.943
  \\[1.5ex]
   $\hat{\theta}_{2sls}$ & - & - & - & - & - & - & - & -  &  - & - & - & - \\[1.5ex]
  $\hat{\theta}_{ols}$ &   0.093 & 0.051  & 0.106  & 0.534 &  0.094 & 0.036  &0.100   & 0.247&  0.094  &0.025 & 0.097 & 0.041
       \\[1.5ex]
 \hline 
&\multicolumn{12}{c}{$\rho=0$} \\
\hline
$\hat{\theta}$ & -0.012 & 0.070&   0.071 & 0.929 & -0.010&   0.048 & 0.049 & 0.946& -0.006 & 0.034&  0.035 & 0.939 
  \\[1.5ex]
 $\hat{\theta}^{*}$ &  -0.019 &  0.065 & 0.068 &  0.940 &  -0.018 & 0.045 &  0.048 & 0.947 & -0.013  &0.033&  0.035 & 0.934
  \\[1.5ex]
 $\hat{\theta}_{disc}$ &  0.000  &  0.073 & 0.073 & 0.946 & -0.001 &  0.050 & 0.050 &  0.956&  0.001 & 0.037  &0.037  &0.943
  \\[1.5ex]
   $\hat{\theta}_{2sls}$ & - & - & - & - & - & - & - & -  &  - & - & - & - \\[1.5ex]
  $\hat{\theta}_{ols}$ &   0.001 & 0.051&  0.051 & 0.940 &  0.000  &  0.036  &0.036 & 0.942 & 0.000&    0.026 & 0.026&  0.944
     \\[1.5ex]
 \hline
 &\multicolumn{12}{c}{$\rho=-0.5$} \\
\hline
$\hat{\theta}$ &  -0.018  & 0.070 &   0.072  & 0.920  & -0.013 & 0.048 & 0.049 & 0.942 & -0.008 & 0.034&  0.035 & 0.934
  \\[1.5ex]
 $\hat{\theta}^{*}$ &  -0.021 &  0.065  & 0.068  & 0.940  & -0.019 & 0.045 & 0.049 & 0.945 & -0.014  &0.032 & 0.035 & 0.932
  \\[1.5ex]
 $\hat{\theta}_{disc}$ &   -0.007  & 0.073 &  0.073  & 0.946 & -0.004 & 0.050 &  0.050  & 0.958 & -0.001 & 0.037& 0.037&  0.944
   \\[1.5ex]
    $\hat{\theta}_{2sls}$ & - & - & - & - & - & - & - & -  &  - & - & - & - \\[1.5ex]
  $\hat{\theta}_{ols}$ &  -0.093  & 0.051  & 0.106  & 0.540  & -0.093 & 0.035 & 0.100 &  0.246 & -0.093 & 0.025 & 0.096 & 0.039
     \\[1.5ex]
 \hline
\end{tabular}
\end{sidewaystable}

\begin{sidewaystable}[!htbp] 
\caption{Cts $X$ with Cts $Z$: Performance of $\hat{\alpha}$ (Intercept) \\
Different Degrees of Exclusion Violations}
 \label{table:Sim3_alpha_exc}
\begin{tabular}{c|cccc|cccc|cccc}
\hline
\hline
& \multicolumn{4}{c|}{$n=250$ } & \multicolumn{4}{c|}{$n=500$ }& \multicolumn{4}{c}{$n=1000$ }  \\
\hline
 Est& Bias & SD & RMSE & CP &Bias & SD & RMSE & CP & Bias & SD & RMSE & CP \\
\hline
&\multicolumn{12}{c}{$\b_0=1$}  \\
\hline
$\hat{\theta}$ &  -0.000 & 0.062 & 0.062 & 0.953 & 0.001 & 0.045 & 0.045 & 0.952 & 0.000 & 0.031 & 0.031 & 0.952
  \\[1.5ex]
 $\hat{\theta}^{*}$ & 0.001 & 0.066 & 0.066 & 0.951 & 0.001 & 0.047 & 0.047 & 0.950 & 0.000 & 0.032 & 0.032 & 0.951
  \\[1.5ex]
 $\hat{\theta}_{disc}$ & 0.000 & 0.063 & 0.063 & 0.951 & 0.001 & 0.046 & 0.046 & 0.952 & 0.000 & 0.032 & 0.032 & 0.951
  \\[1.5ex]
 $\hat{\theta}_{2sls}$ & 1.872 & 61.32 & 61.34 & 1.000 & 0.710 & 101.8 & 101.9 & 1.000 & -0.235 & 142.3 & 142.3 & 1.000
   \\[1.5ex]
  $\hat{\theta}_{ols}$ & -0.002 & 0.057 & 0.057 & 0.951 &  -0.000 & 0.042 & 0.042 & 0.948 & -0.000 & 0.029 & 0.029 & 0.949
       \\[1.5ex]
 \hline 
&\multicolumn{12}{c}{$\b_0=0.5$} \\
\hline
$\hat{\theta}$ & -0.000 & 0.062 & 0.062 & 0.953 & 0.001 & 0.045 & 0.045 & 0.952 & 0.000 & 0.031 & 0.031 & 0.952
  \\[1.5ex]
 $\hat{\theta}^{*}$ & 0.001 & 0.066 & 0.066 & 0.951 & 0.001 & 0.047 & 0.047 & 0.950 & 0.000 & 0.032 & 0.032 & 0.951
   \\[1.5ex]
 $\hat{\theta}_{disc}$ & 0.000 & 0.063 & 0.063 & 0.951 & 0.001 & 0.046 & 0.046 & 0.952 & 0.000 & 0.032 & 0.032 & 0.951
  \\[1.5ex]
   $\hat{\theta}_{2sls}$ & 0.939 & 30.69 & 30.69 & 1.000 & 0.316 & 49.78 & 49.77 & 1.000 & -0.076 & 68.796 & 68.779 & 1.000  
   \\[1.5ex]
  $\hat{\theta}_{ols}$ & -0.002 & 0.057 & 0.057 & 0.951 & -0.000 & 0.042 & 0.042 & 0.948 & -0.000 & 0.029 & 0.029 & 0.949
     \\[1.5ex]
 \hline
 &\multicolumn{12}{c}{$\b_0=0$} \\
\hline
$\hat{\theta}$ & -0.000 & 0.062 & 0.062 & 0.953 & 0.001 & 0.045 & 0.045 & 0.952 & 0.000 & 0.031 & 0.031 & 0.952
  \\[1.5ex]
 $\hat{\theta}^{*}$ & 0.001 & 0.066 & 0.066 & 0.951 & 0.001 & 0.047 & 0.047 & 0.950 & 0.000 & 0.032 & 0.032 & 0.951
  \\[1.5ex]
 $\hat{\theta}_{disc}$ & 0.000 & 0.063 & 0.063 & 0.951 & 0.001 & 0.046 & 0.046 & 0.952 & 0.000 & 0.032 & 0.032 & 0.951
   \\[1.5ex]
 $\hat{\theta}_{2sls}$ & 0.006 & 1.389 & 1.389 & 0.988 & -0.078 & 3.467 & 3.467 & 0.989 & 0.083 & 5.043 & 5.042 & 0.992
   \\[1.5ex]
  $\hat{\theta}_{ols}$ & -0.002 & 0.057 & 0.057 & 0.951 &   -0.000 & 0.042 & 0.042 & 0.948 & -0.000 & 0.029 & 0.029 & 0.949
     \\[1.5ex]
 \hline
\end{tabular}
\end{sidewaystable}

\begin{sidewaystable}[!htbp] 
\caption{Cts $X$ with Cts $Z$: Performance of $\hat{\alpha}$ (Intercept) \\Different Degrees of Endogeneity}
 \label{table:Sim3_alpha}
\begin{tabular}{c|cccc|cccc|cccc}
\hline
\hline
& \multicolumn{4}{c|}{$n=250$ } & \multicolumn{4}{c|}{$n=500$ }& \multicolumn{4}{c}{$n=1000$ }  \\
\hline
 Est& Bias & SD & RMSE & CP &Bias & SD & RMSE & CP & Bias & SD & RMSE & CP \\
\hline
&\multicolumn{12}{c}{$\rho=0.5$}  \\
\hline
$\hat{\theta}$ &  -0.000 & 0.062 & 0.062 & 0.953 & 0.001 & 0.045 & 0.045 & 0.952 & 0.000 & 0.031 & 0.031 & 0.952
  \\[1.5ex]
 $\hat{\theta}^{*}$ & 0.001 & 0.066 & 0.066 & 0.951 & 0.001 & 0.047 & 0.047 & 0.950 & 0.000 & 0.032 & 0.032 & 0.951
  \\[1.5ex]
 $\hat{\theta}_{disc}$ & 0.000 & 0.063 & 0.063 & 0.951 & 0.001 & 0.046 & 0.046 & 0.952 & 0.000 & 0.032 & 0.032 & 0.951
  \\[1.5ex]
 $\hat{\theta}_{2sls}$ & 1.872 & 61.32 & 61.34 & 1.000 & 0.710 & 101.8 & 101.9 & 1.000 & -0.235 & 142.3 & 142.3 & 1.000
   \\[1.5ex]
  $\hat{\theta}_{ols}$ & -0.002 & 0.057 & 0.057 & 0.951 &  -0.000 & 0.042 & 0.042 & 0.948 & -0.000 & 0.029 & 0.029 & 0.949
       \\[1.5ex]
 \hline 
&\multicolumn{12}{c}{$\rho=0$} \\
\hline
$\hat{\theta}$ & -0.001 & 0.063 & 0.063 & 0.946 & 0.000 & 0.045 & 0.045 & 0.949 & -0.000 & 0.031 & 0.031 & 0.951
  \\[1.5ex]
 $\hat{\theta}^{*}$ & -0.001 & 0.063 & 0.063 & 0.946 & 0.000 & 0.045 & 0.045 & 0.950 & -0.000 & 0.031 & 0.031 & 0.950
  \\[1.5ex]
 $\hat{\theta}_{disc}$ & -0.001 & 0.063 & 0.063 & 0.947 & 0.000 & 0.045 & 0.045 & 0.950 & -0.000 & 0.031 & 0.031 & 0.952
  \\[1.5ex]
   $\hat{\theta}_{2sls}$ & -1.150 & 74.98 & 74.97 & 1.000 & -2.277 & 60.00 & 60.03 & 1.000 & -0.699 & 29.92 & 29.92 & 1.000
   \\[1.5ex]
  $\hat{\theta}_{ols}$ & -0.001 & 0.063 & 0.063 & 0.946 & 0.000 & 0.045 & 0.045 & 0.949 & -0.000 & 0.031 & 0.031 & 0.952
     \\[1.5ex]
 \hline
 &\multicolumn{12}{c}{$\rho=-0.5$} \\
\hline
$\hat{\theta}$ & -0.002 & 0.067 & 0.067 & 0.947 & -0.001 & 0.046 & 0.046 & 0.950 & -0.000 & 0.032 & 0.032 & 0.953
  \\[1.5ex]
 $\hat{\theta}^{*}$ & -0.002 & 0.065 & 0.065 & 0.949 & -0.001 & 0.045 & 0.045 & 0.949 & -0.000 & 0.032 & 0.032 & 0.953
  \\[1.5ex]
 $\hat{\theta}_{disc}$ & -0.001 & 0.064 & 0.064 & 0.952 &-0.000 & 0.045 & 0.045 & 0.949 & -0.000 & 0.032 & 0.032 & 0.952
   \\[1.5ex]
      $\hat{\theta}_{2sls}$ & -0.605 & 53.45 & 53.44 & 1.000 & 2.908 & 192.0 & 192.0 & 1.000 & -0.816 & 121.0 & 121.0 & 1.000
   \\[1.5ex]
  $\hat{\theta}_{ols}$ & 0.001 & 0.058 & 0.058 & 0.945 & 0.001 & 0.041 & 0.041 & 0.950 & 0.000 & 0.029 & 0.029 & 0.946
     \\[1.5ex]
 \hline
\end{tabular}
\end{sidewaystable}

\begin{sidewaystable}[!htbp] 
\caption{Cts $X$ with Cts $Z$: Performance of $\hat{\beta}$ (Coeff. on $Z$) \\
Different Degrees of Exclusion Violations}
 \label{table:Sim3_beta_exc}
\begin{tabular}{c|cccc|cccc|cccc}
\hline
\hline
& \multicolumn{4}{c|}{$n=250$ } & \multicolumn{4}{c|}{$n=500$ }& \multicolumn{4}{c}{$n=1000$ }  \\
\hline
 Est& Bias & SD & RMSE & CP &Bias & SD & RMSE & CP & Bias & SD & RMSE & CP \\
\hline
&\multicolumn{12}{c}{$\b_0=1$}  \\
\hline
$\hat{\theta}$ & -0.001 & 0.034 & 0.034 & 0.950 & 0.001 & 0.024 & 0.024 & 0.947 & 0.000 & 0.017 & 0.017 & 0.956\\[1.5ex] 
 $\hat{\theta}^{*}$ & -0.001 & 0.036 & 0.036 & 0.952 &   0.001 & 0.025 & 0.025 & 0.950 & 0.000 & 0.018 & 0.018 & 0.955  \\ [1.5ex]
$\hat{\theta}_{disc}$  & -0.001 & 0.035 & 0.035 & 0.950 &   0.001 & 0.025 & 0.025 & 0.949 & 0.000 & 0.017 & 0.017 & 0.957 \\ [1.5ex]
 $\hat{\theta}_{2sls}$ & - & - & - & - & - & - & - & -  &  - & - & - & - \\[1.5ex]
$\hat{\theta}_{ols}$ & -0.001 & 0.032 & 0.032 & 0.946  &  0.001 & 0.022 & 0.022 & 0.953 & 0.000 & 0.016 & 0.016 & 0.950  \\[1.5ex]
 \hline 
&\multicolumn{12}{c}{$\b_0=0.5$} \\
\hline
$\hat{\theta}$ & -0.001 & 0.034 & 0.034 & 0.950 & 0.001 & 0.024 & 0.024 & 0.947 & 0.000 & 0.017 & 0.017 & 0.956
  \\[1.5ex]
 $\hat{\theta}^{*}$ & -0.001 & 0.036 & 0.036 & 0.952 & 0.001 & 0.025 & 0.025 & 0.950 & 0.000 & 0.018 & 0.018 & 0.955
   \\[1.5ex]
 $\hat{\theta}_{disc}$ & -0.001 & 0.035 & 0.035 & 0.950 & 0.001 & 0.025 & 0.025 & 0.949 & 0.000 & 0.017 & 0.017 & 0.957
  \\[1.5ex]
 $\hat{\theta}_{2sls}$ & - & - & - & - & - & - & - & -  &  - & - & - & - \\[1.5ex]
  $\hat{\theta}_{ols}$ & -0.001 & 0.032 & 0.032 & 0.946 & 0.001 & 0.022 & 0.022 & 0.953 & 0.000 & 0.016 & 0.016 & 0.950
     \\[1.5ex]
 \hline
 &\multicolumn{12}{c}{$\b_0=0$} \\
\hline
$\hat{\theta}$ & -0.001 & 0.034 & 0.034 & 0.950 & 0.001 & 0.024 & 0.024 & 0.947 & 0.000 & 0.017 & 0.017 & 0.956 
  \\[1.5ex]
 $\hat{\theta}^{*}$ & -0.001 & 0.036 & 0.036 & 0.952 & 0.001 & 0.025 & 0.025 & 0.950 & 0.000 & 0.018 & 0.018 & 0.955
  \\[1.5ex]
 $\hat{\theta}_{disc}$ & -0.001 & 0.035 & 0.035 & 0.950 &  0.001 & 0.025 & 0.025 & 0.949 & 0.000 & 0.017 & 0.017 & 0.957
   \\[1.5ex]
 $\hat{\theta}_{2sls}$ & - & - & - & - & - & - & - & -  &  - & - & - & - \\[1.5ex]
  $\hat{\theta}_{ols}$ & -0.001 & 0.032 & 0.032 & 0.946 & 0.001 & 0.022 & 0.022 & 0.953 & 0.000 & 0.016 & 0.016 & 0.950
     \\[1.5ex]
 \hline
\end{tabular}
\end{sidewaystable}

\begin{sidewaystable}[!htbp] 
\caption{Cts $X$ with Cts $Z$: Performance of $\hat{\beta}$ (Coeff. of $Z$) \\Different Degrees of Endogeneity}
 \label{table:Sim3_beta}
\begin{tabular}{c|cccc|cccc|cccc}
\hline
\hline
& \multicolumn{4}{c|}{$n=250$ } & \multicolumn{4}{c|}{$n=500$ }& \multicolumn{4}{c}{$n=1000$ }  \\
\hline
 Est& Bias & SD & RMSE & CP &Bias & SD & RMSE & CP & Bias & SD & RMSE & CP \\
\hline
&\multicolumn{12}{c}{$\rho=0.5$}  \\
\hline
$\hat{\theta}$ & -0.001 & 0.034 & 0.034 & 0.950 & 0.001 & 0.024 & 0.024 & 0.947 & 0.000 & 0.017 & 0.017 & 0.956\\[1.5ex] 
 $\hat{\theta}^{*}$ & -0.001 & 0.036 & 0.036 & 0.952 &   0.001 & 0.025 & 0.025 & 0.950 & 0.000 & 0.018 & 0.018 & 0.955  \\ [1.5ex]
$\hat{\theta}_{disc}$  & -0.001 & 0.035 & 0.035 & 0.950 &   0.001 & 0.025 & 0.025 & 0.949 & 0.000 & 0.017 & 0.017 & 0.957 \\ [1.5ex]
 $\hat{\theta}_{2sls}$ & - & - & - & - & - & - & - & -  &  - & - & - & - \\[1.5ex]
$\hat{\theta}_{ols}$ & -0.001 & 0.032 & 0.032 & 0.946  &  0.001 & 0.022 & 0.022 & 0.953 & 0.000 & 0.016 & 0.016 & 0.950  \\[1.5ex]
 \hline 
&\multicolumn{12}{c}{$\rho=0$} \\
\hline
$\hat{\theta}$ & -0.001 & 0.036 & 0.036 & 0.945 & 0.000 & 0.025 & 0.025 & 0.944 & 0.000 & 0.017 & 0.017 & 0.953
  \\[1.5ex]
 $\hat{\theta}^{*}$ & -0.001 & 0.036 & 0.036 & 0.948 & 0.000 & 0.025 & 0.025 & 0.942 & 0.000 & 0.017 & 0.017 & 0.958
  \\[1.5ex]
 $\hat{\theta}_{disc}$ & -0.001 & 0.036 & 0.036 & 0.946 &0.000 & 0.025 & 0.025 & 0.942 & 0.000 & 0.017 & 0.017 & 0.955
  \\[1.5ex]
 $\hat{\theta}_{2sls}$ & - & - & - & - & - & - & - & -  &  - & - & - & - \\[1.5ex]
  $\hat{\theta}_{ols}$ & -0.001 & 0.036 & 0.036 & 0.948 &0.000 & 0.025 & 0.025 & 0.944 & 0.000 & 0.017 & 0.017 & 0.955
     \\[1.5ex]
 \hline
 &\multicolumn{12}{c}{$\rho=-0.5$} \\
\hline
$\hat{\theta}$ &  0.000 & 0.037 & 0.037 & 0.946 & -0.001 & 0.026 & 0.026 & 0.948 & 0.000 & 0.018 & 0.018 & 0.957
  \\[1.5ex]
 $\hat{\theta}^{*}$ & 0.000 & 0.036 & 0.036 & 0.946 &-0.001 & 0.025 & 0.025 & 0.946 &0.000 & 0.017 & 0.017 & 0.956
  \\[1.5ex]
 $\hat{\theta}_{disc}$ & 0.000 & 0.035 & 0.035 & 0.947 &-0.001 & 0.025 & 0.025 & 0.949 &0.000 & 0.018 & 0.017 & 0.955
   \\[1.5ex]
 $\hat{\theta}_{2sls}$ & - & - & - & - & - & - & - & -  &  - & - & - & - \\[1.5ex]
  $\hat{\theta}_{ols}$ & 0.000 & 0.032 & 0.032 & 0.947 & -0.000 & 0.023 & 0.023 & 0.941 & 0.000 & 0.016 & 0.016 & 0.955
     \\[1.5ex]
 \hline
\end{tabular}
\end{sidewaystable}

\end{document}